\documentclass[twoside,superscriptaddress,nofootinbib]{article}

\usepackage{qic,epsfig}
\usepackage[utf8]{inputenc}
\usepackage{amsmath}
\usepackage{amssymb}
\usepackage{bbold}
\usepackage{graphicx}
\usepackage{braket}
\usepackage{leftidx}
\usepackage{tikz, circuitikz}
\usepackage{ stmaryrd }
\usepackage{circuitikz}
\usepackage{textcomp}
\usepackage{accents}
\usepackage{comment}
\usepackage{breakurl,amsmath,amsthm,amssymb,xspace,xypic,enumerate,color,tikzfig}
\usepackage{hyperref}
\usepackage{caption, subcaption}
\usepackage{float}
\usepackage{relsize}
\usepackage{sidecap}

\usetikzlibrary{shapes.multipart}

\tikzstyle{filled}=[fill=white]

\tikzstyle{bwSpider}=[
       rectangle split,
       rectangle split parts=2,
       rectangle split part fill={black,white},
 minimum size=3.6 mm, inner sep=-2mm, draw=black,scale=0.5,rounded corners=0.8 mm
       ]
 \tikzstyle{wbSpider}=[
       rectangle split,
       rectangle split parts=2,
       rectangle split part fill={white,black},
 minimum size=3.6 mm, inner sep=-2mm, draw=black,scale=0.5,rounded corners=0.8 mm
       ]
\tikzstyle{qWire}=[line width = 1pt, black]
\tikzstyle{cWire}=[color=gray,thin]
\tikzstyle{env}=[copoint,regular polygon rotate=0,minimum width=0.2cm, fill=black]

\tikzstyle{probs}=[shape=semicircle,fill=white,draw=black,shape border rotate=180,minimum width=1.2cm]

\tikzstyle{disc}=[circuit ee IEC,thick,ground,rotate=90,scale=1.5]

\tikzstyle{every picture}=[baseline=-0.25em,scale=0.5]
\tikzstyle{dotpic}=[] 
\tikzstyle{diredges}=[every to/.style={diredge}]
\tikzstyle{math matrix}=[matrix of math nodes,left delimiter=(,right delimiter=),inner sep=2pt,column sep=1em,row sep=0.5em,nodes={inner sep=0pt},text height=1.5ex, text depth=0.25ex]

\tikzstyle{inline text}=[text height=1.5ex, text depth=0.25ex,yshift=0.5mm]
\tikzstyle{label}=[font=\footnotesize,text height=1.5ex, text depth=0.25ex,yshift=0.5mm]
\tikzstyle{left label}=[label,anchor=east,xshift=1.5mm]
\tikzstyle{right label}=[label,anchor=west,xshift=-0.5mm]

\tikzstyle{braceedge}=[decorate,decoration={brace,amplitude=2mm,raise=-1mm}]
\tikzstyle{small braceedge}=[decorate,decoration={brace,amplitude=1mm,raise=-1mm}]

\tikzstyle{doubled}=[line width=1.6pt] 
\tikzstyle{boldedge}=[doubled,shorten <=-0.17mm,shorten >=-0.17mm]
\tikzstyle{boldedgegray}=[doubled,gray,shorten <=-0.17mm,shorten >=-0.17mm]
\tikzstyle{singleedgegray}=[gray]

\tikzstyle{semidoubled}=[line width=1.4pt] 
\tikzstyle{semiboldedgegray}=[semidoubled,gray,shorten <=-0.17mm,shorten >=-0.17mm]

\tikzstyle{boxedge}=[semiboldedgegray]

\tikzstyle{boldedgedashed}=[very thick,dashed,shorten <=-0.17mm,shorten >=-0.17mm]
\tikzstyle{vboldedgedashed}=[doubled,dashed,shorten <=-0.17mm,shorten >=-0.17mm]
\tikzstyle{left hook arrow}=[left hook-latex]
\tikzstyle{right hook arrow}=[right hook-latex]
\tikzstyle{sembracket}=[line width=0.5pt,shorten <=-0.07mm,shorten >=-0.07mm]

\tikzstyle{causal edge}=[->,thick,gray]
\tikzstyle{causal nondir}=[thick,gray]
\tikzstyle{timeline}=[thick,gray, dashed]

\tikzstyle{cedge}=[<->,thick,gray!70!white]

\tikzstyle{empty diagram}=[draw=gray!40!white,dashed,shape=rectangle,minimum width=1cm,minimum height=1cm]
\tikzstyle{empty diagram small}=[draw=gray!50!white,dashed,shape=rectangle,minimum width=0.6cm,minimum height=0.5cm]

\tikzstyle{dot}=[inner sep=0mm,minimum width=2mm,minimum height=2mm,draw,shape=circle]

\tikzstyle{leak}=[white dot, shape=regular polygon, minimum size=3.3 mm, regular polygon sides=3, outer sep=-0.2mm, regular polygon rotate=270]
\tikzstyle{proj}=[regular polygon,regular polygon sides=4,draw,scale=0.75,inner sep=-0.5pt,minimum width=6mm,fill=white]
\tikzstyle{projOut}=[regular polygon,regular polygon sides=3,draw,scale=0.75,inner sep=-0.5pt,minimum width=7.5mm,fill=white,regular polygon rotate=180]
\tikzstyle{projIn}=[regular polygon,regular polygon sides=3,draw,scale=0.75,inner sep=-0.5pt,minimum width=7.5mm,fill=white]
\tikzstyle{Vleak}=[white dot, shape=regular polygon, minimum size=3.3 mm, regular polygon sides=3, outer sep=-0.2mm, regular polygon rotate=90]
\tikzstyle{dleak}=[white dot, line width=1.6pt, shape=regular polygon, minimum size=3.3 mm, regular polygon sides=3, outer sep=-0.2mm, regular polygon rotate=270]

\tikzstyle{Wsquare}=[white dot, shape=regular polygon, rounded corners=0.8 mm, minimum size=3.3 mm, regular polygon sides=3, outer sep=-0.2mm]
\tikzstyle{Wsquareadj}=[white dot, shape=regular polygon, rounded corners=0.8 mm, minimum size=3.3 mm, regular polygon sides=3, outer sep=-0.2mm, regular polygon rotate=180]
\tikzstyle{ddot}=[inner sep=0mm, doubled, minimum width=2.5mm,minimum height=2.5mm,draw,shape=circle]

\tikzstyle{black dot}=[dot,fill=black]
\tikzstyle{white dot}=[dot,fill=white,,text depth=-0.2mm]
\tikzstyle{white Wsquare}=[Wsquare,fill=gray,,text depth=-0.2mm]
\tikzstyle{white Wsquareadj}=[Wsquareadj,fill=white,,text depth=-0.2mm]
\tikzstyle{green dot}=[white dot] 
\tikzstyle{gray dot}=[dot,fill=gray!40!white,,text depth=-0.2mm]
\tikzstyle{red dot}=[gray dot] 

\tikzstyle{black ddot}=[ddot,fill=black]
\tikzstyle{white ddot}=[ddot,fill=white]
\tikzstyle{gray ddot}=[ddot,fill=gray!40!white]

\tikzstyle{gray edge}=[gray!60!white]

\tikzstyle{small dot}=[inner sep=0.5mm,minimum width=0pt,minimum height=0pt,draw,shape=circle]

\tikzstyle{small black dot}=[small dot,fill=black]
\tikzstyle{small white dot}=[small dot,fill=white]
\tikzstyle{small gray dot}=[small dot,fill=gray!40!white]

\tikzstyle{causal dot}=[inner sep=0.4mm,minimum width=0pt,minimum height=0pt,draw=white,shape=circle,fill=gray!40!white]

\tikzstyle{phase dimensions}=[minimum size=5mm,font=\footnotesize,rectangle,rounded corners=2.5mm,inner sep=0.2mm,outer sep=-2mm]
\tikzstyle{dphase dimensions}=[minimum size=5mm,font=\footnotesize,rectangle,rounded corners=2.5mm,inner sep=0.2mm,outer sep=-2mm]

\tikzstyle{white phase dot}=[dot,fill=white,phase dimensions]
\tikzstyle{white phase ddot}=[ddot,fill=white,dphase dimensions]

\tikzstyle{white rect ddot}=[draw=black,fill=white,doubled,minimum size=5mm,font=\footnotesize,rectangle,rounded corners=2.5mm,inner sep=0.2mm]
\tikzstyle{gray rect ddot}=[draw=black,fill=gray!40!white,doubled,minimum size=6mm,font=\footnotesize,rectangle,rounded corners=3mm]

\tikzstyle{gray phase dot}=[dot,fill=gray!40!white,phase dimensions]
\tikzstyle{gray phase ddot}=[ddot,fill=gray!40!white,dphase dimensions]
\tikzstyle{grey phase dot}=[gray phase dot]
\tikzstyle{grey phase ddot}=[gray phase ddot]

\tikzstyle{small phase dimensions}=[minimum size=4mm,font=\tiny,rectangle,rounded corners=2mm,inner sep=0.2mm,outer sep=-2mm]
\tikzstyle{small dphase dimensions}=[minimum size=4mm,font=\tiny,rectangle,rounded corners=2mm,inner sep=0.2mm,outer sep=-2mm]

\tikzstyle{small gray phase dot}=[dot,fill=gray!40!white,small phase dimensions]
\tikzstyle{small gray phase ddot}=[ddot,fill=gray!40!white,small dphase dimensions]

\tikzstyle{small map}=[draw,shape=rectangle,minimum height=4mm,minimum width=4mm,fill=white]

\tikzstyle{cnot}=[fill=white,shape=circle,inner sep=-1.4pt]

\tikzstyle{asym hadamard}=[fill=white,draw,shape=NEbox,inner sep=0.6mm,font=\footnotesize,minimum height=4mm]
\tikzstyle{asym hadamard conj}=[fill=white,draw,shape=NWbox,inner sep=0.6mm,font=\footnotesize,minimum height=4mm]
\tikzstyle{asym hadamard dag}=[fill=white,draw,shape=SEbox,inner sep=0.6mm,font=\footnotesize,minimum height=4mm]

\tikzstyle{hadamard}=[fill=white,draw,inner sep=0.6mm,font=\footnotesize,minimum height=4mm,minimum width=4mm]
\tikzstyle{small hadamard}=[fill=white,draw,inner sep=0.6mm,minimum height=1.5mm,minimum width=1.5mm]
\tikzstyle{small hadamard rotate}=[small hadamard,rotate=45]
\tikzstyle{dhadamard}=[hadamard,doubled]
\tikzstyle{small dhadamard}=[small hadamard,doubled]
\tikzstyle{small dhadamard rotate}=[small hadamard rotate,doubled]
\tikzstyle{antipode}=[white dot,inner sep=0.3mm,font=\footnotesize]

\tikzstyle{scalar}=[diamond,draw,inner sep=0.5pt,font=\small]
\tikzstyle{dscalar}=[diamond,doubled, draw,inner sep=0.5pt,font=\small]

\tikzstyle{small box}=[rectangle,inline text,fill=white,draw,minimum height=5mm,yshift=-0.5mm,minimum width=5mm,font=\small]
\tikzstyle{small gray box}=[small box,fill=gray!30]
\tikzstyle{medium box}=[rectangle,inline text,fill=white,draw,minimum height=5mm,yshift=-0.5mm,minimum width=10mm,font=\small]
\tikzstyle{square box}=[small box] 
\tikzstyle{medium gray box}=[small box,fill=gray!30]
\tikzstyle{semilarge box}=[rectangle,inline text,fill=white,draw,minimum height=5mm,yshift=-0.5mm,minimum width=12.5mm,font=\small]
\tikzstyle{large box}=[rectangle,inline text,fill=white,draw,minimum height=5mm,yshift=-0.5mm,minimum width=15mm,font=\small]
\tikzstyle{large gray box}=[small box,fill=gray!30]

\tikzstyle{Bayes box}=[rectangle,fill=black,draw, minimum height=3mm, minimum width=3mm]

\tikzstyle{gray square point}=[small box,fill=gray!50]

\tikzstyle{dphase box white}=[dhadamard]
\tikzstyle{dphase box gray}=[dhadamard,fill=gray!50!white]
\tikzstyle{phase box white}=[hadamard]
\tikzstyle{phase box gray}=[hadamard,fill=gray!50!white]

\tikzstyle{point}=[regular polygon,regular polygon sides=3,draw,scale=0.75,inner sep=-0.5pt,minimum width=9mm,fill=white,regular polygon rotate=180]
\tikzstyle{point nosep}=[regular polygon,regular polygon sides=3,draw,scale=0.75,inner sep=-2pt,minimum width=9mm,fill=white,regular polygon rotate=180]
\tikzstyle{copoint}=[regular polygon,regular polygon sides=3,draw,scale=0.75,inner sep=-0.5pt,minimum width=9mm,fill=white]
\tikzstyle{dpoint}=[point,doubled]
\tikzstyle{dcopoint}=[copoint,doubled]

\tikzstyle{pointgrow}=[shape=cornerpoint,kpoint common,scale=0.75,inner sep=3pt]
\tikzstyle{pointgrow dag}=[shape=cornercopoint,kpoint common,scale=0.75,inner sep=3pt]

\tikzstyle{wide copoint}=[fill=white,draw,shape=isosceles triangle,shape border rotate=90,isosceles triangle stretches=true,inner sep=0pt,minimum width=1.5cm,minimum height=6.12mm]
\tikzstyle{wide point}=[fill=white,draw,shape=isosceles triangle,shape border rotate=-90,isosceles triangle stretches=true,inner sep=0pt,minimum width=1.5cm,minimum height=6.12mm,yshift=-0.0mm]
\tikzstyle{wide point plus}=[fill=white,draw,shape=isosceles triangle,shape border rotate=-90,isosceles triangle stretches=true,inner sep=0pt,minimum width=1.74cm,minimum height=7mm,yshift=-0.0mm]

\tikzstyle{wide dpoint}=[fill=white,doubled,draw,shape=isosceles triangle,shape border rotate=-90,isosceles triangle stretches=true,inner sep=0pt,minimum width=1.5cm,minimum height=6.12mm,yshift=-0.0mm]

\tikzstyle{tinypoint}=[regular polygon,regular polygon sides=3,draw,scale=0.55,inner sep=-0.15pt,minimum width=6mm,fill=white,regular polygon rotate=180]

\tikzstyle{white point}=[point]
\tikzstyle{white dpoint}=[dpoint]
\tikzstyle{green point}=[white point] 
\tikzstyle{white copoint}=[copoint]
\tikzstyle{gray point}=[point,fill=gray!40!white]
\tikzstyle{gray dpoint}=[gray point,doubled]
\tikzstyle{red point}=[gray point] 
\tikzstyle{gray copoint}=[copoint,fill=gray!40!white]
\tikzstyle{gray dcopoint}=[gray copoint,doubled]

\tikzstyle{white point guide}=[regular polygon,regular polygon sides=3,font=\scriptsize,draw,scale=0.65,inner sep=-0.5pt,minimum width=9mm,fill=white,regular polygon rotate=180]

\tikzstyle{black point}=[point,fill=black,font=\color{white}]
\tikzstyle{black copoint}=[copoint,fill=black,font=\color{white}]

\tikzstyle{tiny gray point}=[tinypoint,fill=gray!40!white]

\tikzstyle{diredge}=[->]
\tikzstyle{ddiredge}=[<->]
\tikzstyle{rdiredge}=[<-]
\tikzstyle{thickdiredge}=[->, very thick]
\tikzstyle{pointer edge}=[->,very thick,gray]
\tikzstyle{pointer edge part}=[very thick,gray]
\tikzstyle{dashed edge}=[dashed]
\tikzstyle{thick dashed edge}=[very thick,dashed]
\tikzstyle{thick gray dashed edge}=[thick dashed edge,gray!40]
\tikzstyle{thick map edge}=[very thick,|->]

\makeatletter
\newcommand{\boxshape}[3]{%
\pgfdeclareshape{#1}{
\inheritsavedanchors[from=rectangle] 
\inheritanchorborder[from=rectangle]
\inheritanchor[from=rectangle]{center}
\inheritanchor[from=rectangle]{north}
\inheritanchor[from=rectangle]{south}
\inheritanchor[from=rectangle]{west}
\inheritanchor[from=rectangle]{east}
\backgroundpath{
\southwest \pgf@xa=\pgf@x \pgf@ya=\pgf@y
\northeast \pgf@xb=\pgf@x \pgf@yb=\pgf@y

\@tempdima=#2
\@tempdimb=#3

\pgfpathmoveto{\pgfpoint{\pgf@xa - 5pt + \@tempdima}{\pgf@ya}}
\pgfpathlineto{\pgfpoint{\pgf@xa - 5pt - \@tempdima}{\pgf@yb}}
\pgfpathlineto{\pgfpoint{\pgf@xb + 5pt + \@tempdimb}{\pgf@yb}}
\pgfpathlineto{\pgfpoint{\pgf@xb + 5pt - \@tempdimb}{\pgf@ya}}
\pgfpathlineto{\pgfpoint{\pgf@xa - 5pt + \@tempdima}{\pgf@ya}}
\pgfpathclose
}
}}

\boxshape{NEbox}{0pt}{3pt}
\boxshape{SEbox}{0pt}{-3pt}
\boxshape{NWbox}{5pt}{0pt}
\boxshape{SWbox}{-5pt}{0pt}
\boxshape{EBox}{-3pt}{3pt}
\boxshape{WBox}{3pt}{-3pt}
\makeatother

\tikzstyle{cloud}=[shape=cloud,draw,minimum width=1.5cm,minimum height=1.5cm]

\tikzstyle{map}=[draw,shape=NEbox,inner sep=1pt,minimum height=4mm,fill=white]
\tikzstyle{dashedmap}=[draw,dashed,shape=NEbox,inner sep=2pt,minimum height=6mm,fill=white]
\tikzstyle{mapdag}=[draw,shape=SEbox,inner sep=1pt,minimum height=4mm,fill=white]
\tikzstyle{mapadj}=[draw,shape=SEbox,inner sep=2pt,minimum height=6mm,fill=white]
\tikzstyle{maptrans}=[draw,shape=SWbox,inner sep=2pt,minimum height=6mm,fill=white]
\tikzstyle{mapconj}=[draw,shape=NWbox,inner sep=2pt,minimum height=6mm,fill=white]

\tikzstyle{medium map}=[draw,shape=NEbox,inner sep=2pt,minimum height=6mm,fill=white,minimum width=7mm]
\tikzstyle{medium map dag}=[draw,shape=SEbox,inner sep=2pt,minimum height=6mm,fill=white,minimum width=7mm]
\tikzstyle{medium map adj}=[draw,shape=SEbox,inner sep=2pt,minimum height=6mm,fill=white,minimum width=7mm]
\tikzstyle{medium map trans}=[draw,shape=SWbox,inner sep=2pt,minimum height=6mm,fill=white,minimum width=7mm]
\tikzstyle{medium map conj}=[draw,shape=NWbox,inner sep=2pt,minimum height=6mm,fill=white,minimum width=7mm]
\tikzstyle{semilarge map}=[draw,shape=NEbox,inner sep=2pt,minimum height=6mm,fill=white,minimum width=9.5mm]
\tikzstyle{semilarge map trans}=[draw,shape=SWbox,inner sep=2pt,minimum height=6mm,fill=white,minimum width=9.5mm]
\tikzstyle{semilarge map adj}=[draw,shape=SEbox,inner sep=2pt,minimum height=6mm,fill=white,minimum width=9.5mm]
\tikzstyle{semilarge map dag}=[draw,shape=SEbox,inner sep=2pt,minimum height=6mm,fill=white,minimum width=9.5mm]
\tikzstyle{semilarge map conj}=[draw,shape=NWbox,inner sep=2pt,minimum height=6mm,fill=white,minimum width=9.5mm]
\tikzstyle{large map}=[draw,shape=NEbox,inner sep=2pt,minimum height=6mm,fill=white,minimum width=12mm]
\tikzstyle{large map conj}=[draw,shape=NWbox,inner sep=2pt,minimum height=6mm,fill=white,minimum width=12mm]
\tikzstyle{very large map}=[draw,shape=NEbox,inner sep=2pt,minimum height=6mm,fill=white,minimum width=17mm]
\tikzstyle{very very large map}=[draw,shape=NEbox,inner sep=2pt,minimum height=6mm,fill=white,minimum width=30mm]

\tikzstyle{medium dmap}=[draw,doubled,shape=NEbox,inner sep=2pt,minimum height=6mm,fill=white,minimum width=7mm]
\tikzstyle{medium dmap dag}=[draw,doubled,shape=SEbox,inner sep=2pt,minimum height=6mm,fill=white,minimum width=7mm]
\tikzstyle{medium dmap adj}=[draw,doubled,shape=SEbox,inner sep=2pt,minimum height=6mm,fill=white,minimum width=7mm]
\tikzstyle{medium dmap trans}=[draw,doubled,shape=SWbox,inner sep=2pt,minimum height=6mm,fill=white,minimum width=7mm]
\tikzstyle{medium dmap conj}=[draw,doubled,shape=NWbox,inner sep=2pt,minimum height=6mm,fill=white,minimum width=7mm]
\tikzstyle{semilarge dmap}=[draw,doubled,shape=NEbox,inner sep=2pt,minimum height=6mm,fill=white,minimum width=9.5mm]
\tikzstyle{semilarge dmap trans}=[draw,doubled,shape=SWbox,inner sep=2pt,minimum height=6mm,fill=white,minimum width=9.5mm]
\tikzstyle{semilarge dmap adj}=[draw,doubled,shape=SEbox,inner sep=2pt,minimum height=6mm,fill=white,minimum width=9.5mm]
\tikzstyle{semilarge dmap dag}=[draw,doubled,shape=SEbox,inner sep=2pt,minimum height=6mm,fill=white,minimum width=9.5mm]
\tikzstyle{semilarge dmap conj}=[draw,doubled,shape=NWbox,inner sep=2pt,minimum height=6mm,fill=white,minimum width=9.5mm]
\tikzstyle{large dmap}=[draw,doubled,shape=NEbox,inner sep=2pt,minimum height=6mm,fill=white,minimum width=12mm]
\tikzstyle{large dmap conj}=[draw,doubled,shape=NWbox,inner sep=2pt,minimum height=6mm,fill=white,minimum width=12mm]
\tikzstyle{large dmap trans}=[draw,doubled,shape=SWbox,inner sep=2pt,minimum height=6mm,fill=white,minimum width=12mm]
\tikzstyle{large dmap adj}=[draw,doubled,shape=SEbox,inner sep=2pt,minimum height=6mm,fill=white,minimum width=12mm]
\tikzstyle{large dmap dag}=[draw,doubled,shape=SEbox,inner sep=2pt,minimum height=6mm,fill=white,minimum width=12mm]
\tikzstyle{very large dmap}=[draw,doubled,shape=NEbox,inner sep=2pt,minimum height=6mm,fill=white,minimum width=19.5mm]

\tikzstyle{muxbox}=[draw,shape=rectangle,minimum height=3mm,minimum width=3mm,fill=white]
\tikzstyle{dmuxbox}=[muxbox,doubled]

\tikzstyle{box}=[draw,shape=rectangle,inner sep=2pt,minimum height=6mm,minimum width=6mm,fill=white]
\tikzstyle{dbox}=[draw,doubled,shape=rectangle,inner sep=2pt,minimum height=6mm,minimum width=6mm,fill=white]
\tikzstyle{dmap}=[draw,doubled,shape=NEbox,inner sep=2pt,minimum height=6mm,fill=white]
\tikzstyle{dmapdag}=[draw,doubled,shape=SEbox,inner sep=2pt,minimum height=6mm,fill=white]
\tikzstyle{dmapadj}=[draw,doubled,shape=SEbox,inner sep=2pt,minimum height=6mm,fill=white]
\tikzstyle{dmaptrans}=[draw,doubled,shape=SWbox,inner sep=2pt,minimum height=6mm,fill=white]
\tikzstyle{dmapconj}=[draw,doubled,shape=NWbox,inner sep=2pt,minimum height=6mm,fill=white]

\tikzstyle{ddmap}=[draw,doubled,dashed,shape=NEbox,inner sep=2pt,minimum height=6mm,fill=white]
\tikzstyle{ddmapdag}=[draw,doubled,dashed,shape=SEbox,inner sep=2pt,minimum height=6mm,fill=white]
\tikzstyle{ddmapadj}=[draw,doubled,dashed,shape=SEbox,inner sep=2pt,minimum height=6mm,fill=white]
\tikzstyle{ddmaptrans}=[draw,doubled,dashed,shape=SWbox,inner sep=2pt,minimum height=6mm,fill=white]
\tikzstyle{ddmapconj}=[draw,doubled,dashed,shape=NWbox,inner sep=2pt,minimum height=6mm,fill=white]

\boxshape{sNEbox}{0pt}{3pt}
\boxshape{sSEbox}{0pt}{-3pt}
\boxshape{sNWbox}{3pt}{0pt}
\boxshape{sSWbox}{-3pt}{0pt}
\tikzstyle{smap}=[draw,shape=sNEbox,fill=white]
\tikzstyle{smapdag}=[draw,shape=sSEbox,fill=white]
\tikzstyle{smapadj}=[draw,shape=sSEbox,fill=white]
\tikzstyle{smaptrans}=[draw,shape=sSWbox,fill=white]
\tikzstyle{smapconj}=[draw,shape=sNWbox,fill=white]

\tikzstyle{dsmap}=[draw,dashed,shape=sNEbox,fill=white]
\tikzstyle{dsmapdag}=[draw,dashed,shape=sSEbox,fill=white]
\tikzstyle{dsmaptrans}=[draw,dashed,shape=sSWbox,fill=white]
\tikzstyle{dsmapconj}=[draw,dashed,shape=sNWbox,fill=white]

\boxshape{mNEbox}{0pt}{10pt}
\boxshape{mSEbox}{0pt}{-10pt}
\boxshape{mNWbox}{10pt}{0pt}
\boxshape{mSWbox}{-10pt}{0pt}
\tikzstyle{mmap}=[draw,shape=mNEbox]
\tikzstyle{mmapdag}=[draw,shape=mSEbox]
\tikzstyle{mmaptrans}=[draw,shape=mSWbox]
\tikzstyle{mmapconj}=[draw,shape=mNWbox]

\tikzstyle{mmapgray}=[draw,fill=gray!40!white,shape=mNEbox]
\tikzstyle{smapgray}=[draw,fill=gray!40!white,shape=sNEbox]

\makeatletter

\pgfdeclareshape{cornerpoint}{
\inheritsavedanchors[from=rectangle] 
\inheritanchorborder[from=rectangle]
\inheritanchor[from=rectangle]{center}
\inheritanchor[from=rectangle]{north}
\inheritanchor[from=rectangle]{south}
\inheritanchor[from=rectangle]{west}
\inheritanchor[from=rectangle]{east}
\backgroundpath{
\southwest \pgf@xa=\pgf@x \pgf@ya=\pgf@y
\northeast \pgf@xb=\pgf@x \pgf@yb=\pgf@y

\pgfmathsetmacro{\pgf@shorten@left}{\pgfkeysvalueof{/tikz/shorten left}}
\pgfmathsetmacro{\pgf@shorten@right}{\pgfkeysvalueof{/tikz/shorten right}}

\pgfpathmoveto{\pgfpoint{0.5 * (\pgf@xa + \pgf@xb)}{\pgf@ya - 5pt}}
\pgfpathlineto{\pgfpoint{\pgf@xa - 8pt + \pgf@shorten@left}{\pgf@yb - 1.5 * \pgf@shorten@left}}
\pgfpathlineto{\pgfpoint{\pgf@xa - 8pt + \pgf@shorten@left}{\pgf@yb}}
\pgfpathlineto{\pgfpoint{\pgf@xb + 8pt - \pgf@shorten@right}{\pgf@yb}}
\pgfpathlineto{\pgfpoint{\pgf@xb + 8pt - \pgf@shorten@right}{\pgf@yb - 1.5 * \pgf@shorten@right}}
\pgfpathclose
}
}

\pgfdeclareshape{cornercopoint}{
\inheritsavedanchors[from=rectangle] 
\inheritanchorborder[from=rectangle]
\inheritanchor[from=rectangle]{center}
\inheritanchor[from=rectangle]{north}
\inheritanchor[from=rectangle]{south}
\inheritanchor[from=rectangle]{west}
\inheritanchor[from=rectangle]{east}
\backgroundpath{
\southwest \pgf@xa=\pgf@x \pgf@ya=\pgf@y
\northeast \pgf@xb=\pgf@x \pgf@yb=\pgf@y

\pgfmathsetmacro{\pgf@shorten@left}{\pgfkeysvalueof{/tikz/shorten left}}
\pgfmathsetmacro{\pgf@shorten@right}{\pgfkeysvalueof{/tikz/shorten right}}

\pgfpathmoveto{\pgfpoint{0.5 * (\pgf@xa + \pgf@xb)}{\pgf@yb + 5pt}}
\pgfpathlineto{\pgfpoint{\pgf@xa - 8pt + \pgf@shorten@left}{\pgf@ya + 1.5 * \pgf@shorten@left}}
\pgfpathlineto{\pgfpoint{\pgf@xa - 8pt + \pgf@shorten@left}{\pgf@ya}}
\pgfpathlineto{\pgfpoint{\pgf@xb + 8pt - \pgf@shorten@right}{\pgf@ya}}
\pgfpathlineto{\pgfpoint{\pgf@xb + 8pt - \pgf@shorten@right}{\pgf@ya + 1.5 * \pgf@shorten@right}}
\pgfpathclose
}
}

\makeatother

\pgfkeyssetvalue{/tikz/shorten left}{0pt}
\pgfkeyssetvalue{/tikz/shorten right}{0pt}

\tikzstyle{kpoint common}=[draw,fill=white,inner sep=1pt,minimum height=4mm]
\tikzstyle{kpoint sc}=[shape=cornerpoint,kpoint common]
\tikzstyle{kpoint adjoint sc}=[shape=cornercopoint,kpoint common]
\tikzstyle{kpoint}=[shape=cornerpoint,shorten left=5pt,kpoint common]
\tikzstyle{kpoint adjoint}=[shape=cornercopoint,shorten left=5pt,kpoint common]
\tikzstyle{kpoint conjugate}=[shape=cornerpoint,shorten right=5pt,kpoint common]
\tikzstyle{kpoint transpose}=[shape=cornercopoint,shorten right=5pt,kpoint common]
\tikzstyle{kpoint symm}=[shape=cornerpoint,shorten left=5pt,shorten right=5pt,kpoint common]

\tikzstyle{wide kpoint sc}=[shape=cornerpoint,kpoint common, minimum width=1 cm]
\tikzstyle{wide kpointdag sc}=[shape=cornercopoint,kpoint common, minimum width=1 cm]

\tikzstyle{black kpoint}=[shape=cornerpoint,shorten left=5pt,kpoint common,fill=black,font=\color{white}]

\tikzstyle{black kpoint sm}=[shape=cornerpoint,shorten left=5pt,kpoint common,fill=black,font=\color{white},scale=0.75]

\tikzstyle{black kpoint adjoint}=[shape=cornercopoint,shorten left=5pt,kpoint common,fill=black,font=\color{white}]
\tikzstyle{black kpointadj}=[shape=cornercopoint,shorten left=5pt,kpoint common,fill=black,font=\color{white}]

\tikzstyle{black kpointadj sm}=[shape=cornercopoint,shorten left=5pt,kpoint common,fill=black,font=\color{white},scale=0.75]

\tikzstyle{black dkpoint}=[shape=cornerpoint,shorten left=5pt,kpoint common,fill=black, doubled,font=\color{white}]
\tikzstyle{black dkpoint adjoint}=[shape=cornercopoint,shorten left=5pt,kpoint common,fill=black, doubled,font=\color{white}]
\tikzstyle{black dkpointadj}=[shape=cornercopoint,shorten left=5pt,kpoint common,fill=black, doubled,font=\color{white}]

\tikzstyle{black dkpoint sm}=[shape=cornerpoint,shorten left=5pt,kpoint common,fill=black, doubled,font=\color{white},scale=0.75]
\tikzstyle{black dkpointadj sm}=[shape=cornercopoint,shorten left=5pt,kpoint common,fill=black, doubled,font=\color{white},scale=0.75]

\tikzstyle{kpointdag}=[kpoint adjoint]
\tikzstyle{kpointadj}=[kpoint adjoint]
\tikzstyle{kpointconj}=[kpoint conjugate]
\tikzstyle{kpointtrans}=[kpoint transpose]

\tikzstyle{big kpoint}=[kpoint, minimum width=1.2 cm, minimum height=8mm, inner sep=4pt, text depth=3mm]

\tikzstyle{wide kpoint}=[kpoint, minimum width=1 cm, inner sep=2pt]
\tikzstyle{wide kpointdag}=[kpointdag, minimum width=1 cm, inner sep=2pt]
\tikzstyle{wide kpointconj}=[kpointconj, minimum width=1 cm, inner sep=2pt]
\tikzstyle{wide kpointtrans}=[kpointtrans, minimum width=1 cm, inner sep=2pt]

\tikzstyle{wider kpoint}=[kpoint, minimum width=1.25 cm, inner sep=2pt]
\tikzstyle{wider kpointdag}=[kpointdag, minimum width=1.25 cm, inner sep=2pt]
\tikzstyle{wider kpointconj}=[kpointconj, minimum width=1.25 cm, inner sep=2pt]
\tikzstyle{wider kpointtrans}=[kpointtrans, minimum width=1.25 cm, inner sep=2pt]

\tikzstyle{gray kpoint}=[kpoint,fill=gray!50!white]
\tikzstyle{gray kpointdag}=[kpointdag,fill=gray!50!white]
\tikzstyle{gray kpointadj}=[kpointadj,fill=gray!50!white]
\tikzstyle{gray kpointconj}=[kpointconj,fill=gray!50!white]
\tikzstyle{gray kpointtrans}=[kpointtrans,fill=gray!50!white]

\tikzstyle{gray dkpoint}=[kpoint,fill=gray!50!white,doubled]
\tikzstyle{gray dkpointdag}=[kpointdag,fill=gray!50!white,doubled]
\tikzstyle{gray dkpointadj}=[kpointadj,fill=gray!50!white,doubled]
\tikzstyle{gray dkpointconj}=[kpointconj,fill=gray!50!white,doubled]
\tikzstyle{gray dkpointtrans}=[kpointtrans,fill=gray!50!white,doubled]

\tikzstyle{white label}=[draw,fill=white,rectangle,inner sep=0.7 mm]
\tikzstyle{gray label}=[draw,fill=gray!50!white,rectangle,inner sep=0.7 mm]
\tikzstyle{black label}=[draw,fill=black,rectangle,inner sep=0.7 mm]

\tikzstyle{dkpoint}=[kpoint,doubled]
\tikzstyle{wide dkpoint}=[wide kpoint,doubled]
\tikzstyle{dkpointdag}=[kpoint adjoint,doubled]
\tikzstyle{wide dkpointdag}=[wide kpointdag,doubled]
\tikzstyle{dkcopoint}=[kpoint adjoint,doubled]
\tikzstyle{dkpointadj}=[kpoint adjoint,doubled]
\tikzstyle{dkpointconj}=[kpoint conjugate,doubled]
\tikzstyle{dkpointtrans}=[kpoint transpose,doubled]

\tikzstyle{kscalar}=[kpoint common, shape=EBox, inner xsep=-1pt, inner ysep=3pt,font=\small]
\tikzstyle{kscalarconj}=[kpoint common, shape=WBox, inner xsep=-1pt, inner ysep=3pt,font=\small]

\tikzstyle{spekpoint}=[kpoint sc,minimum height=5mm,inner sep=3pt]
\tikzstyle{spekcopoint}=[kpoint adjoint sc,minimum height=5mm,inner sep=3pt]

\tikzstyle{dspekpoint}=[spekpoint,doubled]
\tikzstyle{dspekcopoint}=[spekcopoint,doubled]

 \tikzstyle{upground}=[circuit ee IEC,thick,ground,rotate=90,scale=2.5]
 \tikzstyle{downground}=[circuit ee IEC,thick,ground,rotate=-90,scale=2.5]
 \tikzstyle{bigground}=[regular polygon,regular polygon sides=3,draw=gray,scale=0.50,inner sep=-0.5pt,minimum width=10mm,fill=gray]
 \tikzstyle{maxmix}=[regular polygon,regular polygon sides=3,draw=black,xscale=0.4,yscale=0.3,inner sep=-0.5pt,minimum width=10mm,fill=gray,regular polygon rotate=180]

\tikzstyle{arrs}=[-latex,font=\small,auto]
\tikzstyle{arrow plain}=[arrs]
\tikzstyle{arrow dashed}=[dashed,arrs]
\tikzstyle{arrow bold}=[very thick,arrs]
\tikzstyle{arrow hide}=[draw=white!0,-]
\tikzstyle{arrow reverse}=[latex-]
\tikzstyle{cdnode}=[]

\tikzstyle{supermap}=[fill=white]

\usetikzlibrary{backgrounds}
\usetikzlibrary{arrows}
\usetikzlibrary{shapes,shapes.geometric,shapes.misc}
\usetikzlibrary{positioning,chains,fit,shapes,calc}

\tikzstyle{tikzfig}=[baseline=-0.25em,scale=4]

\pgfkeys{/tikz/tikzit fill/.initial=0}
\pgfkeys{/tikz/tikzit draw/.initial=0}
\pgfkeys{/tikz/tikzit shape/.initial=0}
\pgfkeys{/tikz/tikzit category/.initial=0}

\pgfdeclarelayer{edgelayer}
\pgfdeclarelayer{nodelayer}
\pgfsetlayers{background,edgelayer,nodelayer,main}

\tikzstyle{none}=[inner sep=0mm]

\tikzstyle{every loop}=[]

\newcommand\abs[1]{\left|#1\right|}

\newcommand{\mr}{\mathrm}

\def\be{\begin{equation}}
\def\ee{\end{equation}}
\def\ba{\begin{align}}
\def\ea{\end{align}}

\newcommand{\C}{\ensuremath{\mathbb{C}}}

\DeclareMathOperator{\Tr}{Tr}

\newcommand{\ca}{\mathcal A}
\newcommand{\cb}{\mathcal B}
\newcommand{\cc}{\mathcal C}
\newcommand{\cd}{\mathcal D}

\newcommand{\ch}{\mathcal H}
\newcommand{\ci}{\mathcal I}

\newcommand{\ck}{\mathcal K}
\newcommand{\cl}{\mathcal L}

\newcommand{\cs}{\mathcal S}

\newcommand{\cu}{\mathcal U}
\newcommand{\cv}{\mathcal V}
\newcommand{\cw}{\mathcal W}

\newcommand{\cz}{\mathcal Z}

\usepackage{bbm}
\def\C{{\mathbbm C}}

\newcommand{\spc}[1]{\mathcal{#1}}

\newcommand{\ketbra}[2]{\ket{#1} \!  \bra{#2}}

\newcommand{\Proof}{{\bf Proof. \,}}

\newcommand{\map}[1]{\mathcal{#1}}

\newtheorem{definition}{Definition}
\newtheorem{corollary}{Corollary}
\newtheorem{theorem}{Theorem}

\newtheorem{lemma}{Lemma}

\newcommand{\correction}[1]{{ #1}}
\mathchardef\mhyphen="2D

\begin{document}
\setlength{\textheight}{8.0truein}    

\runninghead{Universal control of quantum processes using sector-preserving channels}
            {Augustin Vanrietvelde and Giulio Chiribella}

\normalsize\textlineskip
\thispagestyle{empty}
\setcounter{page}{1320}

\copyrightheading{21}{15\&16}{2021}{1320--1352}

\vspace*{0.88truein}

\alphfootnote

\fpage{1320}

\centerline{\bf
UNIVERSAL CONTROL OF QUANTUM PROCESSES}
\vspace*{0.035truein}
\centerline{\bf 
USING SECTOR-PRESERVING CHANNELS}
\vspace*{0.37truein}
\centerline{\footnotesize
AUGUSTIN VANRIETVELDE
}
\vspace*{0.015truein}
\centerline{\footnotesize\it Quantum Group, Department of Computer Science, University of Oxford}
\baselineskip=10pt
\centerline{\footnotesize\it Department of Physics, Imperial College London}
\baselineskip=10pt
\centerline{\footnotesize\it HKU-Oxford Joint Laboratory for Quantum Information and Computation}\baselineskip=10pt
\centerline{\footnotesize\it a.vanrietvelde18@imperial.ac.uk}
\vspace*{10pt}
\centerline{\footnotesize 
GIULIO CHIRIBELLA 
}
\vspace*{0.015truein}
\centerline{\footnotesize\it QICI  Quantum  Information  and  Computation  Initiative,  Department  of  Computer  Science}
\baselineskip=10pt
\centerline{\footnotesize\it  The  University  of  Hong  Kong}
\baselineskip=10pt
\centerline{\footnotesize\it Quantum Group, Department of Computer Science, University of Oxford}
\baselineskip=10pt
\centerline{\footnotesize\it Perimeter  Institute  for  Theoretical  Physics, Waterloo}\baselineskip=10pt
\centerline{\footnotesize\it giulio@cs.hku.hk}
\vspace*{0.225truein}
\publisher{July 5, 2021}{October 6, 2021}

\vspace*{0.21truein}

\abstracts{
No quantum circuit can turn a completely unknown unitary gate into its coherently controlled version.  Yet,  coherent control of unknown gates has been realised in experiments, making use of a different type of initial resources. Here, we formalise the task achieved by these experiments,  extending it to the control of arbitrary  noisy channels, and to more general types of control involving higher dimensional control systems.   For the standard notion of coherent control, we identify the information-theoretic resource for controlling an arbitrary quantum channel on a $d$-dimensional system: specifically, the resource is  an extended quantum  channel acting as the original channel on a $d$-dimensional sector of a $(d+1)$-dimensional system.  Using this resource, arbitrary controlled channels can be built with a universal circuit architecture. We then extend the standard notion of control to more general notions, including  control of multiple channels with possibly different input and output systems.  Finally, we develop a theoretical framework,  called supermaps on routed channels, which provides a compact representation of coherent control as an operation performed on the extended channels, and highlights  the way the operation acts on different sectors.
}{}{}

\vspace*{10pt}

\keywords{coherent control, superposition of quantum channels, routed quantum circuits, sector-preserving channels}
\vspace*{3pt}
\communicate{R~Jozsa~\&~J~Eisert }

\vspace*{1pt}\textlineskip    

\section{Introduction} \label{sec:intro}

A number of quantum algorithms, such as Kitaev's phase estimation algorithm \cite{1995quant.ph.11026K} and the DQC1 trace estimation algorithm \cite{Knill1998},  are based on the use of controlled unitary gates.  Controlled gates represent a quantum version of the \textit{if-then} clause, in which a subroutine is executed depending on the value of a control variable. In the controlled gate ${\tt ctrl} \mhyphen U$, the quantum state of a control system determines whether or not a target system is subject to a given unitary gate $U$.  When the control system is in a superposition state, the target system experiences a coherent superposition of quantum evolutions \cite{aharonov1990superpositions}. Quantum programming languages that exploit  coherent control of quantum gates have been proposed in Refs.~\cite{altenkirch2005functional,Ying2016, sabry2018symmetric}.

\begin{figure}    \centering

\begin{subfigure}[b]{.85\linewidth}
\centering
\scalebox{.85}{$\mathlarger{\mathlarger{\mathlarger{\nexists}}} \quad %
\begin{tikzpicture}
	\begin{pgfonlayer}{nodelayer}
		\node [style=none] (0) at (-2.75, 3) {};
		\node [style=none] (1) at (-2.75, -3) {};
		\node [style=none] (2) at (2, -3) {};
		\node [style=none] (3) at (2, 3) {};
		\node [style=none] (4) at (-0.25, 1.75) {};
		\node [style=none] (5) at (-0.25, -1.75) {};
		\node [style=none] (6) at (1, 1.75) {};
		\node [style=none] (7) at (1, -0.75) {};
		\node [style=none] (8) at (1, -1.75) {};
		\node [style=none] (9) at (2, 1.75) {};
		\node [style=none] (10) at (2, -1.75) {};
		\node [style=none] (11) at (1, 0.75) {};
		\node [style=none] (12) at (-1.5, -3) {};
		\node [style=none] (13) at (-1.5, -4) {};
		\node [style=right label] (14) at (-1.5, -3.625) {$C$};
		\node [style=none] (15) at (1, 3) {};
		\node [style=none] (16) at (1, 4) {};
		\node [style=right label] (17) at (1, 3.625) {$T$};
		\node [style=right label] (18) at (1, 1) {$T$};
		\node [style=right label] (19) at (1, -1.25) {$T$};
		\node [style=none] (21) at (-1.5, 3) {};
		\node [style=none] (22) at (-1.5, 4) {};
		\node [style=right label] (23) at (-1.5, 3.625) {$C$};
		\node [style=none] (24) at (1, -3) {};
		\node [style=none] (25) at (1, -4) {};
		\node [style=right label] (26) at (1, -3.625) {$T$};
		\node [style=none] (27) at (-1.5, 0) {$\tt{CTRL}$};
	\end{pgfonlayer}
	\begin{pgfonlayer}{edgelayer}
		\draw [style=supermap] (2.center)
			 to (1.center)
			 to (0.center)
			 to (3.center)
			 to (9.center)
			 to (4.center)
			 to (5.center)
			 to (10.center)
			 to cycle;
		\draw (6.center) to (11.center);
		\draw (7.center) to (8.center);
		\draw (12.center) to (13.center);
		\draw (15.center) to (16.center);
		\draw (21.center) to (22.center);
		\draw (24.center) to (25.center);
	\end{pgfonlayer}
\end{tikzpicture}} \textrm{ such that} \quad \forall \,\,\, %
\begin{tikzpicture}
	\begin{pgfonlayer}{nodelayer}
		\node [style=map] (0) at (0, 0) {$\cu$};
		\node [style=none] (1) at (0, 1.75) {};
		\node [style=none] (2) at (0, -1.75) {};
		\node [style=right label] (3) at (0, 1.5) {$T$};
		\node [style=right label] (4) at (0, -1.5) {$T$};
	\end{pgfonlayer}
	\begin{pgfonlayer}{edgelayer}
		\draw (1.center) to (2.center);
	\end{pgfonlayer}
\end{tikzpicture}} , \quad %
\begin{tikzpicture}
	\begin{pgfonlayer}{nodelayer}
		\node [style=none] (0) at (-2.5, 3) {};
		\node [style=none] (1) at (-2.5, -3) {};
		\node [style=none] (2) at (2.25, -3) {};
		\node [style=none] (3) at (2.25, 3) {};
		\node [style=none] (4) at (0, 1.75) {};
		\node [style=none] (5) at (0, -1.75) {};
		\node [style=none] (7) at (1.25, 1.75) {};
		\node [style=none] (8) at (1.25, -1.75) {};
		\node [style=none] (9) at (2.25, 1.75) {};
		\node [style=none] (10) at (2.25, -1.75) {};
		\node [style=none] (12) at (-1.25, -3) {};
		\node [style=none] (13) at (-1.25, -4) {};
		\node [style=right label] (14) at (-1.25, -3.625) {$C$};
		\node [style=none] (15) at (1.25, 3) {};
		\node [style=none] (16) at (1.25, 4) {};
		\node [style=right label] (17) at (1.25, 3.625) {$T$};
		\node [style=right label] (18) at (1.25, 1) {$T$};
		\node [style=right label] (19) at (1.25, -1.25) {$T$};
		\node [style=none] (20) at (-1.25, 3) {};
		\node [style=none] (21) at (-1.25, 4) {};
		\node [style=right label] (22) at (-1.25, 3.625) {$C$};
		\node [style=none] (23) at (1.25, -3) {};
		\node [style=none] (24) at (1.25, -4) {};
		\node [style=right label] (25) at (1.25, -3.625) {$T$};
		\node [style=none] (26) at (-1.25, 0) {$\tt{CTRL}$};
		\node [style=map] (27) at (1.25, 0) {$\cu$};
	\end{pgfonlayer}
	\begin{pgfonlayer}{edgelayer}
		\draw [style=supermap] (0.center)
			 to (3.center)
			 to (9.center)
			 to (4.center)
			 to (5.center)
			 to (10.center)
			 to (2.center)
			 to (1.center)
			 to cycle;
		\draw (7.center) to (8.center);
		\draw (12.center) to (13.center);
		\draw (15.center) to (16.center);
		\draw (20.center) to (21.center);
		\draw (23.center) to (24.center);
	\end{pgfonlayer}
\end{tikzpicture}} \quad = \quad %
\begin{tikzpicture}
	\begin{pgfonlayer}{nodelayer}
		\node [style=none] (2) at (2.5, 1.75) {};
		\node [style=none] (3) at (0.5, -1.75) {};
		\node [style=none] (4) at (2.5, -1.75) {};
		\node [style=none] (5) at (0.5, 1.75) {};
		\node [style=right label] (6) at (2.5, 1.5) {$T$};
		\node [style=right label] (7) at (2.5, -1.5) {$T$};
		\node [style=right label] (8) at (0.5, 1.5) {$C$};
		\node [style=right label] (9) at (0.5, -1.5) {$C$};
		\node [style=large map] (10) at (1.5, 0) {$\mr{ctrl}\mhyphen\cu $};
	\end{pgfonlayer}
	\begin{pgfonlayer}{edgelayer}
		\draw (5.center) to (3.center);
		\draw (2.center) to (4.center);
	\end{pgfonlayer}
\end{tikzpicture}}$}
\subcaption{ \footnotesize {\correction{{\bf No-go theorem on coherent control with black box channels} \cite{Soeda2013,Araujo2014,Chiribella_2016,Bisio_2016,Thompson2018,gavorova2020topological}.    No supermap can convert  an arbitrary  unitary channel $\map U$, acting on a target system $T$,  into its controlled version ${\tt ctrl}\mhyphen \map U$, acting on a control system $C$ and on the target $T$.   
}}}
\label{fig:NoGoTheorem}
\end{subfigure}%

\begin{subfigure}[b]{.85\linewidth}
\centering
\scalebox{.85}{$\mathlarger{\mathlarger{\mathlarger{\exists}}} \quad %
\begin{tikzpicture}
	\begin{pgfonlayer}{nodelayer}
		\node [style=none] (0) at (-2, 3) {};
		\node [style=none] (1) at (-2, -3) {};
		\node [style=none] (2) at (2.75, -3) {};
		\node [style=none] (3) at (2.75, 3) {};
		\node [style=none] (4) at (0.5, 1.75) {};
		\node [style=none] (5) at (0.5, -1.75) {};
		\node [style=none] (6) at (1.75, 1.75) {};
		\node [style=none] (7) at (1.75, -0.75) {};
		\node [style=none] (8) at (1.75, -1.75) {};
		\node [style=none] (9) at (2.75, 1.75) {};
		\node [style=none] (10) at (2.75, -1.75) {};
		\node [style=none] (11) at (1.75, 0.75) {};
		\node [style=none] (12) at (-0.75, -3) {};
		\node [style=none] (13) at (-0.75, -4) {};
		\node [style=right label] (14) at (-0.75, -3.625) {$C$};
		\node [style=none] (15) at (1.75, 3) {};
		\node [style=none] (16) at (1.75, 4) {};
		\node [style=right label] (17) at (1.75, 3.625) {$T$};
		\node [style=right label] (18) at (1.75, 1) {$S$};
		\node [style=right label] (19) at (1.75, -1.25) {$S$};
		\node [style=none] (20) at (-0.75, 3) {};
		\node [style=none] (21) at (-0.75, 4) {};
		\node [style=right label] (22) at (-0.75, 3.625) {$C$};
		\node [style=none] (23) at (1.75, -3) {};
		\node [style=none] (24) at (1.75, -4) {};
		\node [style=right label] (25) at (1.75, -3.625) {$T$};
		\node [style=none] (26) at (-0.75, 0) {$\tt{CTRL}$};
	\end{pgfonlayer}
	\begin{pgfonlayer}{edgelayer}
		\draw (6.center) to (11.center);
		\draw (7.center) to (8.center);
		\draw (12.center) to (13.center);
		\draw (15.center) to (16.center);
		\draw (20.center) to (21.center);
		\draw (23.center) to (24.center);
		\draw [style=supermap] (0.center)
			 to (1.center)
			 to (2.center)
			 to (10.center)
			 to (5.center)
			 to (4.center)
			 to (9.center)
			 to (3.center)
			 to cycle;
	\end{pgfonlayer}
\end{tikzpicture}} \textrm{ such that} \quad \forall \,\,\, %
\begin{tikzpicture}
	\begin{pgfonlayer}{nodelayer}
		\node [style=none] (0) at (0, 1.75) {};
		\node [style=none] (1) at (0, -1.75) {};
		\node [style=right label] (2) at (0, 1) {$S$};
		\node [style=right label] (3) at (0, -1.25) {$S$};
		\node [style=medium map] (4) at (0, 0) {$\widetilde{\cc}$};
	\end{pgfonlayer}
	\begin{pgfonlayer}{edgelayer}
		\draw (0.center) to (1.center);
	\end{pgfonlayer}
\end{tikzpicture}} \, , \quad %
\begin{tikzpicture}
	\begin{pgfonlayer}{nodelayer}
		\node [style=none] (0) at (-2.5, 3) {};
		\node [style=none] (1) at (-2.5, -3) {};
		\node [style=none] (2) at (2.5, -3) {};
		\node [style=none] (3) at (2.5, 3) {};
		\node [style=none] (4) at (0, 1.75) {};
		\node [style=none] (5) at (0, -1.75) {};
		\node [style=none] (7) at (1.5, 1.75) {};
		\node [style=none] (8) at (1.5, -1.75) {};
		\node [style=none] (9) at (2.5, 1.75) {};
		\node [style=none] (10) at (2.5, -1.75) {};
		\node [style=none] (12) at (-1.25, -3) {};
		\node [style=none] (13) at (-1.25, -4) {};
		\node [style=right label] (14) at (-1.25, -3.625) {$C$};
		\node [style=none] (15) at (1.5, 3) {};
		\node [style=none] (16) at (1.5, 4) {};
		\node [style=right label] (17) at (1.5, 3.625) {$T$};
		\node [style=right label] (18) at (1.5, 1) {$S$};
		\node [style=right label] (19) at (1.5, -1.25) {$S$};
		\node [style=none] (20) at (-1.25, 3) {};
		\node [style=none] (21) at (-1.25, 4) {};
		\node [style=right label] (22) at (-1.25, 3.625) {$C$};
		\node [style=none] (23) at (1.5, -3) {};
		\node [style=none] (24) at (1.5, -4) {};
		\node [style=right label] (25) at (1.5, -3.625) {$T$};
		\node [style=none] (26) at (-1.25, 0) {$\tt{CTRL}$};
		\node [style=medium map] (27) at (1.5, 0) {$\widetilde{\cc}$};
	\end{pgfonlayer}
	\begin{pgfonlayer}{edgelayer}
		\draw [style=supermap] (3.center)
			 to (9.center)
			 to (4.center)
			 to (5.center)
			 to (10.center)
			 to (2.center)
			 to (1.center)
			 to (0.center)
			 to cycle;
		\draw (7.center) to (8.center);
		\draw (12.center) to (13.center);
		\draw (15.center) to (16.center);
		\draw (20.center) to (21.center);
		\draw (23.center) to (24.center);
	\end{pgfonlayer}
\end{tikzpicture}} \quad = \quad %
\begin{tikzpicture}
	\begin{pgfonlayer}{nodelayer}
		\node [style=none] (2) at (1, 1.75) {};
		\node [style=none] (3) at (-1, -1.75) {};
		\node [style=none] (4) at (1, -1.75) {};
		\node [style=none] (5) at (-1, 1.75) {};
		\node [style=right label] (6) at (1, 1.5) {$T$};
		\node [style=right label] (7) at (1, -1.5) {$T$};
		\node [style=right label] (8) at (-1, 1.5) {$C$};
		\node [style=right label] (9) at (-1, -1.5) {$C$};
		\node [style=large map] (10) at (0, 0) {$\mr{ctrl}\mhyphen\cc $};
	\end{pgfonlayer}
	\begin{pgfonlayer}{edgelayer}
		\draw (5.center) to (3.center);
		\draw (2.center) to (4.center);
	\end{pgfonlayer}
\end{tikzpicture}}$}
\subcaption{ \footnotesize \correction{
     {\bf  Universal coherent control with sector-preserving channels.}  
     There exists a supermap $\tt CTRL $ that transforms arbitrary sector-preserving channels $\widetilde {\cal C} $ acting on an extended input system $S$ (with Hilbert space $\spc H_S   =  \C   \oplus \spc H_T $) into arbitrary  controlled channels ${\tt ctrl}\mhyphen {\cal C}$. In particular, the supermap $\tt CTRL$ maps arbitrary sector-preserving unitary channels $\widetilde {\cal U}$ into   the corresponding controlled unitary channels ${\tt ctrl}\mhyphen {\cal U}$. }}
     \label{fig:YesGoTheorem}
\end{subfigure}%

\setcounter{figure}{0}

\fcaption{\correction{Comparison between the standard no-go theorem and our universal controllisation circuit.}}

   \label{fig:2Theorems}
\end{figure}

The standard way to construct quantum controlled gates is via universal gate sets. To build the controlled gate ${\tt ctrl} \mhyphen U$, one first decomposes the gate $U$ into elementary gates, and then adds control to each of these gates \cite{Barenco_1995}.  This construction, however,  requires a decomposition of the gate $U$ into elementary gates.    In many applications, such as quantum factoring \cite{shor}, the decomposition  is known, because  the gate $U$ is the quantum realisation of a classical function, for which a classical program is given.  In other applications, however, the gate $U$ may be completely unknown:  in a cloud computing scenario,  for example, the gate $U$ may be implemented remotely by a server, and the program that generated $U$ may be unknown to the client.  In  these situations, it would be  convenient to have a way to generate the controlled gate  ${\tt ctrl} \mhyphen U$ from the access to an unknown, uncontrolled gate $U$.  The ability to generate controlled gates would also benefit  the implementation of standard quantum algorithms, providing them with an appealing modularity feature~\cite{Thompson2018}.  Besides quantum computation, the ability to   control an unknown quantum process would be beneficial to other information-processing tasks, such as quantum communication \cite{Chiribella_2019, kristjansson2020singleparticle, Kristjansson_2020}, quantum metrology \cite{zhao2020quantum, Frey19}, and quantum machine learning \cite{Lloyd_2014, briegel2012projective}.

The problem of the coherent control of an unknown channel  can be phrased   in the following way: `Is there a universal protocol which, from the use of a black-box channel $\cc$, implements its coherently controlled version?'.  It has been proven several times, in ever stronger ways \cite{Soeda2013,Araujo2014,Chiribella_2016,Bisio_2016,Thompson2018,gavorova2020topological}, that the answer to this question is a resounding `No': no quantum circuit can `controllise' arbitrary operations.   For general non-unitary channels, such a controllisation is not even unambiguously defined in the first place, as observed in Ref.  \cite{dong2020controlled}.

Yet, as has been noted at the same time, coherent control is actually easily implementable in various contexts, such as optical systems \cite{zhou2011AddingControl, zhou2013CalculatingEigenvalues, Friis2014, Dunjko_2015},  trapped ions \cite{Friis2014, Dunjko_2015}, and superconducting qubits \cite{Friis_2015}. 
These realisations are not in contradiction with the no-go theorems because the resources they use are not black boxes: in the computer science terminology, they are grey boxes, whose action is partially  known~\cite{Araujo2014, Thompson2018} (see also  Section \ref{sec:experimental} of this paper for a further elaboration of this point). 

This mismatch between  theory and  experiments suggests that it may be necessary to revisit the terms of the problem.  A suitable formulation of the problem would help understanding in which situations, from which resources, and with which protocols, one can implement a coherently controlled quantum channel. This understanding would allow to go beyond the existing examples of implementations of coherent control, and to compare their respective advantages. Another upshot of a better theoretical understanding is that  it allows to neatly distinguish the informational, implementation-independent aspects of coherent control from the specific, system-dependent  features of  experimental implementations.  In particular, it would help shift the focus away from optical implementations and towards a more implementation-neutral perspective. Finally, identifying the operational ingredients of coherent control helps  elucidate some aspects of the existing no-go theorems, as studying  protocols that can perform a certain  task usually helps understanding why other protocols cannot.

In this paper  we analyse the key features of  the experimental implementations,  and put forward a new formulation  of the problem of coherent control based on these features.  Our starting point is the observation that the crucial feature of the existing implementations is that  they use  \textit{sector-preserving channels}; i.e.,  channels whose   input systems can be partitioned into sectors (orthogonal subspaces), with the property that a state in a given sector always remains in this sector after  the channel has acted.   In this work, we  focus on the case where some sectors are one-dimensional and others are $d$-dimensional. \correction{A sector-preserving channel acting on a system with a 1-dimensional sector and a $d$-dimensional sector will be called a {\em sector-preserving channel of type $(1,d)$}.  More generally, a sector-preserving channel acting on a system with $m$ 1-dimensional sectors and $n$ $d$-dimensional sectors will be called a {\em sector-preserving channel of type  $(\underbrace{1,...,1}_{m~{\rm times}}, \underbrace{d,...,d}_{n~{\rm times}})$ }}.

The idea of regarding  sector-preserving channels as resources originates from Ref.\,\cite{Chiribella_2019}, and was further explored in Refs.\,\cite{kristjansson2020singleparticle, Kristjansson_2020}\footnote{In the past, a similar approach had independently been explored in Refs.\,\cite{aaberg2004subspace, aaberg2004operations}. A different approach, based on the unitary extension of quantum channels, was developed in  Refs.\,\cite{Abbott_2020, clement2020, branciard2021}.}. In these works, the focus was put on the use of sector-preserving channels  for communication\footnote{This was part of a wider discussion about the communication advantages of coherent control of causal order \cite{Kristjansson_2020, Chiribella_2021, Abbott_2020, Guerin_2019, Rubino_2021}.}. In contrast,  the relevance of sector-preserving channels to the task of coherent control has not been explored before, and will be the focus of this paper.

Our main results are summarised in the following.  For the standard notion of coherent control, we establish a perfect, one-to-one correspondence between sector-preserving channels of type $(1,d)$ and coherently controlled channels with target systems of dimension $d$.  We then show that this one-to-one correspondence can be implemented physically, by inserting  sector-preserving channels into a fixed, universal quantum circuit that generates  the corresponding controlled channels.   \correction{  Mathematically, this universal circuit can be represented as a {\em quantum supermap} \cite{chiribella2008supermaps, chiribella2009combs, chiribella2009switch}, that is, a transformation of quantum channels.  }  We call this particular supermap  the ${\tt CTRL}$ supermap, and show that it is invertible. Its inverse ${\tt CTRL}^{-1}$ is also a  supermap, corresponding to  a universal circuit that transforms controlled channels on $d$-dimensional systems into sector-preserving channels of type $(1,d)$.

Summarising,  coherently controlled channels on $d$-dimensional systems and sector-preserving channels of type $(1,d)$ are fully equivalent resources, and the interconversion of these resources is implemented by the ${\tt CTRL}$ supermap and by  its inverse.   It is worth contrasting this result  with the existing no-go theorems on coherent control: while control cannot be achieved from general channels on $d$-dimensional systems, it can be achieved from sector-preserving channels of type $(1,d)$. 
 \correction{The comparison is illustrated in Figure\,\ref{fig:2Theorems}.}

After establishing the above results, we extend them to more general versions of coherent control. For example,  we show a one-to-one correspondence between sector-preserving channels of type $(1,1,d)$ and compositely-controlled channels with two branches leading to application of the identity, and we build  universal circuits that implement this correspondence in both ways. \correction{  We then extend this result to compositely-controlled channels with any number of branches leading to application of the identity.  }  

We also extend our results to the coherent control of $N$ isometric channels, whose input and output spaces can be of different dimensions.  \correction{ As the initial resource, we take  $N$ sector-preserving isometric channels of type $(1 \to 1 ,d_\mr{in} \to d_\mr{out})$,  meaning that  {\em (i)} the input  (output) is partitioned into a 1-dimensional sector and a $d_{\rm in}$-dimensional  ($d_{\rm out}$-dimensional) sector, and {\em (ii)}   states in the $1$-dimensional input sector are mapped into states of in the 1-dimensional output sector, while states in the $d_{\rm in}$-dimensional input sector are mapped into states of in the $d_{\rm out}$-dimensional output sector.  We then show that this resource can be used to construct a  channel with coherent control between corresponding isometries.   }
 We study explicitly the $N=2$ case, which readily generalises to arbitrary $N$.   Mathematically, we show that there exists an invertible  supermap ${\tt 2 \mhyphen CTRL}$  that transforms \correction{every pair of} sector-preserving isometric channels into the corresponding controlled channel.   
 
\correction{  In the non-isometric case, however, we find that sector-preserving channels of type $(1 \to 1 ,d_\mr{in} \to d_\mr{out})$ are generally not sufficient to achieve all possible controlled channels.  Such channels can instead be realised using sector-preserving channels of type $(1 \to 1, d_\mr{in} \to d_\mr{out} d')$, where $d'$ is the dimension of an auxiliary system, used to extend the original channels (from a $d_{\rm in}$-dimensional system to a $d_{\rm out}$-dimensional system) to isometries. Using this extra resource, we provide a  universal protocol for the implementation of coherent control from $N$ sector-preserving channels. }

\correction{  We conclude the paper by building a general framework for the manipulation of sector-preserving channels, and, more generally, of channels that maps input sectors into output sectors according to a prescribed rule, called the {\em route} \cite{vanrietvelde2020routed}.    The key ingredient of our framework is  the notion of  `supermaps on routed channels,'  a new kind of supermaps whose input is restricted to channels with a prescribed route.  Examples of supermaps on routed channels are the  ${\tt CTRL}$ and ${\tt 2 \mhyphen CTRL}$ supermaps constructed earlier in the paper (or, more precisely, the restrictions of such maps to sector-preserving channels).} 

\correction{
Our results open the way to several  applications. First, by identifying the resources for the task of coherent control, we lay the basis for a resource-theoretic analysis of existing protocols and experiments.   
  Second,  the supermaps defined in this work can be easily extended to multiple channels, and to more elaborate architectures involving multiple instances of coherent control at different moments of time.    This flexibility can help the design  of complex protocols and algorithms, offering a built-in modularity feature.   Finally, the new notions of  composite control introduced in this paper have the potential to stimulate new theoretical protocols and experimental implementations with higher dimensional control systems.
}

The structure of the paper is as follows. In Section \ref{sec:CoherentControlCharacterisation}, we  review the existing definitions of controlled unitaries and channels, and we address their  extension to multiple channels, defining a new notion of compositely-controlled channels. In Section \ref{sec:sectorpreserv}, we analyze the structure of the existing implementations of coherent control, and use it to motivate a study of sector-preserving channels of type $(1,d)$.  We then show that these channels are in one-to one correspondence with controlled channels on a $d$-dimensional system.  In Section \ref{sec:ControlSupermap+Equivalence}, we show that the correspondence between sector-preserving channels of type $(1,d)$ and controlled channels  can be physically implemented by a universal protocol, formalised by the ${\tt CTRL}$ supermap.  In Section \ref{sec:CoherentControl2}, we generalise this correspondence to the coherent control between $N$ isometries, showing that it can also be realised via a universal protocol, and we discuss the case of the coherent control between $N$ general channels, showing that it requires more involved resources. In Section \ref{sec:compositelycontrolled}, we extend the results of the previous sections to \textit{compositely-controlled} channels. Finally, in Section \ref{sec:SupermapsRouted}  we  define supermaps on routed channels, providing a general framework for the manipulation of sector-preserving channels and more general channels that transform sectors in a prescribed way. 

\section{Coherently controlled quantum channels}  \label{sec:CoherentControlCharacterisation}

In this section, we review the existing definitions of coherently controlled unitaries and channels. Then, we provide  a one-to-one parametrisation of the possible controlled versions of a channel in terms of a   {\em `pinned Kraus operator'}.  Finally, we discuss more general types of controlled quantum channels, and we provide one-to-one parametrisations for these in terms of pinned Kraus operators.

\subsection{Controlled channels and pinned Kraus operators} \label{sec:ControlledChannelsAndPinnedKrausOps}
Let us start with the most basic definition  of controlled  operation:   controlled unitary gates. Given a unitary operator $U$ acting on a $d$-dimensional Hilbert space
$\ch_T$, there is a standard notion of a `controlled-$U$' channel: it is the channel corresponding to the  unitary operator 
\begin{align}\label{eq:ctrlU}
{\tt ctrl} \mhyphen   U   :=   |0\rangle \langle 0|  \otimes  I  +   |1\rangle \langle 1|    \otimes U\, ,
\end{align}
acting on a composite system, made of a two-dimensional {\em control} system $C$ and  of a $d$-dimensional {\em target} system $T$. 

More generally, one may want to control the evolution of an open system.  The general evolution of an open system $T$  is described by  a quantum channel $\cc$, that is, a completely positive, trace-preserving map mapping density matrices on $\ch_T$ into density matrices on $\ch_T$. The action of the channel $\cc$  on a generic density matrix $\rho$ can be conveniently described in the Kraus representation, as $\map C (\rho)  =  \sum_{i=1}^n   C_i  \rho  C_i^\dag$, where the operators $(C_i)_{i=1}^n$, called Kraus operators, satisfy the normalisation  condition \begin{align}\label{krausnorm} 
\sum_{i=1}^n   \,C_i^\dag C_i  = I \, ,
\end{align}
$I$ being the identity operator on $\spc H_T$.

Crucially, the Kraus representation of a channel is not unique:  if $V$ is a $l\times n$ isometry with matrix elements $V_{ji}$, the operators $(C'_j)_{j=1}^l$ defined by $C'_j  :  =  \sum_{i}  \, V_{ji}  \,  C'_i$ also form a Kraus representation of channel  $\map C$.  The non-uniqueness of the Kraus representation will play an important role in this paper.

For a general quantum channel $\map C$, the definition of coherent control is not straightforward. The naive generalisation of Eq. (\ref{eq:ctrlU}) would be to pick a Kraus representation $(C_i)$ and define the controlled operators ${\tt ctrl} \mhyphen   C_i   =  |0\rangle \langle 0| \otimes C_i +  |1\rangle \langle 1| \otimes   I $.  This definition, however, would fail to give a quantum channel, because the   above operators   fail to satisfy the normalisation condition~(\ref{krausnorm}). A suitable generalisation of Eq. (\ref{eq:ctrlU}) was put forward in Ref.\,\cite{dong2020controlled}: a controlled version   of channel  $\cc$ is the channel with Kraus operators
\begin{align}\label{eq:ctrlKraus}
{\tt ctrl}_{\alpha_i} \mhyphen   C_i   :=   \alpha_i\,    |0\rangle \langle 0|  \otimes I +  |1\rangle \langle 1|  \otimes C_i \, ,
\end{align}
where $(\alpha_i)_{i=1}^n$ are complex amplitudes satisfying the normalisation condition $\sum_{i=1}^n  \, |\alpha_i|^2  =1$. 

This definition is a special case of the definition of coherent control of two general channels considered in  Refs. \cite{oi2003interference,aaberg2004subspace,Chiribella_2019,Abbott_2020}, in the special case where one of the two channels is the identity channel. 

 It is important to observe that the definition of the controlled channel does not depend only on the channel $\cc$.  In general, it can depend both on the set  of Kraus operators $  {\bf C}  :  =(C_i)_{i=1}^n$ and on the set  of amplitudes   $\boldsymbol{\alpha} := (\alpha_i)_{i=1}^n$ used in Eq.~(\ref{eq:ctrlKraus}).  To emphasise the dependence on the Kraus operators $\bf C$ and on the  amplitudes   $\boldsymbol\alpha$, we will denote the controlled channel by     ${\tt ctrl}_{  \boldsymbol \alpha }^{\bf C} \mhyphen \map C$.

Different choices of Kraus operators and amplitudes generally give rise to different versions of  controlled channels, with none of these versions being straightforwardly more natural than the other (although some may be more or less coherent\,\cite{dong2020controlled}).  Given that the definition of controlled channels is non-unique, an important question  is how  to parametrise the possible controlled channels in a compact way.  As it turns out, the parametrisation  ${\tt ctrl}_{  \boldsymbol \alpha }^{\bf C} \mhyphen \map C$ is quite redundant:  in fact, many choices of $\bf C$ and of $\boldsymbol \alpha$ give rise to the same controlled channel   ${\tt ctrl}_{  \boldsymbol \alpha }^{\bf C} \mhyphen \map C$.

In the following, we provide  a simple one-to-one parametrisation of the possible controlled channels corresponding to a given uncontrolled channel $\cc$: the controlled channels are  in one-to-one correspondence with pairs of the form $(\cc,  C_1)$, where  $C_1$ is a fixed  Kraus operator of $\cc$.  We call the pair $(\cc ,  C_1)$ a {\em channel with a pinned Kraus operator}. 

First, we prove that any controlled version of $\cc$ has  a Kraus representation in which  only the first Kraus operator is coherent with the identity:
\medskip 

\begin{lemma} \label{lem:characterisation1}
For every controlled channel   $ {\tt ctrl}_{  \boldsymbol \alpha }^{\bf C} \mhyphen \map C  $,  one can find a Kraus representation    in which one Kraus operator is of the form  $  |0\rangle \langle 0| \otimes  I +   |1\rangle \langle 1|  \otimes C_1'$ and all the others are of the form $ |1\rangle \langle 1| \otimes  C_j'$, where ${\bf C}' : = (C'_j)_{j=1}^{n}$ is a suitable Kraus representation of channel $\cc$.  
In other words, one has 
\be \label{eq:characterisation1}
  {\tt ctrl}_{  \boldsymbol \alpha }^{\bf C} \mhyphen \map C   =   {\tt ctrl}_{  {\bf  u_n}  }^{{\bf C}'} \mhyphen \map C \, ,
\ee
where ${\bf u}_{n}$ is the $n$-dimensional column vector with a $1$ in the first entry, and 0 in the remaining $n-1$ entries. 
\end{lemma}
\begin{proof}
As $\boldsymbol \alpha$ is a normalised vector in ${\mathbb C}^n$, one can find a unitary matrix $V$ sending it to the basis vector ${\bf u}_n$, i.e.   $  V \boldsymbol \alpha   =   {\bf u}_n$. Then, the Kraus operators $(C_j')_{j=1}^n$ defined by $C_j':= \sum_j V_{ji} C_i$  form an alternative Kraus representation of $\cc$, and the Kraus operators $(K_j)_{j=1}^n$ defined by $K_j:= \sum_j V_{ji} ( {\tt ctrl}_{\alpha_i} \mhyphen   C_i)$ form an alternative Kraus representation of ${\tt ctrl}_{  \boldsymbol \alpha }^{\bf C} \mhyphen \map C$. It is straightforward to see that $ K_1 =  {\tt ctrl}_{1} \mhyphen   C_1'$ and $K_j  ={\tt ctrl}_{0} \mhyphen C-j'$ for every $j>1$. Hence,  ${\tt ctrl}_{  \boldsymbol \alpha }^{\bf C} \mhyphen \map C$ can be characterised as in (\ref{eq:characterisation1}).
\end{proof}

This result removes the freedom in the choice of the amplitudes $(\alpha_i)_{i=1}^n$:  one can simply  set the first amplitude to 1, and all the other amplitudes to zero. All the variability of the controlled channels is then included in the choice of Kraus representation for channel $\cc$.

We now show a further simplification:  the definition of the controlled channel depends only on the choice of the {\em first} Kraus operator in a Kraus representation of $\cc$.  In other words, the choice of the other Kraus operators does not affect the type of control one obtains.

\medskip 

\begin{lemma}\label{lem:characterisation2} Let ${\bf C} := (C_i)_{i=1}^m$ and ${\bf C}': = (C'_j)_{j=1}^{n}$  be two Kraus representations for channel $\cc$. Then, the controlled channels     ${\tt ctrl}_{ {\bf u}_m }^{\bf C} \mhyphen \map C$ and ${\tt ctrl}_{ \bf{u}_n }^{{\bf C}'} \mhyphen \map C$ coincide if and only if the operators $C_1$ and $C_1'$ coincide.  In formula,   

\begin{align} 
\label{eq:characterisation2}
      {\tt ctrl}_{  {\bf u}_m  }^{\bf C} \mhyphen \map C = {\tt ctrl}_{  {\bf  u}_n}^{{\bf C}'} \mhyphen \map C \qquad  \iff  \qquad  \, C_1 = C'_1 \,. 
\end{align}
\end{lemma}

\begin{proof}
We start with the direct implication. Without loss of generality, we take $m=n$, as one can always include zero Kraus operators and match the cardinality of the  Kraus representations of   ${\tt ctrl}_{  {\bf u}_m  }^{\bf C} \mhyphen \map C$ and ${\tt ctrl}_{  {\bf  u}_n}^{{\bf C}'} \mhyphen \map C $.  If the two controlled channels coincide,  then there exists a unitary matrix  $W$ that connects their Kraus representations. In particular, one must have  
\begin{align}\label{aaaaa}
C_1 \otimes |1\rangle\langle 1|  +  I\otimes |0\rangle\langle0|    =  W_{11}  \,\left( \, C_1' \otimes |1\rangle\langle 1|  +  I\otimes |0\rangle\langle0| \,\right)  +   \sum_{j>1}   \, W_{1j}\,  C_j'  \otimes |1\rangle \langle 1| \, . 
\end{align}
Taking the expectation value on the vector $|0\rangle$ on both sides of the equation, we then obtain the relation $I  =  W_{11} \,  I$, which implies $W_{11}=1$, and, since $W$ is a unitary matrix, $W_{1j} = 0$ for every $j>1$.  Inserting this condition in Eq. (\ref{aaaaa}) we obtain $C_1 =  C_1'$.

For the converse implication, suppose that $ C_1 = C'_1$.  Then, for an arbitrary product state $\rho_{C}\otimes \rho_{T}$ of the control and the target,  we have   \begin{align}
 \nonumber  {\tt ctrl}_{  {\bf u}_m  }^{\bf C} \mhyphen \map C  (\rho_C\otimes \rho_T) &   =  {\tt ctrl}_1\mhyphen C_1   ~ (\rho_C\otimes \rho_T) ~  \left(  {\tt ctrl}_1\mhyphen C_1\right)^\dag    +  \sum_{i>1} {\tt ctrl}_0\mhyphen C_i   ~ (\rho_C\otimes \rho_T)~  \left( {\tt ctrl}_0\mhyphen C_i\right)^\dag \\  
 \nonumber &   =  {\tt ctrl}_1\mhyphen C_1    ~(\rho_C\otimes \rho_T) ~ \left({\tt ctrl}_1\mhyphen C_1\right)^\dag    +    |0\rangle \langle 0| \rho_C |0\rangle\langle 0|  \otimes  \left( {  \map C }   (\rho_T)   -   C_1 \rho_T  C_1^\dag\right)   \\ 
 \nonumber &   =  {\tt ctrl}_1\mhyphen C'_1    ~(\rho_C\otimes \rho_T) ~ \left({\tt ctrl}_1\mhyphen C_1^{\prime}\right)^\dag    +     |0\rangle \langle 0| \otimes \rho_C |0\rangle\langle 0|   \left( {  \map C }   (\rho_T)   -   C'_1 \rho_T  C_1^{\prime\, \dag}\right)  \\
 \nonumber   &   =  {\tt ctrl}_1\mhyphen C'_1    ~(\rho_C\otimes \rho_T)~  \left({\tt ctrl}_1\mhyphen C_1^{\prime}  \right)^\dag    +  \sum_{j>1} {\tt ctrl}_0\mhyphen C'_j   ~ (\rho_C\otimes \rho_T)~  
\left(  {\tt ctrl}_0\mhyphen C_j^{\prime}  \right)^\dag \\
 &={\tt ctrl}_{  {\bf u}_n  }^{{\bf C}'} \mhyphen \map C  (\rho_C\otimes \rho_T) \, .
    \end{align}
    Since $\rho_C$ and $\rho_T$ are arbitrary, we conclude $    {\tt ctrl}_{  {\bf u}_m  }^{\bf C} \mhyphen \map C = {\tt ctrl}_{  {\bf  u}_n}^{{\bf C}'} \mhyphen \map C $.  
\end{proof}

Combining Lemmas~\ref{lem:characterisation1} and~\ref{lem:characterisation2}, we obtain a non-redundant  parametrisation of the possible controlled versions of a given channel: 
\medskip 

\begin{theorem}\label{th:characterisation}
The  controlled versions of  channel $\cc$, as defined by  Eq.~(\ref{eq:ctrlKraus}), are in one-to-one correspondence with the possible choices of a single Kraus operator for channel $\cc$. 
\end{theorem}

By `a Kraus operator for channel $\cc$', we mean a Kraus operator appearing in at least one Kraus representation for  $\cc$. Equivalently, the possible Kraus operators for a given channel can be characterised as follows: 
\medskip
\begin{lemma} An operator $C_1$ is a Kraus operator for channel $\cc$ if and only if  the map $\cc_-  :   \rho  \mapsto \cc  (\rho) -   C_1 \rho C_1^\dag$ is completely positive.  
\end{lemma}
\begin{proof} The `only if' part is immediate. For the `if' part,  a Kraus representation for $\cc$ containing the operator $C_1$ can be built by picking an arbitrary Kraus representation for the map $\cc_-$, say $(C_i)_{i=2}^n$. For any such choice, the operators $(C_i)_{i=1}^n$ form a Kraus representation for channel $\cc$. 
\end{proof}

\medskip

Hereafter, we will call the single  Kraus operator picked in Theorem \ref{th:characterisation}   a {\em pinned Kraus operator}. A channel with a pinned Kraus operator will be represented by the pair $({\cal C},  C_1)$.   Given a pinned Kraus operator $C_1$, and an arbitrary completion of it into a Kraus representation $(C_i)_i$, the corresponding controlled version of $\cc$ is given by the Kraus operators
\be \begin{cases}
     \widehat{C}_1 = \ketbra{0}{0} \otimes I + \ketbra{1}{1} \otimes  C_1  \\
     \widehat{C}_i = \ketbra{1}{1} \otimes  C_i \quad\quad \forall i \geq 2  \, .
\end{cases}\ee
From now on, we will use the notation  ${\tt ctrl}_{C_1} \mhyphen \cc$ to denote the controlled channel with the above Kraus operators.  The action of the controlled channel ${\tt ctrl}_{C_1} \mhyphen \cc$  on a generic product state of the target system and of the control is 
\begin{align}
\nonumber  {\tt ctrl}_{C_1} \mhyphen \cc  (\rho_C \otimes \rho_T)   &=  \sum_i    \,   \widehat{C}_i   \,  (\rho_C \otimes \rho_T)\,   \widehat{C}_i^\dag  \\
 \nonumber &   = \langle 0  |  \rho_C  |0\rangle ~  \ketbra{0}{0}_C \otimes  ~    \rho_T     \\
   \nonumber & \quad    +      \langle 1  |  \rho_C  |1\rangle ~  \ketbra{1}{1}_C \otimes  ~    \cc(\rho_T)     \\
 &  \quad  +   \langle 1  |  \rho_C  |0\rangle ~  \ketbra{1}{0}_C \otimes  ~   C_1 \, \rho_T   +  {\rm h.c.}\, , 
  \end{align}
  where h.c. denotes the Hermitian conjugate.    In the above formula,  the first two terms in the sum represent the classical control on the channel,  while  the second two terms represent the `coherent part' of the controlled operation.  

This pinned Kraus operator $C_1$ coincides with  the `transformation matrix' of Ref.\,\cite{Abbott_2020},  the `vacuum interference operator' of Ref.\,\cite{Chiribella_2019}, and the `$K$ operator' of Ref.\,\cite{dong2020controlled}. Ref.\,\cite{Abbott_2020} derived  the `transformation matrix' from a Stinespring dilation of the channel $\cc$, and interpreted it as the additional information that has to be provided about the physical implementation of  channel $\cc$ in order to build a controlled channel.  In contrast,  Ref.\,\cite{Chiribella_2019} derived the `vacuum interference operator' from an extension of channel $\cc$ to a larger channel that can act also on the vacuum. In this paper, we will make connection with the latter approach, showing that the controlled channel ${\tt ctrl}_{C_1} \mhyphen \cc$ is in one-to-one correspondence, both mathematically and physically, with a particular extension of the original channel $\cc$, corresponding to the vacuum extension of Ref.\,\cite{Chiribella_2019}.

Compared to Refs.\,\cite{Abbott_2020,Chiribella_2019,dong2020controlled}, our presentation  makes it evident that  the operator characterising  a controlled version of channel $\cc$ can be simply understood as a Kraus operator of this channel, a fact that has not been pointed out  before\footnote{A  proof in Ref.\,\cite{dong2020controlled} mentioned  that any possible `$K$ operator' is a Kraus operator of $\cc$, without however discussing the reverse implication. }.  In addition, the explicit relation between control and pinned Kraus operators suggests further extensions of the notion of quantum control, as discussed in the next  subsection.

\subsection{Control between multiple noisy channels} \label{sec:2ControlDef}
We now consider a generalisation of the notion of coherent control: the case in which each of the two values of the control is associated to the execution of a different channel on the target system. In other words, we now consider the coherent control between the execution of two channels $\ca$ and $\cb$, rather than between one channel and the identity channel. We will now take the input and output target systems, $T_\mr{in}$ and $T_\mr{out}$, to be of possibly different dimensions.

\correction{Before entering into the technical details, it may be helpful to note that different authors have used different names for what is essentially the same notion: Refs. \cite{aharonov1990superpositions}, \cite{oi2003interference}, \cite{aaberg2004subspace}, \cite{Abbott_2020}, \cite{Chiribella_2019} use the expressions `superposition of time evolutions', `interference of CP maps', `gluing of CP maps', `coherent control of quantum channels', and `superposition of quantum channels', respectively. We review the existing terminologies in Appendix \ref{app:terminology}.} 

If we start with the basic case of two isometric gates, represented by two isometries $U, V : \ch_{T_\mr{in}} \to \ch_{T_\mr{out}}$, the standard notion of a `controlled-$(U,V)$' channel is given by the isometry 

\begin{align}\label{eq:ctrlUV}
{\tt ctrl} \mhyphen   (U,V)   :=   |0\rangle \langle 0|  \otimes  U  +  |1\rangle \langle 1|    \otimes  V\, .
\end{align}

Extending this definition to the case of the control between two noisy evolutions, represented by CPTP maps $\ca, \cb: \cl(\ch_{T_\mr{in}}) \to \cl(\ch_{T_\mr{out}})$, requires more work. Once again, there are a variety of ways of defining the controlled version of $\ca$ and $\cb$. These different versions can be obtained by picking Kraus representations of same length\footnote{Note that any two Kraus representations can be taken to be of the same length by adjoining 0's to the shortest one.} $(A_i)_{i=1}^n$ and $(B_i)_{i=1}^n$ for $\ca$ and $\cb$ and defining the Kraus operators:

\begin{align}\label{eq:ctrl2Kraus}
{\tt ctrl} \mhyphen   (A_i,B_i)   :=  |0\rangle \langle 0| \otimes A_i  +  |1\rangle \langle 1|  \otimes B_i \, .
\end{align}

A one-to-one parametrisation of the possible choices is provided in the  following theorem, proven in Appendix \ref{app:2ControlParam}:
\medskip
\begin{theorem}\label{th:2ControlParam}
Given a Kraus representation $(A_i)_{i=1}^n$ of minimal length of $\ca$, the choice of a control between $\ca$ and $\cb$ is in one-to-one correspondence with the choice of $n$ Kraus operators of $\cb$.
\end{theorem}
\medskip

By `$n$ Kraus operators of $\cb$', we mean $n$ operators that appear together in at least one Kraus representation of $\cb$. Calling these operators $B_i$'s, and arbitrarily completing them into a Kraus representation $(B_i)_{i=1}^{n'}$ of $\cb$, Kraus operators for the corresponding controlled channel are given by the concatenation of the $({\tt ctrl} \mhyphen   (A_i,B_i))_{i=1}^n$ and the $({\tt ctrl} \mhyphen   (0,B_i))_{n < i \leq n'}$. Note that in this parametrisation, only the Kraus operators of $\cb$ vary; those of $\ca$ are fixed from the start.

The previous considerations can be extended to the case of a control system of dimension $N$, controlling between the execution of $N$ channels $\cc^1,\dots,\cc^N$. A strategy would be to proceed via recursion, first picking a control between $\cc^1$ and $\cc^2$, then picking a control between this controlled channel and $\cc^3$, etc.

\section{A new resource for coherent control: sector-preserving channels} \label{sec:sectorpreserv}

Here we discuss the physical resources needed to implement coherent control of general quantum channels.  

\subsection{A no-go theorem for coherent  control of  unitary gates, and a way to evade it}  
It has been proven in various ways that it is impossible to construct a controlled unitary gate starting from a black box that implements the corresponding uncontrolled unitary gate  \cite{Araujo2014, Soeda2013, Thompson2018, Chiribella_2016,  Bisio_2016, gavorova2020topological}.    Mathematically,  the no-go theorem is that it is impossible to find a quantum supermap that transforms  a generic unitary channel  $\map U: \rho \mapsto  U \rho U^\dag$ into the controlled unitary channel ${\tt ctrl} \mhyphen \map U:  \rho \mapsto    {\tt ctrl} \mhyphen U\, \rho  \,  {\tt ctrl} \mhyphen U^\dag$ with the  operator ${\tt ctrl} \mhyphen U$  defined in Eq.~(\ref{eq:ctrlU}).    

The origin of the impossibility is that the uncontrolled unitary channel  $\map U$  is provided as a {\em black box},  without any further information on its action except for the fact that $\map U$ is known to be unitary.  One way to evade the no-go theorem is to start from a device that is not a complete black box, but rather a {\em grey box}, whose action is partially known. For example, one could be given a device that implements a unitary gate  $\widetilde U  = \ketbra{ \phi_0 }{\phi_0}  \oplus U  $,  where $\tilde{U}$ acts on $\ch$ and $U$ is an unknown unitary gate acting on a $d$-dimensional sector (i.e. orthogonal subspace) $\ch^1 \subseteq \ch$, and $\ket{\phi_0}$ is another state, orthogonal to all the states in $\ch^1$. In this case, the action of the device in the sector $\ch^1$ is  unknown, while the action of the device on the vector $\ket{\phi_0}$ is known.      In this setting, the controlled gate   ${\tt ctrl} \mhyphen U$  can be built from the gate $\widetilde U$ using a simple quantum circuit \cite{Araujo2014, Friis2014, Thompson2018}.  

The use of grey boxes that act in a known way on some input states   is central to all existing proposals for experimental implementations of coherent controls of unitary gates.    For example, photonic implementations \cite{zhou2011AddingControl,zhou2013CalculatingEigenvalues} achieve coherent control of certain optical devices, such as polarisation rotators, by exploiting the fact that such devices are passive, and therefore transform the vacuum state into itself.  In these examples, the sector $\ch^1$ is spanned by single-photon polarisation states, and the state $\ket{\phi_0}$ is the zero-photon Fock state.    

In  trapped-ions implementations \cite{Friis2014, Dunjko_2015}, the input device uses  a laser pulse to implement a unitary gate  by stimulating the transition between the two electronic levels.  The pulse is far off resonance with the transition between the other  electronic levels of the ion, and therefore the device acts trivially on such levels. In this case, the state $\ket{\phi_0}$ can be any of the levels that are unaffected by the pulse.    A similar  situation arises in superconducting-qubits implementations \cite{Friis_2015}.

In summary, all the existing proposals of experimental implementations use grey box unitary gates $\widetilde U$ that  act
\begin{enumerate}
\item as unknown gates $U$ on a sector $\ch^1  \simeq  \ch_T$, and  
\item  as the identity gate $I$ on another sector $\ch^0$, orthogonal to $\ch^1$.   
\end{enumerate}   
In the following we will extend this scheme from unitary gates to arbitrary noisy channels, and to the case of gates acting as the identity on several sectors, showing that access to a suitable grey box channel allows one to build a controlled channel that is in one-to-one correspondence with it. 

We will restrict ourselves to the case in which the sectors on which the identity is applied are one-dimensional; however, all our arguments could be extended to the case in which they are multi-dimensional and the grey boxes act as the identity on each of them. Note that when the \correction{extension} sectors have the same dimension as $\ch^0$, the above requirements lead to the usual definition of controlled channels.

\subsection{Modelling noisy grey boxes: sector-preserving channels} \label{sec:experimental}

We now consider how the grey box approach of the previous section can be extended from unitary gates to arbitrary noisy channels. 

To this purpose, we consider a noisy quantum channel $\widetilde \cc$ that  acts on a system $S$ with a Hilbert space $\ch_S$ partitioned into two sectors, $\ch_S = \ch_S^0 \oplus \ch_S^1$, with $\ch_S^0$ one-dimensional and  $\ch_S^1 \simeq \ch_T$. The channel $\widetilde \cc$ will act 
 \begin{enumerate}
\item as a completely unknown channel  $\map C : \cl (\ch_S^1) \to \cl (\ch_S^1) $ on the input states in $\cl(\ch_S^1) $, and  
\item as the identity channel $I$  on the unique input state in $\cl(\ch_S^0)$.    
\end{enumerate}   

Such grey boxes have a simple characterisation: they are the channels that preserve the sectors $\ch_S^m$, thus called sector-preserving channels\footnote{\correction{We note that the notion of sector-preservingness has been independently introduced in the past, under the name `subspace-preservingness'; see Ref.\,\cite{aaberg2004subspace}.}}.

\medskip
\begin{definition} \label{def:sectorpres}
Let $\ch_S = \bigoplus_{k=0}^m \ch_S^k$ be a Hilbert space with a preferred partition into sectors.
A channel $ \widetilde \cc :  \cl(\ch_S)  \to \cl  (\ch_S)$   is  {\em sector-preserving} if it preserves the set of states with support in the subspace $\ch_S^k$, for every $k \in  \{0,\dots, m\}$. In formula, 
\be\label{secpres} 
   \forall k, \, \forall \rho \in \cl(\ch_S^k), \quad \cc(\rho) \in \cl(\ch_S^k)  \, ,
\ee
Note that $\rho \in \cl(\ch_S^k)$ equivalently means that ${\sf Supp}  (\rho) \subseteq \ch_S^k$, where ${\sf Supp}  (\rho)$ denotes the support of $\rho$.    
\end{definition}

\medskip
Sector-preserving channels can be seen as \correction{a} special case of the notion of channels `following a route' (i.e., satisfying given sectorial constraints), introduced in Ref.\,\cite{vanrietvelde2020routed}: namely, they are the channels that follow the identity route $\delta \times \delta$.
The condition (\ref{secpres}) was  called  the `no-leakage condition' in   Ref.\,\cite{Chiribella_2019}. 

When some of the sectors  $\ch_S^k$ are one-dimensional,   the condition of sector preservation  (\ref{secpres}) implies that the channel $\widetilde \cc$ acts as the identity channel on each of them. In the following, we will denote the sector preserving channels with $\dim (\ch_S^k)  =1 \, \forall k < m$   and $\dim (\ch_S^m)  =d $  as {\em sector-preserving channels of type $(\underbrace{1,\dots, 1}_{m~{\rm times}}, d )$}. In particular, the channels we asked for in this Section are the sector-preserving channels of type $(1,d)$.
 
 The approach of considering an extended channel that acts as $\map C$  on a given sector was introduced  in  Ref.\,\cite{Chiribella_2019}. There, there was only one one-dimensional sector, which was called the `vacuum sector', and the  channel $\widetilde \cc$ was called a `vacuum extension', with this terminology  motivated by the photonic implementations.   Here, however, we prefer to use the expressions `\correction{extension} sectors' and `extended channel', which are  neutral with respect to the choice of experimental implementations.   
 
The key point of our paper is that the grey box channel $\widetilde \cc$, and not the black box channel $\map C$, should be regarded as the initial  resource for the implementation of coherent control. In other words, we argue that one should  shift the terms of the problem  away from the question `what can one do with an unknown channel $\cc$?'.  Instead, one should ask  the question `what can one do with a channel $\widetilde{\cc}$ that acts as an unknown channel on a given sector?'.   

A similar shift of perspective was proposed in Refs.\,\cite{Chiribella_2019, Kristjansson_2020, kristjansson2020singleparticle} for the purpose of defining quantum communication protocols where messages can travel in a coherent superposition of multiple trajectories.    In this context, extended channels were used to describe  communication devices that can take as input either one particle (corresponding, in our notations, to the sector $\ch_S^1$) or the vacuum (corresponding to the sector $\ch_S^0$).     This modelling was essential to define  resource theories of quantum communication\,\cite{Kristjansson_2020}, where the initial resources are communication devices that can be connected in a coherent superposition of multiple configurations.  Our paper can be viewed as an application of the same approach to the task of the coherent control of quantum channels: the extended channel represents the initial resource,  and the question is which types of controlled channel can be constructed from such resource.

\subsection{The case of one \correction{extension} sector} \label{sec:OneSuppSector}

The case  where there is only one \correction{extension} sector $\ch_S^1$ (i.e., of sector-preserving channels of type $(1,d)$) is particularly relevant in this paper, because, as we will show later, it provides the fundamental resource for the realisation of the controlled channels defined in Eq. (\ref{eq:ctrlKraus}).  

In terms of Kraus representation, the sector-preserving  channels of  type $(1,d)$ can be characterised  as the channels with Kraus operators of the form 
\be \label{eq:SectorPreservingGeneralForm} \widetilde C_i = \alpha_i \oplus\,  C_i ,\ee
where $(C_i)_i$ is a Kraus representation of some channel acting on  sector $\ch_S^1 \simeq \ch_T$, and the $\alpha_i$'s are amplitudes satisfying the normalisation condition $\sum_i  |\alpha_i|^2 = 1$.  
For a proof of the above equation, see  Lemma 1 in Ref.\,\cite{Chiribella_2019} (this can also be seen as a consequence of the more general Theorem 6 in Ref.\,\cite{vanrietvelde2020routed}). 

A one-to-one parametrisation of the sector-preserving  channels of type $(1,d)$  can be obtained with the same approach as in Section \ref{sec:CoherentControlCharacterisation}.  
 
 \medskip
\begin{lemma} \label{lem:characterisationSectorPreserving1}
Every sector-preserving channel of type $(1,d)$ has a  Kraus representation of the form
\be \label{eq:characterisationSectorPreserving1}
\begin{cases}
  \widetilde C_1 = 1 \oplus C_1     \\
  \widetilde C_i = 0 \oplus C_i   \quad\quad \forall i \geq 2 \, ,
\end{cases}
\ee
where $(C_i)_i$ is a Kraus representation of some channel on the $d$-dimensional sector. 
\end{lemma}

\begin{proof}
As in the proof of Lemma \ref{lem:characterisation1}, this alternative Kraus representation can be found by using a unitary matrix $(V_{ji})_{ji}$ that sends the normalised vector $(\alpha_i)_i$ to $(1, 0, \dots, 0)$.
\end{proof}

Using the same arguments as in Section \ref{sec:CoherentControlCharacterisation}, it is easy to see that the sector-preserving channels $\widetilde {\cc}$ are in one-to-one correspondence with pairs $(  \map C,   C_1)$, consisting of a channel acting on sector $\cl(\ch_S^1)$, and of a Kraus operator for $\cc$.  In short, we have the following.

\medskip 

\begin{theorem}\label{th:characterisationSectorPreserving}
The sector-preserving channels of type $(1,d)$ are in one-to-one correspondence with channels with a pinned Kraus operator on their $d$-dimensional sector.  
\end{theorem}
\medskip

The sector-preserving channel of type $(1,d)$ that corresponds to the channel $\cc$ with the pinned Kraus operator $C_1$ on its $d$-dimensional sector shall be called $\widetilde{\cc}[C_1]$. In the case of unitary channels, the characterisation is particularly simple.

\medskip 

\begin{corollary}\label{cor:UnitaryCorrespondence}
Sector-preserving unitary channels of type $(1,d)$  are in one-to-one correspondence with unitary operators  in dimension $d$. 
Explicitly, the correspondence between  sector-preserving unitary channels $\widetilde {\cu}$ and unitary operators $U$ is given by the relation
\be  \quad \widetilde{\cu}  (\rho) = (1 \oplus U) \, \, \rho \, \, (1  \oplus U)^\dag  \qquad \forall \rho  \in \cl (\ch_S) \, . \ee
\end{corollary}
\medskip 

This is in contrast with the general situation for unitary channels, which correspond to unitary operators only up to an arbitrary global phase. The crucial fact here is that the one-dimensional \correction{extension} sector can be used to fix this phase gauge in the $d$-dimensional sector.

Going back to the  case of general channels,   Theorem \ref{th:characterisationSectorPreserving} establishes a one-to-one correspondence between sector-preserving channels of type $(1,d)$  and controlled channels:  

\begin{corollary}\label{cor:correspondence(1,d)}
For any $d$, the following sets are in one-to-one correspondence:  
\begin{enumerate}
\item  controlled channels as defined in (\ref{eq:ctrlKraus}), with a $d$-dimensional target system; 
\item  sector-preserving channels of type $(1,d)$;
\item channels with a pinned Kraus operator in dimension $d$.
\end{enumerate}
\end{corollary}

\correction{Let us comment on the respective roles, for our purposes, of the three notions which Corollary \ref{cor:correspondence(1,d)} shows to be mathematically equivalent. The first (controlled channels) is essentially an informational notion, with practical use in quantum protocols: this is typically what one wants to eventually obtain. The second (sector-preserving channels of type $(1,d)$) can be understood as the physical resource (with the sector-preserving property often corresponding to physical features of an interaction, such as conservation laws) allowing to implement the first one. Finally, the third (channels with a pinned Kraus operator) is a purely mathematical notion, with no direct practical interpretation, which serves to provide a simple one-to-one mathematical parametrisation to the first two.}

In fact, a more careful inspection also reveals that the one-to-one correspondence between the above sets can be implemented by linear maps.   For the sets of controlled channels and sector-preserving channels, the correspondence can be implemented physically by quantum circuits that convert sector-preserving channels into controlled channels, and vice-versa.  This physical correspondence is the object of the next section.

\section{The control supermap and the equivalence between sector-preserving and controlled channels} \label{sec:ControlSupermap+Equivalence}
\subsection{The control supermap} \label{sec:ControlSupermapOne}

In the previous section, we showed that the controlled channels on target systems of dimension $d$ (the ${\tt ctrl}_{C_1} \mhyphen \cc$) are in one-to-one correspondence with sector-preserving channels of type $(1,d)$ (the $\widetilde{\cc}[C_1]$).

Our point is now to show that for any given $d$, there is a universal circuit architecture in which an agent who possesses the sector-preserving channel $\widetilde{\cc}[C_1]$ can insert this channel in order to implement the controlled channel ${\tt ctrl}_{C_1} \mhyphen \cc$.

\begin{figure*}
    \centering
    $%
\begin{tikzpicture}
	\begin{pgfonlayer}{nodelayer}
		\node [style=none] (45) at (-1.75, 3) {};
		\node [style=none] (46) at (-1.75, -3) {};
		\node [style=none] (47) at (3, -3) {};
		\node [style=none] (48) at (3, 3) {};
		\node [style=none] (49) at (0.75, 1.75) {};
		\node [style=none] (50) at (0.75, -1.75) {};
		\node [style=none] (51) at (2, 1.75) {};
		\node [style=none] (52) at (2, -0.75) {};
		\node [style=none] (53) at (2, -1.75) {};
		\node [style=none] (54) at (3, 1.75) {};
		\node [style=none] (55) at (3, -1.75) {};
		\node [style=none] (56) at (2, 0.75) {};
		\node [style=none] (57) at (-0.5, -3) {};
		\node [style=none] (58) at (-0.5, -4) {};
		\node [style=right label] (59) at (-0.5, -3.625) {$C$};
		\node [style=none] (60) at (2, 3) {};
		\node [style=none] (61) at (2, 4) {};
		\node [style=right label] (62) at (2, 3.625) {$T$};
		\node [style=right label] (63) at (2, 1) {$S$};
		\node [style=right label] (64) at (2, -1.25) {$S$};
		\node [style=none] (65) at (-0.5, 3) {};
		\node [style=none] (66) at (-0.5, 4) {};
		\node [style=right label] (67) at (-0.5, 3.625) {$C$};
		\node [style=none] (68) at (2, -3) {};
		\node [style=none] (69) at (2, -4) {};
		\node [style=right label] (70) at (2, -3.625) {$T$};
		\node [style=none] (71) at (-0.5, 0) {$\tt{CTRL}$};
	\end{pgfonlayer}
	\begin{pgfonlayer}{edgelayer}
		\draw (51.center) to (56.center);
		\draw (52.center) to (53.center);
		\draw (57.center) to (58.center);
		\draw (60.center) to (61.center);
		\draw (65.center) to (66.center);
		\draw (68.center) to (69.center);
		\draw (45.center)
			 to (46.center)
			 to (47.center)
			 to (55.center)
			 to (50.center)
			 to (49.center)
			 to (54.center)
			 to (48.center)
			 to cycle;
	\end{pgfonlayer}
\end{tikzpicture}} \quad\quad := \quad %
\begin{tikzpicture}
	\begin{pgfonlayer}{nodelayer}
		\node [style=none] (2) at (-2, 6.5) {};
		\node [style=none] (3) at (-2, -5.5) {};
		\node [style=right label] (4) at (-2, -5.375) {$C$};
		\node [style=right label] (5) at (0, 6.125) {$T$};
		\node [style=right label] (6) at (2, 1) {$S$};
		\node [style=right label] (7) at (2, -1) {$S$};
		\node [style=right label] (8) at (-2, 6.125) {$C$};
		\node [style=none] (9) at (0, 6.5) {};
		\node [style=none] (10) at (0, -5.5) {};
		\node [style=right label] (11) at (0, -5.375) {$T$};
		\node [style=very very large map] (12) at (0.5, 4.75) {$\cd$};
		\node [style=none] (13) at (2, 4.75) {};
		\node [style=none] (14) at (2, -4) {};
		\node [style=right label] (15) at (-2, 0) {$C$};
		\node [style=none] (17) at (0, 6.5) {};
		\node [style=right label] (18) at (0, 0) {$S$};
		\node [style=none] (20) at (0, 6.5) {};
		\node [style=none] (21) at (0, -5.5) {};
		\node [style=white dot] (22) at (-2, -2.5) {};
		\node [style=large map] (23) at (1, -2.5) {$\mr{SWAP}$};
		\node [style=none] (24) at (2.075, -4.35) {$s^0$};
		\node [style=none] (25) at (2, -5) {};
		\node [style=none] (26) at (1.25, -4) {};
		\node [style=none] (27) at (2.75, -4) {};
		\node [style=none] (28) at (2, -4) {};
		\node [style=map] (29) at (0, -4.25) {$\cv$};
		\node [style=right label] (30) at (2, -3.625) {$S$};
		\node [style=white dot] (31) at (-2, 2.5) {};
		\node [style=large map] (32) at (1, 2.5) {$\mr{SWAP}$};
		\node [style=right label] (33) at (0, -3.625) {$S$};
		\node [style=right label] (34) at (2, 3.5) {$S$};
		\node [style=right label] (35) at (0, 3.5) {$S$};
		\node [style=none] (36) at (2, 0.75) {};
		\node [style=none] (37) at (2, -0.75) {};
	\end{pgfonlayer}
	\begin{pgfonlayer}{edgelayer}
		\draw (2.center) to (3.center);
		\draw (25.center) to (27.center);
		\draw (27.center) to (26.center);
		\draw (26.center) to (25.center);
		\draw (22) to (23);
		\draw (31) to (32);
		\draw (37.center) to (28.center);
		\draw (36.center) to (13.center);
		\draw (20.center) to (21.center);
	\end{pgfonlayer}
\end{tikzpicture}}$
    \fcaption{Quantum circuit for ${\tt CTRL}$ supermap. The supermap transforms sector-preserving  channels acting on a system $S$ with Hilbert space  $\spc H_S   =  \spc  H_S^0 \oplus \spc H_S^1$ into controlled channels acting on the composite system $C\otimes T$, consisting of  a control system $C$ and of a target system $T$ with Hilbert space $\spc H_T  \simeq  \spc  H_S^1$. The sector-preserving channel in input is inserted between two controlled {\tt SWAP} operations, which in turn are placed between two quantum channels $\cv$ and $\cd$, which serve as `adaptors', between the systems $T$ and $S$, and between the systems    $C\otimes S\otimes S$ and $C\otimes T$, respectively.   }
     \label{fig:CoherentControl}
\end{figure*}

We thus introduce the \textit{control supermap}, a supermap which takes as input any sector-preserving channel $\widetilde{\cc}[C_1]$ of type $(1,d)$, and yields the controlled channel ${\tt ctrl}_{C_1} \mhyphen \cc$ acting on a target system of dimension $d$.

\medskip 

\begin{theorem} \label{th:ControlSupermap}
Let $\ch_S = \ch_S^0 \oplus \ch_S^1$ be a Hilbert space, with $\dim(\ch_S^0)=1$ and $\dim(\ch_S^1)=d$, let $\ch_C$ be a control space of dimension 2, and $\ch_T$ be a target space, with $\ch_T  \simeq \ch_S^1$.

There exists a supermap ${\tt CTRL}$ of type $(S \to S) \to (C \otimes T \to C \otimes T)$ such that for any sector-preserving channel $\widetilde{\cc}[C_1]$,

\be
    {\tt CTRL}[\widetilde{\cc}[C_1]] = {\tt ctrl}_{C_1} \mhyphen \cc \, .
 \ee
 
Furthermore, this supermap is unitary-preserving on the sector-preserving channels on $S$.

\end{theorem}
\medskip 
\correction{
\begin{proof}
Let $ V   :  \spc H_T    \to \spc H_S$ be the isometry that maps $\spc H_T$ into the subspace $\spc H_S^1 \simeq \spc H_T$, let $|s^0\rangle$ be a unit vector in $\spc H_S^0$,    let $W: \spc H_C \otimes \spc H_S \otimes \spc H_S  \to  \spc H_C \otimes \spc H_T$ be the coisometry defined by  $W: = I  \otimes V^\dag  \otimes \langle s^0|$, 
and let $\map D$ be the quantum channel defined by   $  \map D (\rho)   :=  W  \rho W^\dag  +    \rho_0  \,  \Tr [P \,\rho]$, where $\rho_0$ is a fixed density matrix on  $\spc  H_C\otimes H_T  $ and $P   :=    I-   W^\dag W$\footnote{Note that the only thing that matters is how $\cd$ acts on the sector $\ch_C \otimes \ch_S^1 \otimes \ch_S^0$ of its input; its action on other sectors is irrelevant and can be defined in an arbitrary way, as long as it gives a CPTP map.} 
.  We then define the supermap ${\tt CTRL}$ through its  action on a generic linear map $ {\map M}  :  \cl (\spc H_C \otimes \spc  H_T)  \to \cl (\spc H_C \otimes \spc  H_T)$:

\begin{align} 
{\tt CTRL}  ({\map M}) :  =  \map D \circ {\tt ctrl}\mhyphen\map{SWAP} \circ   (\map  I_C \otimes \map I_S \otimes {\map M})  \circ {\tt ctrl}\mhyphen\map{SWAP} \circ (\map I_C \otimes \map V \otimes |s^0\rangle\langle s^0|   )\,
\end{align}

where $\map V$ is the quantum   channel corresponding to the isometry $V$, and ${\tt  ctrl}\mhyphen \map{SWAP}$ is the unitary channel corresponding to the controlled SWAP operator (see Figure \ref{fig:CoherentControl} for an illustration).

With this definition, one can verify that the condition $  {\tt CTRL}[\widetilde{\cc}[C_1]] = {\tt ctrl}_{C_1} \mhyphen \cc$ holds.  Let us prove it by showing that they act in the same way on pure states, using a Kraus representation  for the channel  $\widetilde{\cc}[C_1]$  with Kraus operators $\widetilde C_i  = \delta_{i1}\oplus C_i$.  From there, it can then be deduced by linearity that the two channels act in the same way on  any density matrix, and therefore that they are equal. We take a strict equality $T = S^1$ to avoid unnecessary clutter. 

Taking an arbitrary state $\ket{\psi}_{CT}$, we obtain

$$ {\tt ctrl}\mhyphen\mr{SWAP} (V \otimes \ket{s^0}) \ket{\psi}_{CT} = \ket{0}_C \otimes {}_{C}\!\braket{0|\psi}_{CS} \otimes \ket{s^0}_S + \ket{1}_C \otimes \ket{s^0}_S \otimes {}_{C}\!\braket{1|\psi}_{CS} $$

and thus

\begin{align}
\nonumber |\psi_i\rangle & :   ={\tt ctrl}\mhyphen\mr{SWAP} \,   (I_C \otimes I_S   \otimes \widetilde C_i) \, {\tt ctrl}\mhyphen\mr{SWAP} (V \otimes \ket{s^0}) \ket{\psi}_{CT} \\
 \nonumber  &= {\tt ctrl}\mhyphen\mr{SWAP} \,  \big(  \delta_{i1}  \,   \ket{0}_C \otimes {}_{C}\!\braket{0|\psi}_{CS} \otimes \ket{s^0}_S + \ket{1}_C \otimes \ket{s^0}_S \otimes   C_i  ~  {}_{C}\!     \braket{1|\psi}_{CS}   \big)    \\  &  = \big( \delta_{i1} \ket{0}_C \otimes \braket{0|\psi}_{CS} +   C_i ~  \ket{1}_C \otimes \braket{1|\psi}_{CS} \big) \otimes \ket{s^0} \, .
  \end{align}

Now, one has $P  \ket{\psi_i} = 0$, and $W \ket{\psi_i} = \delta_{i1}\,   \ket{0} \otimes {}_C\!\braket{0|\psi}_{CS} + \ket{1}_C \otimes C_i ~ {}_{C}\!\braket{1|\psi}_{CS} \equiv {\tt ctrl}_{\delta_{i1}}\mhyphen C_i~|\psi\rangle_{CT}$.   Summarising, if the control and target start off in the state $|\psi\rangle_{CT}$, and if the subprocess corresponding to the Kraus operator $\widetilde C_i$ takes place, then the final (subnormalized) state is $W|\psi_i\rangle   =  {\tt ctrl}_{\delta_{i1}}\mhyphen C_i~|\psi\rangle_{CT}  $.    On average over  all possible values of $i$, we obtain the evolution
\begin{align*}
    {\tt CTRL}(\widetilde{\cc}[C_1]) ( \ket{\psi}\bra{\psi}) &= \sum_i    W|\psi_i\rangle \langle \psi_i| W^\dag \\
    &= \sum_i ({\tt ctrl}_{\delta_{i1}}\mhyphen C_i) \ket{\psi}\bra{\psi} ({\tt ctrl}_{\delta_{i1}}\mhyphen C_i)^\dagger \\
    &= {\tt ctrl}_{C_1} \mhyphen \cc (\ket{\psi}\bra{\psi}) \,.
\end{align*}

As for the preservation of unitarity on sector-preserving channels, it is sufficient to recall Corollary \ref{cor:UnitaryCorrespondence}: unitary sector-preserving channels of type $(1,d)$ are of the form $\widetilde{\cu}: \rho \mapsto (1 \oplus U) \rho (1 \oplus U)^\dagger$. By the previous calculation, one then has ${ \tt CTRL}(\widetilde{\cu}) : \rho \mapsto (\ketbra{0}{0} \otimes I + \ketbra{1}{1} \otimes U) \rho (\ketbra{0}{0} \otimes I + \ketbra{1}{1} \otimes U)^\dagger$, which is a unitary channel.

\end{proof}}

The supermap ${\tt CTRL}$   constitutes a rigorous  theoretical formalisation of the existing experimental schemes for the implementation of coherent control. It is the universal protocol through which sector-preserving channels of type $(1,d)$ can be turned into their corresponding controlled channel.

We note that even though we defined this supermap as accepting as input \textit{any} possible channel $S \to S$, the only thing we are interested in is in fact its action on \textit{sector-preserving} channels. An alternative way of defining it would be to formally restrict its inputs to be only sector-preserving channels (or extensions of those); this would make clearer the fact that this protocol is only useful when sector-preserving channels are used, and would also allow to get rid of superfluous information in the specification of the supermap -- namely, information that only modifies the action of the supermap on non-sector-preserving channels. We will do this in Section \ref{sec:SupermapsRouted}, coining the notion of \textit{supermaps on routed channels}.

Let us also comment on the specific case of unitary channels. Per Corollary \ref{cor:UnitaryCorrespondence}, we know that  sector-preserving unitary \textit{channels} of type $(1,d)$ are in one-to-one correspondence with unitary \textit{operators} on their $d$-dimensional sector. Noting as ${U}$ the unitary operator corresponding to the unitary sector-preserving channel $\widetilde{\cu}$, the control supermap will then precisely map any  sector-preserving unitary channel $\widetilde{\cu}$ to the gate applying the controlled-unitary ${\tt ctrl} \mhyphen U$ defined in equation (\ref{eq:ctrlU}):

\be \forall \, \widetilde{\cu} \, \mr{unitary},\,\, {\tt CTRL}[\widetilde{\cu}] = {\tt ctrl} \mhyphen U \,. \ee

The control supermap thus also realises, in particular, the coherent control of unitary gates.

\subsection{Sector-preserving and controlled channels are equivalent  resources} \label{sec:SecPreservEquivControlled}

The previous section showed that there is a universal circuit structure which turns sector-preserving channels of type $(1,d)$ into their corresponding controlled channel. As resources, sector-preserving channels of type $(1,d)$ thus allow one to obtain controlled channels. We now show the opposite: from a controlled channel, one can obtain its corresponding sector-preserving channel of type $(1,d)$, once again using a universal circuit structure.
\medskip 

\begin{theorem} \label{th:InverseControlSupermap}
Let $\ch_T  \simeq \ch_S^1$ be a target space, and let  $\ch_C$ be a control space of dimension 2.    
Taking $\ch_S^0 \cong \mathbb{C}$, $\spc H_S^1\simeq  \spc  H_T $ and $\ch_S := \ch_S^0 \oplus \ch_S^1$, there exists a supermap ${\tt CTRL}^{-1}$ of type $(C \otimes T \to C \otimes T) \to (S \to S)$ such that for any controlled channel ${\tt ctrl}_{C_1} \mhyphen \cc$,

\be \label{eq:InverseControlSupermap}
    {\tt CTRL}^{-1}[{\tt ctrl}_{C_1} \mhyphen \cc] = \widetilde{\cc}[C_1]  \, .
 \ee
 
Furthermore, this supermap is unitary-preserving on the controlled channels on $C \otimes T$.

\end{theorem}

\begin{proof}
One can define ${\tt CTRL}^{-1}$'s action on a given map  $\ck$ of type $C \otimes T \to C \otimes T$ as ${\tt CTRL}^{-1}[\ck] = \cw \circ \ck \circ \cv$, where $\cv$ is the channel corresponding to  the isometry $V: \ch_S \to \ch_C \otimes \ch_T$ that acts as  $V \ket{\psi} = \ket{1} \otimes \ket{\psi}$ for $\ket{\psi} \in \ch_S^1$, and $V \ket{\psi} = \ket{0} \otimes \ket{\phi_0}$ for $\ket{\psi} \in \ch_S^0$ where $\ket{\phi^0}$ is a fixed arbitrary state in $\ch_T$, and channel   $\cw$ acts as $\cv^\dagger$ on $\cv$'s range and in an arbitrary way elsewhere.

From this definition, a simple computation shows that (\ref{eq:InverseControlSupermap}) holds.
\end{proof}

The existence of this inverse control supermap shows that sector-preserving channels of type $(1,d)$ and controlled channels are fully equivalent resources: one can go from a sector-preserving channel to its corresponding controlled channel and back again, using a universal circuit architecture in both cases. This concludes our demonstration of the main claim of this paper.

Note that ${\tt CTRL}^{-1} \circ {\tt CTRL}$ acts as the identity supermap only on input channels that are sector-preserving. A way of formally restricting the ${\tt CTRL}$ supermap to only act on sector-preserving channels will be described in Section \ref{sec:SupermapsRouted}. Once viewed in this way, the ${\tt CTRL}$ supermap can be said to be unitary-preserving and invertible.

\section{Implementing coherent control of multiple  channels} \label{sec:CoherentControl2}

\subsection{The case of isometric channels}

We now show how the previous methods apply to the coherent control of $N\ge 2$ channels, as defined in Section \ref{sec:2ControlDef}. For simplicity, we restrict ourselves to the case of isometric channels, and to $N=2$.  The methods we present are readily extendable to the $N>2$.  Note that the coherent control of isometric gates includes that of unitary gates and of pure states, as both are specific examples of isometric gates.

If we define the task of coherent control between two isometric gates as that of implementing controlled-$(U,V)$ (as defined in equation (\ref{eq:ctrlUV})) from uses of the isometric gates $U$ and $V$, then it is a direct consequence of the aforementioned no-go theorems that such a task cannot be achieved via a universal circuit architecture.

To circumvent this, we will instead keep our perspective of  considering coherent control as a task performed on sector-preserving channels. Here, as in Section \ref{sec:2ControlDef}, we take the input and output target systems to be of possibly different dimensions. Accordingly, we will slightly extend the relevant definitions. For instance, Definition \ref{def:sectorpres} can be extended in a straightforward way to encompass sector-preserving channels from $\ch_{S_\mr{in}} := \bigoplus_k \ch_{S_\mr{in}}^k$ to $\ch_{S_\mr{out}} := \bigoplus_k \ch_{S_\mr{out}}^k$. In the case in which the Hilbert spaces are both partitioned between a multi-dimensional sector and several one-dimensional ones, we will refer to these channels as being sector-preserving of type $(1 \to 1, \dots, 1 \to 1, d \to d')$. Structural theorems about these channels can be seen to extend from those of Section \ref{sec:OneSuppSector} (Lemma \ref{lem:characterisationSectorPreserving1}, Theorem \ref{th:characterisationSectorPreserving} and Corollaries \ref{cor:UnitaryCorrespondence} and \ref{cor:correspondence(1,d)}) in a natural way.

In particular, Corollary \ref{cor:UnitaryCorrespondence} can be extended to a statement about isometric sector-preserving channels $\cc$ of type $(1 \to 1, d \to d')$: they are in one-to-one correspondence with isometric operators $U_\cc$ in dimension $d \to d'$. Our point is to implement this correspondence physically in order to create a control between two isometric gates. We single out a  version of the control supermap that allows one to build the coherent control between two isometric gates from the two sector-preserving isometric channels of type $(1 \to 1, d \to d')$ corresponding to these isometries. This supermap was originally  introduced in Ref.\,\cite{Chiribella_2019} (in the case $d = d'$), in a slightly different framework. 

\medskip 

\begin{theorem} \label{th:2ControlSupermap}
Let $\ch_{S_\mr{in}} = \ch_{S_\mr{in}}^0 \oplus \ch_{S_\mr{in}}^1$  and $\ch_{S_\mr{out}} = \ch_{S_\mr{out}}^0 \oplus \ch_{S_\mr{out}}^1$ be partitioned spaces, with $\ch_{S_\mr{in}}^0$ and $\ch_{S_\mr{out}}^0$ one-dimensional, let $\ch_C$ be a control space of dimension 2, and let  $\ch_{T_\mr{in}} $ and  $\ch_{T_\mr{out}}$ be target spaces, with $\ch_{T_\mr{in}} \simeq  \ch_{S_\mr{in}}^1$ and  $\ch_{T_\mr{out}}\simeq \ch_{S_\mr{out}}^1$.  

There exists a supermap ${\tt 2 \mhyphen CTRL}$ of type $(S_\mr{in} \to S_\mr{out}) \otimes (S_\mr{in} \to S_\mr{out}) \to (C \otimes T_\mr{in} \to C \otimes T_\mr{out})$ such that for any pair of isometric sector-preserving channels $\cc$ and $\cd$,

\be
    {\tt 2 \mhyphen CTRL}[\cc \otimes \cd] = {\tt ctrl} \mhyphen (U_\cc,U_\cd) \, .
 \ee

\end{theorem}

\begin{proof}
This can be easily computed from the formulation of the ${\tt 2 \mhyphen CTRL}$ supermap shown in Figure \ref{fig:2CoherentControl}, in full analogy to the computation in the proof of Theorem \ref{th:ControlSupermap}.
\end{proof}

Theorem \ref{th:2ControlSupermap} can serve as a   formalisation of the existing experimental schemes for coherently controlling two unitaries, such as the superposition of paths \cite{Chiribella_2019}. It is easy to see that it could be readily generalised to the coherent control between $N$ isometries by a control system of dimension $N$.

In particular, one can see in this formulation that the coherent control of two isometries can be implemented with a simple parallel combination of the two resource sector-preserving channels.

\begin{figure*} 
    \centering
    $%
\begin{tikzpicture}
	\begin{pgfonlayer}{nodelayer}
		\node [style=none] (0) at (3.25, 3) {};
		\node [style=none] (1) at (3.25, 1.75) {};
		\node [style=none] (2) at (1.25, 1.75) {};
		\node [style=none] (3) at (1.25, -1.75) {};
		\node [style=none] (4) at (3.25, -1.75) {};
		\node [style=none] (5) at (3.25, -3) {};
		\node [style=none] (6) at (-3.25, -3) {};
		\node [style=none] (7) at (-3.25, 3) {};
		\node [style=none] (8) at (2.75, 1.75) {};
		\node [style=none] (9) at (2.75, 0.75) {};
		\node [style=none] (10) at (2.75, -0.75) {};
		\node [style=none] (11) at (2.75, -1.75) {};
		\node [style=none] (13) at (-1.25, 1.75) {};
		\node [style=none] (14) at (-1.25, -1.75) {};
		\node [style=none] (15) at (-2.75, 1.75) {};
		\node [style=none] (16) at (-2.75, -0.75) {};
		\node [style=none] (17) at (-2.75, -1.75) {};
		\node [style=none] (18) at (-3.25, 1.75) {};
		\node [style=none] (19) at (-3.25, -1.75) {};
		\node [style=none] (20) at (-2.75, 0.75) {};
		\node [style=none] (21) at (-1, -3) {};
		\node [style=none] (22) at (-1, -4) {};
		\node [style=right label] (23) at (-1, -3.625) {$C$};
		\node [style=none] (24) at (1, 3) {};
		\node [style=none] (25) at (1, 4) {};
		\node [style=right label] (26) at (1, 3.625) {$T_\mr{out}$};
		\node [style=right label] (27) at (2.75, 1) {$S_\mr{out}$};
		\node [style=right label] (28) at (2.75, -1.25) {$S_\mr{in}$};
		\node [style=right label] (29) at (-2.75, 1) {$S_\mr{out}$};
		\node [style=right label] (30) at (-2.75, -1.25) {$S_\mr{in}$};
		\node [style=none] (36) at (0, 0) {$\tt{2 \mhyphen CTRL}$};
		\node [style=none] (39) at (-1, 3) {};
		\node [style=none] (40) at (-1, 4) {};
		\node [style=right label] (41) at (-1, 3.625) {$C$};
		\node [style=none] (42) at (1, -3) {};
		\node [style=none] (43) at (1, -4) {};
		\node [style=right label] (44) at (1, -3.625) {$T_\mr{in}$};
	\end{pgfonlayer}
	\begin{pgfonlayer}{edgelayer}
		\draw (0.center) to (1.center);
		\draw (1.center) to (2.center);
		\draw (2.center) to (3.center);
		\draw (3.center) to (4.center);
		\draw (4.center) to (5.center);
		\draw (5.center) to (6.center);
		\draw (7.center) to (0.center);
		\draw (8.center) to (9.center);
		\draw (10.center) to (11.center);
		\draw (18.center) to (13.center);
		\draw (13.center) to (14.center);
		\draw (14.center) to (19.center);
		\draw (15.center) to (20.center);
		\draw (16.center) to (17.center);
		\draw (18.center) to (7.center);
		\draw (19.center) to (6.center);
		\draw (21.center) to (22.center);
		\draw (24.center) to (25.center);
		\draw (39.center) to (40.center);
		\draw (42.center) to (43.center);
	\end{pgfonlayer}
\end{tikzpicture}} \quad\quad  = \quad %
\begin{tikzpicture}
	\begin{pgfonlayer}{nodelayer}
		\node [style=none] (0) at (2, 0.75) {};
		\node [style=none] (1) at (2, -0.75) {};
		\node [style=none] (2) at (-2, 6.5) {};
		\node [style=none] (3) at (-2, -5.5) {};
		\node [style=right label] (4) at (-2, -5.375) {$C$};
		\node [style=right label] (5) at (0, 6.125) {$T_\mr{out}$};
		\node [style=right label] (6) at (2, 1) {$S_\mr{out}$};
		\node [style=right label] (7) at (2, -1) {$S_\mr{in}$};
		\node [style=right label] (9) at (-2, 6.125) {$C$};
		\node [style=none] (10) at (0, 6.5) {};
		\node [style=none] (11) at (0, -5.5) {};
		\node [style=right label] (12) at (0, -5.375) {$T_\mr{in}$};
		\node [style=very very large map] (14) at (0.5, 4.75) {$\cd$};
		\node [style=none] (15) at (2, 4.75) {};
		\node [style=none] (16) at (2, -4) {};
		\node [style=right label] (17) at (-2, 0) {$C$};
		\node [style=none] (18) at (0, 0.75) {};
		\node [style=none] (19) at (0, -0.75) {};
		\node [style=right label] (20) at (0, 1) {$S_\mr{out}$};
		\node [style=right label] (21) at (0, -1) {$S_\mr{in}$};
		\node [style=none] (22) at (0, 6.5) {};
		\node [style=none] (23) at (0, -5.5) {};
		\node [style=white dot] (24) at (-2, -2.5) {};
		\node [style=large map] (25) at (1, -2.5) {$\mr{SWAP}$};
		\node [style=none] (26) at (2.075, -4.35) {$s^0$};
		\node [style=none] (27) at (2, -5) {};
		\node [style=none] (28) at (1.25, -4) {};
		\node [style=none] (29) at (2.75, -4) {};
		\node [style=none] (30) at (2, -4) {};
		\node [style=map] (33) at (0, -4.5) {$\cv$};
		\node [style=right label] (34) at (0, -3.625) {$S_\mr{in}$};
		\node [style=white dot] (35) at (-2, 2.5) {};
		\node [style=large map] (36) at (1, 2.5) {$\mr{SWAP}$};
		\node [style=right label] (37) at (2, -3.625) {$S_\mr{in}$};
		\node [style=right label] (38) at (0, 3.5) {$S_\mr{out}$};
		\node [style=right label] (39) at (2, 3.5) {$S_\mr{out}$};
	\end{pgfonlayer}
	\begin{pgfonlayer}{edgelayer}
		\draw (2.center) to (3.center);
		\draw (15.center) to (0.center);
		\draw (16.center) to (1.center);
		\draw (22.center) to (18.center);
		\draw [in=270, out=90] (23.center) to (19.center);
		\draw (27.center) to (29.center);
		\draw (29.center) to (28.center);
		\draw (28.center) to (27.center);
		\draw (24) to (25);
		\draw (35) to (36);
	\end{pgfonlayer}
\end{tikzpicture}}$ 
    \fcaption{Quantum circuit for the ${\tt 2 \mhyphen CTRL}$ supermap.  The input of the supermap are  two sector-preserving  channels transforming  a system $S_{\rm in}$ with Hilbert space  $\spc H_{S_{\rm in}}   =  \spc  H_{S_{\rm in}}^0 \oplus \spc H_{S_{\rm in}}^1$ into a system $S_{\rm out}$ with Hilbert space  $\spc H_{S_{\rm out}}   =  \spc  H_{S_{\rm out}}^0 \oplus \spc H_{S_{\rm out}}^1$. 
    The output of the supermap  is  a controlled channel transforming the composite system $C\otimes T_{\rm in}$ with $\spc  H_{T_{\rm in}}  \simeq  \spc H_{S_{\rm in}^1}$  into the composite system $C\otimes T_{\rm out}$ with $\spc  H_{T_{\rm out}}  \simeq  \spc H_{S_{\rm out}^1}$.  The channels  $\cv$ and $\cd$ and the state $|s^0\rangle$ are defined as in Theorem \ref{th:ControlSupermap}.      A very similar supermap  was  defined  in Ref.\,\cite{Chiribella_2019} for the case $T_\mr{in} = T_\mr{out}$.}
    \label{fig:2CoherentControl}
\end{figure*}

\subsection{What about general channels?} \label{sec:AboutGeneralChannels}

A natural question to ask would be whether the previous result can be extended to the case of controls between two general noisy channels, as defined in equation (\ref{eq:ctrl2Kraus}) and classified in Theorem \ref{th:2ControlParam}: i.e., whether a given version of a control between two channels $\ca$ and $\cb$ can be obtained from the application of the ${\tt 2 \mhyphen CTRL}$ supermap on suitably chosen sector-preserving channels of type $(1 \to 1, d \to d')$. The answer to this question, however, is negative.

To see this, take $\ca = \cb = \cd$, where $\cd$ is the depolarising channel on a qubit, i.e. $\cd: \rho \mapsto \frac{1}{2} ( \rho + Z \rho Z)$. One natural version of a control between $\ca$ and $\cb$ is then given by the channel $ \ci_C \otimes \cd_{S^1}$: i.e., $\cd$ is always applied to $S^1$ and the control doesn't play any part. However, no use of the ${\tt 2 \mhyphen CTRL}$ supermap on $\ca$ and $\cb$ can yield this channel. This is essentially because, in channels obtained from the use of the ${\tt 2 \mhyphen CTRL}$ supermap, there can only be full coherence between one Kraus operator of $\ca$ and one Kraus operator of $\cb$.

Implementing the control between two noisy channels in general will therefore require the use of a more elaborate scheme, using more involved resources. In Appendix \ref{app:2ControlGeneral}, we propose such a scheme. Rather than sector-preserving channels of the form $\mathbb{C} \oplus \ch^1_{S_\mr{in}} \to \mathbb{C} \oplus \ch^1_{S_\mr{out}}$, this scheme will require the use of sector-preserving channels of the form $\mathbb{C} \oplus \ch^1_{S_\mr{in}} \to \mathbb{C} \oplus \ch^1_{S_\mr{out}} \otimes \ch^1_E$, where $\ch^1_E$ is an auxiliary Hilbert space, representing the environment.  In such a scheme, the number of Kraus operators of $\ca$ and $\cb$ that can be coherent with each other in the produced controlled channel is capped by the dimension of $\ch^1_E$.

\section{Compositely-controlled channels} \label{sec:compositelycontrolled}
In this section, we consider another generalisation of the notion of controlled quantum channels, corresponding to higher-dimensional control systems, and we show how to implement it, via a universal architecture, using as resources sector-preserving channels of type $(d,1, \dots, 1)$.

\subsection{Compositely-controlled channels and multiple pinned operators}  

We introduce a generalisation that can be useful in the description of quantum programs, which may contain instructions of the form `if $f(x) =1$, then execute channel $\cc$, otherwise do nothing', where  $f$ is a Boolean function taking as input  a parameter $x$ labelling the different branches of the computational process. 

To get started, consider a three-dimensional control system $C$, with basis states $\{|0\rangle, |1\rangle,  |2\rangle\}$. We associate state $|0\rangle$ to the execution of the given channel $\cc$, and states $|1\rangle$ and $|2\rangle$ to the `do nothing' option. This corresponds to choosing the Boolean function $f$ to be $f(0)=1$ and $f(1)=  f(2)  =  0$. A controlled channel can then be defined in terms of the Kraus operators 

\begin{align}\label{eq:ctrlKraus2}
{\tt ctrl}_{\alpha_i, \beta_i} \mhyphen   C_i   :=  |0\rangle \langle 0|   \otimes C_i   +   \alpha_i \,  |1\rangle \langle 1|    \otimes I  +  \beta_i \,   |2\rangle \langle 2|  \otimes I   \, ,
\end{align}
where ${\bf C}   :  =(C_i)_{i=1}^n$ is a Kraus representation of channel $\cc$, and $\boldsymbol{\alpha} : =(\alpha_i)_{i=1}^n$ and $\boldsymbol{\beta}:=(\beta_i)_{i=1}^n$ are complex amplitudes satisfying the normalisation conditions $\sum_{i=1}^n  \, |\alpha_i|^2  =1$ and $\sum_{i=1}^n  \, |\beta_i|^2  =1$, respectively. We shall call a controlled channel as defined in (\ref{eq:ctrlKraus2}) a \textit{$2$-compositely-controlled} channel.  In the following, this controlled channel will be denoted by ${\tt ctrl}_{ \boldsymbol{\alpha},  \boldsymbol{\beta}}^{\bf C} \mhyphen \map C$.

As with standard controlled channels, different choices of Kraus representations and of amplitudes generally lead to different kinds of controlled channels, and again, one may ask for a one-to-one parametrisation.  The  generalisation of Theorem \ref{th:characterisation} is the following.
\medskip 

\begin{theorem}\label{th:characterisation2}
The  $2$-compositely-controlled versions of  channel $\cc$, as defined by  Eq.~(\ref{eq:ctrlKraus2})  are in one-to-one\footnote{Except in the case where $\abs{\gamma_{12}} =1$; the choice of $C_2'$ is then irrelevant. Given that this is a set of measure $0$, we will neglect the existence of this case in the rest of this paper.} correspondence with  triples of the form $  (C_1', C_2',  \gamma_{12})$, where $C_1'$ and $C_2'$ are two Kraus operators for channel $\cc$,  and of a complex amplitude $\gamma_{12} \in \C$ satisfying $|\gamma_{12} |  \le 1$. Explicitly, the Kraus operators for the controlled channel can be written as 
\be \label{newkraus} \begin{cases}
     \widehat{C}_1' = \ketbra{0}{0} \otimes  C_1'   + \ketbra{1}{1} \otimes   I  +   \gamma_{12}\,   \ketbra{2}{2} \otimes  I    \\
     \widehat{C}_2' = \ketbra{0}{0}  \otimes C_2'  +   \sqrt{  1-  |\gamma_{12}|^2}\,   \ketbra{2}{2} \otimes  I \\
     \widehat{C}_i'  =    \ketbra{0}{0}  \otimes  C_i' \quad\quad \forall i \geq 3  \, ,
\end{cases}
\ee
for some suitable Kraus representation $(C_i')_{i}$ of channel $\cc$. 
\end{theorem}

\Proof The proof is a generalisation of the proof of Theorem \ref{th:characterisation}.    Starting from the Kraus operators in  Eq.~(\ref{eq:ctrlKraus2}), one can generate a new Kraus representation of the controlled channel using  a unitary matrix.    To choose the appropriate unitary matrix, we apply the Gram-Schmidt construction to the column vectors $\ket{\alpha}  = (\alpha_i)_{i=1}^n$ and $\ket{\beta} =  (\beta_i)_{i=1}^n$.  In other words, we construct an orthonormal basis   $(|v_i\rangle)_{i=1}^n$ where the first vector is  $\ket{v_1}  =  \ket{\alpha}$ and the second vector is $\ket{v_2}  =  \ket{\beta}  -   \braket{  \alpha| \beta}  \,  \ket{\alpha}/\|    \ket{\beta}  -   \braket{  \alpha|  \beta}  \,  \ket{\alpha} \|$.   One can then define the unitary operator $U  =  \sum_{j=1}^n \ketbra{j}{ v_j}$, and use its matrix elements to define a new Kraus representation   $\widehat{C}_j    =\sum_i     U_{ji} \,   {\tt ctrl}_{\alpha_i, \beta_i} \mhyphen   C_i   $.   Explicit calculation of the Kraus operators yields  Eq.~(\ref{newkraus}), with $\gamma_{12}  =  \braket{\alpha|\beta}$,  $ C_1'   =    \sum_i  \overline \alpha_i  \,    C_i $ and   $ C_2'   =    \sum_i   (\overline \beta_i-  \braket{\beta |\alpha} \, \overline \alpha_i  )/\sqrt{1-|\gamma_{12}|^2}\,    C_i$.  

For every given controlled channel ${\tt ctrl}_{\boldsymbol{\alpha},  \boldsymbol{\beta}}^{\bf C} \mhyphen \map C$, the pinned Kraus operators $C_1$ and $C_2$, and the amplitude $\gamma_{12}$ can be uniquely determined from the action of the channel on a generic product state of the target and the control. Explicitly, one has
\begin{align}
\nonumber 
&{\tt ctrl}_{\boldsymbol{\alpha},  \boldsymbol{\beta}}^{\bf C} \mhyphen \map C  (\rho_C\otimes \rho_T)   \\
\nonumber & \quad     =       \langle 0  |  \rho_C  |0\rangle    ~ \ketbra{0}{0}_C  \otimes  \cc ( \rho_T)     
   \\
 \nonumber  &\qquad         +   \langle 1  |  \rho_C  |1\rangle  ~   \ketbra{1}{1}_C  \otimes \rho_T  \\ 
 \nonumber  &\qquad         +   \langle 2 |  \rho_C  |2\rangle  ~   \ketbra{2}{2}_C   \otimes \rho_T \\ 
\nonumber   &  \qquad  +    \langle 0  |  \rho_C  |1\rangle    ~   C_1'\,  \ketbra{0}{1}_C  \otimes \rho_T   +  {\rm h.c.}  \\
\nonumber   &  \qquad  +    \langle 0  |  \rho_C  |2\rangle    ~  \sqrt{  1-  |\gamma_{12}|^2}\,   \,   \ketbra{0}{2}_C   \otimes C_2'\,  \rho_T   +  {\rm h.c.}  \\
  &  \qquad  +    \langle 1  |  \rho_C  |2\rangle    ~ \gamma_{12} \, \ketbra{1}{2}_C   \otimes     \rho_T +  {\rm h.c.}    \, ,
  \end{align} 
  
 from which the  operators $C_1$ and $C_2$, and the amplitude $\gamma_{12}$ can be extracted by taking the appropriate matrix elements of the output state.   

In summary, every $2$-compositely-controlled channel as defined by  Eq.~(\ref{eq:ctrlKraus2})   can be parameterised by two pinned Kraus operators   $(C_1',C_2')$ and one amplitude  $\gamma_{12}$, and the triple  $(C_1',C_2',\gamma_{12})$ is uniquely determined by the channel.    \qed 

\smallskip  
The above notion of controlled channel can be easily extended  to higher dimensional systems, introducing controlled Kraus operators of the form    
\begin{align}\label{eq:ctrlKrausm}  
{\tt ctrl} \mhyphen C_i  =  \ketbra{0}{0} \otimes   C_i   +  \sum_{k=1}^m  \,  \alpha_i^{k}  \,  \ketbra{k}{k}  \otimes   I \,,
\end{align}  where, for each $k  \in  \{1,\dots, m\}$,  the amplitudes $(  \alpha_i^k)_i$ satisfy the normalisation condition $\sum_i    \,  \left|\alpha_i^{k}\right|^2  =1$. Controlled channels of  the form (\ref{eq:ctrlKrausm}) will be called $m$-compositely controlled channels. In this case, the controlled channel is in one-to-one correspondence with $m$ pinned Kraus operators,  and $m-1$ complex amplitudes:  using the same argument as in the proof of Theorem  \ref{th:characterisation2}, one can show that the controlled channel has a Kraus representation of the form 

 \be   \begin{cases}
     \widehat{C}_1 = \ketbra{0}{0} \otimes C_1'   + \ketbra{1}{1} \otimes  I +    \sum_{k  >1} \gamma_{1k}\,  \ketbra{k}{k}  \otimes   I   \\
     \widehat{C}_2 = \ketbra{0}{0} \otimes C_2'   +   \sqrt{  1-  |\gamma_{12}|^2}\,    \ketbra{2}{2}  \otimes I +  \sum_{k>2}    \gamma_{2k}  \, \ketbra{k}{k} \otimes I  \\
     \widehat{C}_3 = \ketbra{0}{0}  \otimes C_3'  +   \sqrt{  1-  |\gamma_{13}|^2  -  |\gamma_{23}|^2}\, \ketbra{3}{3}   \otimes  I  +  \sum_{k>3}    \gamma_{3k}  \,  \ketbra{k}{k} \otimes I \\
    \vdots  \\ 
   \widehat{C}_{j} = \ketbra{0}{0} \otimes  C_{j}'   +   \sqrt{  1- \sum_{i<j} |\gamma_{ij}|^2  }\, \ketbra{j}{j} \otimes  I   +  \sum_{k>j}    \gamma_{jk}  \, \ketbra{k}{k}  \otimes    I  \qquad  \forall j  \le  m\\
     \widehat{C}_j  =   \ketbra{0}{0}   \otimes  C_j' \qquad \forall j  > m  \, ,
\end{cases}
\ee

where $(C_i')$ is a Kraus representation of channel $\cc$, and $(\gamma_{ij})_{i<k\le m}$ are suitable amplitudes. 
 
In summary, controlled channels can represent if-then clauses in the execution of a quantum program, and every branch of the program corresponding to the `do nothing' instruction corresponds to a pinned Kraus operator.  

\subsection{A resource: sector-preserving channels of type $(d, 1 \dots, 1)$} 

\begin{figure*}
    \centering
    $%
\begin{tikzpicture}
	\begin{pgfonlayer}{nodelayer}
		\node [style=none] (0) at (-1.75, 3) {};
		\node [style=none] (1) at (-1.75, -3) {};
		\node [style=none] (2) at (3.5, -3) {};
		\node [style=none] (3) at (3.5, 3) {};
		\node [style=none] (4) at (1.25, 1.75) {};
		\node [style=none] (5) at (1.25, -1.75) {};
		\node [style=none] (6) at (2.5, 1.75) {};
		\node [style=none] (7) at (2.5, -0.75) {};
		\node [style=none] (8) at (2.5, -1.75) {};
		\node [style=none] (9) at (3.5, 1.75) {};
		\node [style=none] (10) at (3.5, -1.75) {};
		\node [style=none] (11) at (2.5, 0.75) {};
		\node [style=none] (12) at (0, -3) {};
		\node [style=none] (13) at (0, -4) {};
		\node [style=right label] (14) at (0, -3.625) {$C$};
		\node [style=none] (15) at (2.5, 3) {};
		\node [style=none] (16) at (2.5, 4) {};
		\node [style=right label] (17) at (2.5, 3.625) {$T \simeq S^0$};
		\node [style=right label] (18) at (2.5, 1) {$S$};
		\node [style=right label] (19) at (2.5, -1.25) {$S$};
		\node [style=none] (20) at (0, 3) {};
		\node [style=none] (21) at (0, 4) {};
		\node [style=right label] (22) at (0, 3.625) {$C$};
		\node [style=none] (23) at (2.5, -3) {};
		\node [style=none] (24) at (2.5, -4) {};
		\node [style=right label] (25) at (2.5, -3.625) {$T \simeq S^0$};
		\node [style=none] (26) at (-0.25, 0) {$\tt{CTRL}_{(2)}$};
	\end{pgfonlayer}
	\begin{pgfonlayer}{edgelayer}
		\draw (6.center) to (11.center);
		\draw (7.center) to (8.center);
		\draw (12.center) to (13.center);
		\draw (15.center) to (16.center);
		\draw (20.center) to (21.center);
		\draw (23.center) to (24.center);
		\draw (0.center)
			 to (1.center)
			 to (2.center)
			 to (10.center)
			 to (5.center)
			 to (4.center)
			 to (9.center)
			 to (3.center)
			 to cycle;
	\end{pgfonlayer}
\end{tikzpicture}} \quad\quad := \quad %
\begin{tikzpicture}
	\begin{pgfonlayer}{nodelayer}
		\node [style=none] (2) at (-2, 6.5) {};
		\node [style=none] (3) at (-2, -5.5) {};
		\node [style=right label] (4) at (-2, -5.375) {$C$};
		\node [style=right label] (5) at (4, 6.125) {$T$};
		\node [style=right label] (6) at (2, 0) {$S$};
		\node [style=right label] (7) at (4, -1) {$S$};
		\node [style=right label] (8) at (-2, 6.125) {$C$};
		\node [style=none] (10) at (4, -5.5) {};
		\node [style=right label] (11) at (4, -5.375) {$T$};
		\node [style=very very large map] (12) at (1, 4.75) {$\cd$};
		\node [style=none] (13) at (4, 6.5) {};
		\node [style=none] (14) at (2, -4) {};
		\node [style=right label] (15) at (-2, 0) {$C$};
		\node [style=right label] (18) at (4, 1) {$S$};
		\node [style=none] (20) at (2, 4.5) {};
		\node [style=none] (21) at (4, -5.5) {};
		\node [style=none] (24) at (2.075, -4.35) {$s^1$};
		\node [style=none] (25) at (2, -5) {};
		\node [style=none] (26) at (1.25, -4) {};
		\node [style=none] (27) at (2.75, -4) {};
		\node [style=none] (28) at (2, -4) {};
		\node [style=map] (29) at (4, -4.25) {$\ci$};
		\node [style=right label] (30) at (2, -3.625) {$S$};
		\node [style=right label] (33) at (4, -3.625) {$S$};
		\node [style=right label] (34) at (2, 3.5) {$S$};
		\node [style=right label] (35) at (4, 3.5) {$S$};
		\node [style=none] (39) at (0, -4) {};
		\node [style=none] (40) at (0.075, -4.35) {$s^2$};
		\node [style=none] (41) at (0, -5) {};
		\node [style=none] (42) at (-0.75, -4) {};
		\node [style=none] (43) at (0.75, -4) {};
		\node [style=none] (44) at (0, -4) {};
		\node [style=right label] (45) at (0, -3.625) {$S$};
		\node [style=right label] (47) at (0, 0) {$S$};
		\node [style=none] (48) at (0, 4.5) {};
		\node [style=right label] (49) at (0, 3.5) {$S$};
		\node [style=none] (51) at (4, 0.75) {};
		\node [style=none] (52) at (4, -0.75) {};
		\node [style=very very large map] (53) at (1, 2.5) {$\mr{c-perm}^\dagger$};
		\node [style=very very large map] (54) at (1, -2.5) {$\mr{c-perm}$};
	\end{pgfonlayer}
	\begin{pgfonlayer}{edgelayer}
		\draw (2.center) to (3.center);
		\draw (25.center) to (27.center);
		\draw (27.center) to (26.center);
		\draw (26.center) to (25.center);
		\draw (41.center) to (43.center);
		\draw (43.center) to (42.center);
		\draw (42.center) to (41.center);
		\draw (21.center) to (52.center);
		\draw (51.center) to (13.center);
		\draw (44.center) to (48.center);
		\draw (20.center) to (28.center);
	\end{pgfonlayer}
\end{tikzpicture}}$
    \fcaption{The ${\tt CTRL}_{(2)}$ supermap. $s^1$ and $s^2$ are the only normalised states in $\cl(\ch_S^1)$ and $\cl(\ch_S^2)$, respectively. $\ci: \cl(\ch_T) \to \cl(\ch_S)$ is an isometric channel such that $\ci(\rho) = \rho$ (i.e. it just embeds $\cl(\ch_T)$ within $\cl(\ch_S)$). $\mr{c-perm}$ is the unitary gate which, depending on the state of $C$, performs a cyclic permutation of the three $S$ wires: it performs the identity if $C$ is in state $|0\rangle$, connects each wire to its right neighbour if $C$ is in the state $|1\rangle$, and connects each wire to its left neighbour if $C$ is in the state $|2\rangle$. $\cd : \cl(\ch_C \otimes \ch_S \otimes \ch_S \otimes \ch_S) \to \cl(\ch_C \otimes \ch_T)$ reduces to the identity on the sector $\cl(\ch_C \otimes \ch_S^2 \otimes \ch_S^1 \otimes \ch_S^0) \simeq \cl(\ch_C \otimes \ch_S^0) \simeq \cl(\ch_C \otimes \ch_T)$ of its input space; its action on other sectors is irrelevant and can be arbitrarily defined, as long as it makes $\cd$ into a CPTP map.}
     \label{fig:CoherentControl2}
\end{figure*}

We now consider the types of channels that can be used as resources for the implementation of $m$-composite control. To do this, we extend the approach of section \ref{sec:experimental} to consider sector-preserving channels with one $d$-dimensional sector and $m$ $1$-dimensional sectors, i.e. those of type $(d,  \underbrace{1,\dots, 1}_{m~{\rm times}} )$.   

The Kraus operators of such channels have the form  
\begin{align}
\widetilde C_i   =   C_i  \oplus   \alpha^{1}_i  \,   \oplus   \alpha^{2}_i  \oplus \dots \oplus  \alpha^{m}_i  \, ,  
\end{align} 
where $(C_i)_i$ is a Kraus representation of a channel $\map C$ in dimension $d$, and, for every $k  \in \{1,\dots, m\}$,  $(\alpha_i^{k})_i$ are amplitudes satisfying the condition $\sum_i  |\alpha_i^{k}|^2 =1$. The existence of this form is a consequence of Theorem 6 in Ref.\,\cite{vanrietvelde2020routed}.   

A one-to-one parametrisation can be obtained using the same argument as in the proof of Theorem  \ref{th:characterisation2}, which allows us to show that every sector-preserving channel of this type is in one-to-one
 correspondence with  $m$ pinned Kraus operators  of a channel $\cc$ in dimension $d$, and with a set of complex amplitudes $(  \gamma_{ij})_{1\le i  <j \le m}$.       
To illustrate the situation, we consider the $m=2$  case.   In this case, it is possible to show that every sector-preserving channel $\widetilde \cc$ admits a Kraus representation of the form        
\be  \begin{cases}
     \widetilde{C}_1 = C_1  \oplus 1   \oplus    \gamma_{12}     \\
     \widetilde{C}_2 = C_2  \oplus 0  \oplus   \sqrt{  1-  |\gamma_{12}|^2}   \\
     \widetilde{C}_i  =     C_i  \oplus 0  \oplus 0  \quad\quad \forall i \geq 3  \, ,
\end{cases}
\ee
where $(C_i)_i$ are Kraus operators of a suitable quantum channel $\cc$ in dimension $d$.  
 Note that this expression is completely analogous to Eq.\,(\ref{newkraus}).   In this case, it is possible to show that the quadruple $(\cc, C_1,  C_2,  \gamma_{12})$ provides a one-to-one parametrisation: 
 \medskip 

 \begin{theorem}\label{th:sectorpreservingd11}
The sector-preserving channels of type $(d,1,1)$ are in one-to-one correspondence with quadruples of the form  $(\cc,  C_1,  C_2,  \gamma_{12})$, where  $\cc :  \cl (\ch_S^0) \to \cl (\ch_S^0)$ is a channel in dimension $d$,  $C_1$ and $C_2$ are two Kraus operators for $\cc$,   and $\gamma$ is a complex amplitude satisfying $|\gamma_{12}|  \le 1$.   
 \end{theorem} 
 \medskip 

  Combining Theorem \ref{th:sectorpreservingd11}  with Theorem \ref{th:characterisation2}, we obtain the following.
  \medskip 

 \begin{corollary}\label{cor:correspondence(d,1,1)}
The following sets are in one-to-one correspondence:  
\begin{enumerate}
\item  $2$-compositely-controlled channels, as defined in (\ref{eq:ctrlKraus2}), with a $d$-dimensional target system;
\item  quadruples of the form  $(\cc,  C_1,  C_2,  \gamma_{12})$, where  $(\map C,  C_1, C_2)$ is a channel  with two pinned Kraus operators in dimension $d$, and $\gamma_{12}$ is a complex amplitude satisfying $|\gamma_{12}|\le 1$;
\item  sector-preserving channels of type $(d,1,1)$.
\end{enumerate}
\end{corollary}
\medskip 

The case of arbitrary $m\ge2$ can be treated similarly, and also in this case, one can show that there exists a one-to-one correspondence between the set of $m$-compositely-controlled quantum channels of type (\ref{eq:ctrlKrausm}) and the set of sector-preserving channels of type $(d,  \underbrace{1,\dots, 1}_{m~{\rm times}} )$.

\subsection{Implementing compositely-controlled channels via a universal circuit architecture}

We now turn to the generalisation of the result of Section \ref{sec:ControlSupermap+Equivalence} to the implementation of $m$-compositely-controlled channels. For illustration, we once again focus on the case $m=2$.

As stated in Corollary \ref{cor:correspondence(d,1,1)}, for any given $d$, there is indeed a one-to-one correspondence between the $2$-compositely-controlled channels on target systems of dimension $d$ (which can be written as the $\rm{ctrl}_{C_1,C_2, \gamma_{12}} \mhyphen \cc$), and the sector-preserving channels of type $(d,1,1)$ (which can be written as the $\widetilde{\cc}[C_1,C_2, \gamma_{12}]$). This correspondence can also be implemented via a universal circuit architecture.
\medskip 

\begin{theorem} \label{th:Control(2)Supermap}
Let $\ch_S = \ch_S^0 \oplus \ch_S^1 \oplus \ch_S^2$ be a Hilbert space, with $\dim(\ch_S^0)=d$ and $\dim(\ch_S^1)=\dim(\ch_S^2)=1$, and let $\ch_C$ be a control space of dimension 3.

There exists a supermap ${\tt CTRL}_{(2)}$ of type $(S \to S) \to (C \otimes S^0 \to C \otimes S^0)$ such that for any sector-preserving channel $\widetilde{\cc}[C_1,C_2, \gamma_{12}]$,

\be
    {\tt CTRL}_{(2)}[\widetilde{\cc}[C_1,C_2, \gamma_{12}]] = \rm{ctrl}_{C_1,C_2, \gamma_{12}} \mhyphen \cc \, .
 \ee
 
Furthermore, this supermap is unitary-preserving on the sector-preserving channels on $\ch_S$.
\end{theorem}
\medskip 

This Theorem can be proven in a straightforward way using the formulation of the ${\tt CTRL}_{(2)}$ supermap presented in Figure \ref{fig:CoherentControl2}.

Similarly, an inverse ${\tt CTRL}_{(2)}^{-1}$ of this control map can easily be defined, showing that sector-preserving channels of type $(d,1,1)$ and $2$-compositely-controlled channels are fully equivalent resources.

These results can be generalised in a straightforward way to the case of general $m$: for any given $m$, $d$, there exists a universal circuit architecture (represented by a supermap $\rm{CTRL}_{(m)}$) turning a sector-preserving channel of type $(d, \underbrace{1,\dots, 1}_{m~{\rm times}} )$ into its corresponding $m$-compositely-controlled channel, and a universal circuit architecture realising the converse task. 

\section{Supermaps on routed channels} \label{sec:SupermapsRouted}

We now turn to a formal construction allowing to describe the ${\tt CTRL}$  and ${\tt 2 \mhyphen CTRL}$ supermaps as acting solely on sector-preserving channels. We achieve this through the introduction of the notion of supermaps on routed channels.

Supermaps, first introduced in \cite{chiribella2008supermaps}, can be conceptually defined as `operations on operations': they are linear transformations taking quantum channels as input and mapping them to output quantum channels. Their main use is to model the different ways of using and connecting together `black-box' operations \cite{chiribella2008supermaps}, for example in a quantum comb \cite{chiribella2009combs} or in more exotic setups, such as the quantum switch \cite{chiribella2009switch}; they provide a rigorous  framework for  studying the features and relative advantages of these manipulations of the black boxes.

Here, we define (deterministic) `supermaps on routed channels' as supermaps which only accept a subset of all channels as input; namely, in the language of Ref.\,\cite{vanrietvelde2020routed}, those that follow a certain route -- i.e. satisfy a given set of sectorial constraints. These restrictions will make the possible supermaps more diverse, as they are no longer required to be well-defined on all possible input channels. Fortunately, a good deal of the formal work necessary in order to define such supermaps on routed channels has been undertaken already: in \cite{chiribella2009switch}, deterministic supermaps on a restricted subset of quantum channels were defined in general. We will recall the main parts of this definition, then apply it to the definition of supermaps on routed channels.

We denote a system $X$ as corresponding to a finite-dimensional Hilbert space $\ch_X$. For two systems $A_\mr{in}$ and $A_\mr{out}$, we denote $\mr{Herm}(A_\mr{in} \to A_\mr{out})$ to be the real vector space of Hermitian-preserving linear maps from $\cl(\ch_{A_\mr{in}})$ to $\cl(\ch_{A_\mr{out}})$, and $\mr{QChan}(A_\mr{in} \to A_\mr{out})$ to be its subset containing quantum channels of type $A_\mr{in} \to A_\mr{out}$. We also note $\mr{St}(X) \subseteq \mr{Herm}(\ch_X)$ to be the set of states for system $X$. The first notion we need is that of an extension of a set of channels, which allows us to take into consideration channels which also act on an auxiliary system. Given a subset of channels $S \subseteq \mr{QChan}(A_\mr{in} \to A_\mr{out})$ and two systems $X_\mr{in}$, $X_\mr{out}$, the extension of $S$ in $\mr{QChan}(A_\mr{in} X_\mr{in} \to A_\mr{out} X_\mr{out})$ is the set $\mr{Ext}_{X_\mr{in} \to X_\mr{out}}(S):= \{\cc \in \mr{QChan}(A_\mr{in} X_\mr{in} \to A_\mr{out} X_\mr{out}) \,  | \, \forall \sigma \in \mr{St}(X_\mr{in}), \Tr_{X_\mr{out}}\left( \cc \circ (\mathbb{1}_{A_\mr{in}} \otimes \sigma_{X_\mr{in}}) \right) \in S  \}$. With this notion, one can  define deterministic supermaps on a restricted subset of channels \cite{chiribella2009switch}.

\medskip 

\begin{definition}[Deterministic supermaps on a restricted subset of quantum channels] \label{def:SupermapRestricted}
Let $S \subseteq \mr{QChan}(A_\mr{in} \to A_\mr{out})$ and $T \subseteq \mr{QChan}(P \to F)$ be subsets of channels. A deterministic supermap of type $S \to T$ is a linear map $\cs$ from $\mr{Herm}(A_\mr{in} \to A_\mr{out})$ to $\mr{Herm}(P \to F)$ such that, for any auxiliary systems $X_\mr{in}$, $X_\mr{out}$ and for any channel $\cc \in \mr{Ext}_{X_\mr{in} \to X_\mr{out}}(S)$, one has

\be (\cs \otimes \ci_{X_\mr{in} \to X_\mr{out}}) [\cc] \in \mr{Ext}_{X_\mr{in} \to X_\mr{out}}(T) \, , \ee

where $\ci_{X_\mr{in} \to X_\mr{out}}$ is the identity supermap on $\mr{Herm}(X_\mr{in} \to X_\mr{out})$.
\end{definition}
\medskip 

\begin{figure*}
    \centering
    \scalebox{.8}{$%
\begin{tikzpicture}
	\begin{pgfonlayer}{nodelayer}
		\node [style=none] (0) at (1.5, 2.75) {};
		\node [style=none] (5) at (1.5, -2.75) {};
		\node [style=none] (6) at (-1.5, -2.75) {};
		\node [style=none] (7) at (-1.5, 2.75) {};
		\node [style=none] (12) at (1, 0) {$\cs$};
		\node [style=none] (13) at (0.5, 1.75) {};
		\node [style=none] (14) at (0.5, -1.75) {};
		\node [style=none] (15) at (-1, 1.75) {};
		\node [style=none] (16) at (-1, -0.75) {};
		\node [style=none] (17) at (-1, -1.75) {};
		\node [style=none] (18) at (-1.5, 1.75) {};
		\node [style=none] (19) at (-1.5, -1.75) {};
		\node [style=none] (20) at (-1, 0.75) {};
		\node [style=none] (21) at (0, -2.75) {};
		\node [style=none] (22) at (0, -3.75) {};
		\node [style=right label] (23) at (0, -3.375) {$P^m$};
		\node [style=none] (24) at (0, 2.75) {};
		\node [style=none] (25) at (0, 3.75) {};
		\node [style=right label] (26) at (0, 3.375) {$F^n$};
		\node [style=right label] (29) at (-1, 1) {$A^l_\mr{out}$};
		\node [style=right label] (30) at (-1, -1.25) {$A^k_\mr{in}$};
		\node [style=none] (31) at (-2.5, 0) {$\lambda^l_k$};
		\node [style=none] (33) at (2.5, 0) {$\mu_m^n$};
	\end{pgfonlayer}
	\begin{pgfonlayer}{edgelayer}
		\draw (5.center) to (6.center);
		\draw (7.center) to (0.center);
		\draw (18.center) to (13.center);
		\draw (13.center) to (14.center);
		\draw (14.center) to (19.center);
		\draw (15.center) to (20.center);
		\draw (16.center) to (17.center);
		\draw (18.center) to (7.center);
		\draw (19.center) to (6.center);
		\draw (21.center) to (22.center);
		\draw (24.center) to (25.center);
		\draw (0.center) to (5.center);
	\end{pgfonlayer}
\end{tikzpicture}} \quad : \quad %
\begin{tikzpicture}
	\begin{pgfonlayer}{nodelayer}
		\node [style=right label] (1) at (-1, -1.25) {$X_\mr{in}$};
		\node [style=right label] (3) at (1, -1.25) {$A^k_\mr{in}$};
		\node [style=large map] (4) at (0, 0) {$\ca$};
		\node [style=none] (5) at (-1, 1.5) {};
		\node [style=right label] (6) at (-1, 1.25) {$X_\mr{out}$};
		\node [style=none] (7) at (1, 1.5) {};
		\node [style=right label] (8) at (1, 1.25) {$A^l_\mr{out}$};
		\node [style=none] (9) at (-1, -1.5) {};
		\node [style=none] (10) at (1, -1.5) {};
		\node [style=none] (11) at (-2.25, 0) {$\lambda^l_k$};
	\end{pgfonlayer}
	\begin{pgfonlayer}{edgelayer}
		\draw (9.center) to (5.center);
		\draw (10.center) to (7.center);
	\end{pgfonlayer}
\end{tikzpicture}}  \quad \mapsto \quad %
\begin{tikzpicture}
	\begin{pgfonlayer}{nodelayer}
		\node [style=right label] (1) at (-1, -1.25) {$X_\mr{in}$};
		\node [style=semilarge map] (4) at (0, 0) {$(\ci \otimes \cs)[\ca]$};
		\node [style=none] (5) at (-1, 1.5) {};
		\node [style=right label] (6) at (-1, 1.25) {$X_\mr{out}$};
		\node [style=none] (9) at (-1, -1.5) {};
		\node [style=right label] (11) at (1, -1.25) {$P^m$};
		\node [style=none] (12) at (1, 1.5) {};
		\node [style=right label] (13) at (1, 1.25) {$F^n$};
		\node [style=none] (14) at (1, -1.5) {};
		\node [style=none] (15) at (-3, 0) {$\mu_m^n$};
	\end{pgfonlayer}
	\begin{pgfonlayer}{edgelayer}
		\draw (9.center) to (5.center);
		\draw (14.center) to (12.center);
	\end{pgfonlayer}
\end{tikzpicture}}  := \,\,\, %
\begin{tikzpicture}
	\begin{pgfonlayer}{nodelayer}
		\node [style=none] (0) at (1.75, 2.75) {};
		\node [style=none] (5) at (1.75, -2.75) {};
		\node [style=none] (6) at (-1.25, -2.75) {};
		\node [style=none] (7) at (-1.25, 2.75) {};
		\node [style=none] (12) at (1.25, 0) {$\cs$};
		\node [style=none] (13) at (0.75, 1.75) {};
		\node [style=none] (14) at (0.75, -1.75) {};
		\node [style=none] (15) at (-0.5, 1.75) {};
		\node [style=none] (17) at (-0.5, -1.75) {};
		\node [style=none] (18) at (-1.25, 1.75) {};
		\node [style=none] (19) at (-1.25, -1.75) {};
		\node [style=none] (21) at (0.5, -2.75) {};
		\node [style=none] (22) at (0.5, -3.75) {};
		\node [style=right label] (23) at (0.5, -3.375) {$P^m$};
		\node [style=none] (24) at (0.5, 2.75) {};
		\node [style=none] (25) at (0.5, 3.75) {};
		\node [style=right label] (26) at (0.5, 3.375) {$F^n$};
		\node [style=none] (33) at (-4, 0) {$\mu_m^n$};
		\node [style=right label] (44) at (-2.5, -3.25) {$X_\mr{in}$};
		\node [style=large map] (47) at (-1.5, 0) {$\ca$};
		\node [style=none] (48) at (-2.5, 3.75) {};
		\node [style=right label] (49) at (-2.5, 3.5) {$X_\mr{out}$};
		\node [style=none] (52) at (-2.5, -3.75) {};
	\end{pgfonlayer}
	\begin{pgfonlayer}{edgelayer}
		\draw (5.center) to (6.center);
		\draw (7.center) to (0.center);
		\draw (18.center) to (13.center);
		\draw (13.center) to (14.center);
		\draw (14.center) to (19.center);
		\draw (18.center) to (7.center);
		\draw (19.center) to (6.center);
		\draw (21.center) to (22.center);
		\draw (24.center) to (25.center);
		\draw (52.center) to (48.center);
		\draw (15.center) to (17.center);
		\draw (0.center) to (5.center);
	\end{pgfonlayer}
\end{tikzpicture}}$}
    \fcaption{Diagrammatic representation of a supermap $\cs$ on routed channels, and of its action on a routed channel $(\lambda, \ca)$ (also acting on an auxiliary system), yielding a routed channel $(\mu, (\ci \otimes \cs)[\ca])$.}
    \label{fig:RoutedSupermap}
\end{figure*}

We can first apply this notion to the definition of supermaps acting on a single routed channel. First, we briefly recall the basic notions introduced in Ref.\,\cite{vanrietvelde2020routed}. Here, we will restrict ourselves to routes with \textit{full coherence}, i.e., only encoding sectorial constraints and not coherence constraints\footnote{This leads us to adopting notations that are somewhat different from those of Ref.\,\cite{vanrietvelde2020routed}. There, routes for general channels were taken to be completely positive relations $\Lambda_{kk'}^{ll'}$. As fully coherent routes are those $\Lambda$'s which can be written as $\Lambda_{kk'}^{ll'} = \lambda_k^l \lambda_{k'}^{l'}$, we simplify our notations in the present article by just referring to them as $\lambda_k^l$.}. A partitioned Hilbert space $X^k$ is a Hilbert space with a preferred orthogonal partition, labelled by a finite set $\cz_X$; i.e., $\ch_X := \oplus_{k \in \cz_X} \ch_X^k$. Given two such partitioned spaces $A_\mr{in}^k$ and $A_\mr{out}^l$, and a relation $\lambda: \cz_{A_\mr{in}} \to \cz_{A_\mr{out}}$ (or, in other terms, a Boolean matrix $(\lambda_k^l)_{k \in \cz_{A_\mr{in}}}^{l \in \cz_{A_\mr{out}}}$), we say that a channel $\ca \in \mr{QChan}(A_\mr{in} \to A_\mr{out})$ \textit{follows the route} $\lambda$ if

\be 
\forall k, \, \forall \rho \in \cl(\ch_{A_\mr{in}}^k), \,\, \ca(\rho) \in \cl \left( \bigoplus \,_{l | \lambda_k^l=1} \ch_{A_\mr{out}}^l \right) \,.
\ee 

Equivalently (see Theorem 6 in Ref.\,\cite{vanrietvelde2020routed}), given any Kraus representation $(K_i)_i$ of $\ca$, $\ca$ follows $\lambda$ if and only if

\be \forall i, \forall k, \, \forall |\psi\rangle \in \ch_{A_\mr{in}}^k, \,\, K_i \,|\psi\rangle \in \bigoplus \,_{l | \lambda_k^l=1} \ch_{A_\mr{out}}^l \,. \ee

We denote the set of channels of type $A_\mr{in} \to A_\mr{out}$ that follow the route $\lambda$ as $\mr{QChan}^\lambda(A_\mr{in}^k \to A_\mr{out}^l) \subseteq \mr{QChan}(A_\mr{in} \to A_\mr{out})$. We will also say that these channels have type $A_\mr{in}^k \overset{\lambda}{\to} A_\mr{out}^l$; it is this type of channels on which we want to define supermaps. It is easy to prove that the condition defining the extension of $A_\mr{in}^k \overset{\lambda}{\to} A_\mr{out}^l$ to auxiliary systems can be simplified.

\medskip 

\begin{lemma} \label{lem:RoutedExtensions}
For a type $A_\mr{in}^k \overset{\lambda}{\to} A_\mr{out}^l$ and auxiliary systems $X_\mr{in}$, $X_\mr{out}$, one has:

\be \begin{split}
    &\mr{Ext}_{X_\mr{in} \to X_\mr{out}}(\mr{QChan}^\lambda(A_\mr{in}^k \to A_\mr{out}^l)) \\ &= \mr{QChan}^\lambda (A_\mr{in}^k \otimes X_\mr{in} \to A_\mr{out}^l \otimes X_\mr{out}) \, .
\end{split}\ee
\end{lemma}
\medskip 

In other terms, the extension of the set of channels $A_\mr{in} \to A_\mr{out}$ following a route $\lambda$ to a type $X_\mr{in} \to X_\mr{out}$ is simply the set of channels $A_\mr{in} \otimes X_\mr{in} \to A_\mr{out} \otimes X_\mr{out}$ following $\lambda$. The definition of supermaps on routed channels then derives naturally from Definition \ref{def:SupermapRestricted}.
\medskip 

\begin{definition}[Supermap on routed channels] \label{def:SupermapRoutedMono}
Let $A_\mr{in}^k$, $A_\mr{out}^l$, $P^m$ and $F^n$ be partitioned Hilbert spaces, and let $\lambda: \cz_{A_\mr{in}} \to \cz_{A_\mr{out}}$ and $\mu: \cz_P \to \cz_F$ be two relations. A deterministic supermap of type $(A_\mr{in}^k \overset{\lambda}{\to} A_\mr{out}^l) \to (P^m \overset{\mu}{\to} F^n)$ is a linear map $\cs$ from $\mr{Herm}(A_\mr{in} \to A_\mr{out})$ to $\mr{Herm}(P \to F)$ such that, for any auxiliary systems $X_\mr{in}$, $X_\mr{out}$ and for any channel $\cc \in \mr{QChan}^\lambda (A_\mr{in}^k \otimes X_\mr{in} \to A_\mr{out}^l \otimes X_\mr{out})$, one has

\be (\cs \otimes \ci_{X_\mr{in} \to X_\mr{out}}) [\cc] \in \mr{QChan}^\mu (P^m \otimes X_\mr{in} \to F^n \otimes X_\mr{out}) \, . \ee
\end{definition}
\medskip 

\begin{figure*}
    \centering
    \scalebox{.7}{$%
\begin{tikzpicture}
	\begin{pgfonlayer}{nodelayer}
		\node [style=none] (0) at (2.5, 2.75) {};
		\node [style=none] (1) at (2.5, 1.75) {};
		\node [style=none] (2) at (0.5, 1.75) {};
		\node [style=none] (3) at (0.5, -1.75) {};
		\node [style=none] (4) at (2.5, -1.75) {};
		\node [style=none] (5) at (2.5, -2.75) {};
		\node [style=none] (6) at (-2.5, -2.75) {};
		\node [style=none] (7) at (-2.5, 2.75) {};
		\node [style=none] (8) at (2, 1.75) {};
		\node [style=none] (9) at (2, 0.75) {};
		\node [style=none] (10) at (2, -0.75) {};
		\node [style=none] (11) at (2, -1.75) {};
		\node [style=none] (12) at (0, 0) {$\cs$};
		\node [style=none] (13) at (-0.5, 1.75) {};
		\node [style=none] (14) at (-0.5, -1.75) {};
		\node [style=none] (15) at (-2, 1.75) {};
		\node [style=none] (16) at (-2, -0.75) {};
		\node [style=none] (17) at (-2, -1.75) {};
		\node [style=none] (18) at (-2.5, 1.75) {};
		\node [style=none] (19) at (-2.5, -1.75) {};
		\node [style=none] (20) at (-2, 0.75) {};
		\node [style=none] (21) at (0, -2.75) {};
		\node [style=none] (22) at (0, -3.75) {};
		\node [style=right label] (23) at (0, -3.375) {$P^q$};
		\node [style=none] (24) at (0, 2.75) {};
		\node [style=none] (25) at (0, 3.75) {};
		\node [style=right label] (26) at (0, 3.375) {$F^r$};
		\node [style=right label] (27) at (2, 1) {$B^n_\mr{out}$};
		\node [style=right label] (28) at (2, -1.25) {$B^m_\mr{in}$};
		\node [style=right label] (29) at (-2, 1) {$A^l_\mr{out}$};
		\node [style=right label] (30) at (-2, -1.25) {$A^k_\mr{in}$};
		\node [style=none] (31) at (-3.5, 0) {$\lambda^l_k$};
		\node [style=none] (32) at (2.75, 0) {$\sigma_m^n$};
		\node [style=none] (33) at (-3.5, 2.5) {$\mu_q^r$};
	\end{pgfonlayer}
	\begin{pgfonlayer}{edgelayer}
		\draw (0.center) to (1.center);
		\draw (1.center) to (2.center);
		\draw (2.center) to (3.center);
		\draw (3.center) to (4.center);
		\draw (4.center) to (5.center);
		\draw (5.center) to (6.center);
		\draw (7.center) to (0.center);
		\draw (8.center) to (9.center);
		\draw (10.center) to (11.center);
		\draw (18.center) to (13.center);
		\draw (13.center) to (14.center);
		\draw (14.center) to (19.center);
		\draw (15.center) to (20.center);
		\draw (16.center) to (17.center);
		\draw (18.center) to (7.center);
		\draw (19.center) to (6.center);
		\draw (21.center) to (22.center);
		\draw (24.center) to (25.center);
	\end{pgfonlayer}
\end{tikzpicture}} \,\, : \, %
\begin{tikzpicture}
	\begin{pgfonlayer}{nodelayer}
		\node [style=right label] (1) at (-1, -1.25) {$X_\mr{in}$};
		\node [style=right label] (3) at (1, -1.25) {$A^k_\mr{in}$};
		\node [style=large map] (4) at (0, 0) {$\ca$};
		\node [style=none] (5) at (-1, 1.5) {};
		\node [style=right label] (6) at (-1, 1.25) {$X_\mr{out}$};
		\node [style=none] (7) at (1, 1.5) {};
		\node [style=right label] (8) at (1, 1.25) {$A^l_\mr{out}$};
		\node [style=none] (9) at (-1, -1.5) {};
		\node [style=none] (10) at (1, -1.5) {};
		\node [style=none] (11) at (-2.25, 0) {$\lambda^l_k$};
	\end{pgfonlayer}
	\begin{pgfonlayer}{edgelayer}
		\draw (9.center) to (5.center);
		\draw (10.center) to (7.center);
	\end{pgfonlayer}
\end{tikzpicture}} \otimes \,\, %
\begin{tikzpicture}
	\begin{pgfonlayer}{nodelayer}
		\node [style=right label] (1) at (-1, -1.25) {$B^m_\mr{in}$};
		\node [style=right label] (3) at (1, -1.25) {$Y_\mr{in}$};
		\node [style=large map] (4) at (0, 0) {$\cb$};
		\node [style=none] (5) at (-1, 1.5) {};
		\node [style=right label] (6) at (-1, 1.25) {$B^n_\mr{out}$};
		\node [style=none] (7) at (1, 1.5) {};
		\node [style=right label] (8) at (1, 1.25) {$Y_\mr{out}$};
		\node [style=none] (9) at (-1, -1.5) {};
		\node [style=none] (10) at (1, -1.5) {};
		\node [style=none] (11) at (-2.5, 0) {$\sigma_m^n$};
	\end{pgfonlayer}
	\begin{pgfonlayer}{edgelayer}
		\draw (9.center) to (5.center);
		\draw (10.center) to (7.center);
	\end{pgfonlayer}
\end{tikzpicture}} \mapsto \,\, %
\begin{tikzpicture}
	\begin{pgfonlayer}{nodelayer}
		\node [style=right label] (1) at (-1.5, -1.25) {$X_\mr{in}$};
		\node [style=right label] (3) at (1.5, -1.25) {$Y_\mr{in}$};
		\node [style=very large map] (4) at (0, 0) {$(\ci \otimes \cs \otimes \ci)[\ca,\cb]$};
		\node [style=none] (5) at (-1.5, 1.5) {};
		\node [style=right label] (6) at (-1.5, 1.25) {$X_\mr{out}$};
		\node [style=none] (7) at (1.5, 1.5) {};
		\node [style=right label] (8) at (1.5, 1.25) {$Y_\mr{out}$};
		\node [style=none] (9) at (-1.5, -1.5) {};
		\node [style=none] (10) at (1.5, -1.5) {};
		\node [style=right label] (11) at (0, -1.25) {$P^q$};
		\node [style=none] (12) at (0, 1.5) {};
		\node [style=right label] (13) at (0, 1.25) {$F^r$};
		\node [style=none] (14) at (0, -1.5) {};
		\node [style=none] (15) at (-4, 0) {$\mu_q^r$};
	\end{pgfonlayer}
	\begin{pgfonlayer}{edgelayer}
		\draw (9.center) to (5.center);
		\draw (10.center) to (7.center);
		\draw (14.center) to (12.center);
	\end{pgfonlayer}
\end{tikzpicture}} \, := \,\,\,\, %
\begin{tikzpicture}
	\begin{pgfonlayer}{nodelayer}
		\node [style=none] (0) at (2.5, 2.75) {};
		\node [style=none] (1) at (2.5, 1.75) {};
		\node [style=none] (2) at (0.5, 1.75) {};
		\node [style=none] (3) at (0.5, -1.75) {};
		\node [style=none] (4) at (2.5, -1.75) {};
		\node [style=none] (5) at (2.5, -2.75) {};
		\node [style=none] (6) at (-2.5, -2.75) {};
		\node [style=none] (7) at (-2.5, 2.75) {};
		\node [style=none] (8) at (1.75, 1.75) {};
		\node [style=none] (9) at (1.75, -1.75) {};
		\node [style=none] (12) at (0, 0) {$\cs$};
		\node [style=none] (13) at (-0.5, 1.75) {};
		\node [style=none] (14) at (-0.5, -1.75) {};
		\node [style=none] (15) at (-1.75, 1.75) {};
		\node [style=none] (17) at (-1.75, -1.75) {};
		\node [style=none] (18) at (-2.5, 1.75) {};
		\node [style=none] (19) at (-2.5, -1.75) {};
		\node [style=none] (21) at (0, -2.75) {};
		\node [style=none] (22) at (0, -3.75) {};
		\node [style=right label] (23) at (0, -3.375) {$P^q$};
		\node [style=none] (24) at (0, 2.75) {};
		\node [style=none] (25) at (0, 3.75) {};
		\node [style=right label] (26) at (0, 3.375) {$F^r$};
		\node [style=none] (33) at (-5.25, 0) {$\mu_q^r$};
		\node [style=right label] (35) at (3.75, -3.25) {$Y_\mr{in}$};
		\node [style=large map] (36) at (2.75, 0) {$\cb$};
		\node [style=none] (39) at (3.75, 3.75) {};
		\node [style=right label] (40) at (3.75, 3.5) {$Y_\mr{out}$};
		\node [style=none] (42) at (3.75, -3.75) {};
		\node [style=right label] (44) at (-3.75, -3.25) {$X_\mr{in}$};
		\node [style=large map] (47) at (-2.75, 0) {$\ca$};
		\node [style=none] (48) at (-3.75, 3.75) {};
		\node [style=right label] (49) at (-3.75, 3.5) {$X_\mr{out}$};
		\node [style=none] (52) at (-3.75, -3.75) {};
	\end{pgfonlayer}
	\begin{pgfonlayer}{edgelayer}
		\draw (0.center) to (1.center);
		\draw (1.center) to (2.center);
		\draw (2.center) to (3.center);
		\draw (3.center) to (4.center);
		\draw (4.center) to (5.center);
		\draw (5.center) to (6.center);
		\draw (7.center) to (0.center);
		\draw (8.center) to (9.center);
		\draw (18.center) to (13.center);
		\draw (13.center) to (14.center);
		\draw (14.center) to (19.center);
		\draw (18.center) to (7.center);
		\draw (19.center) to (6.center);
		\draw (21.center) to (22.center);
		\draw (24.center) to (25.center);
		\draw (42.center) to (39.center);
		\draw (52.center) to (48.center);
		\draw (15.center) to (17.center);
	\end{pgfonlayer}
\end{tikzpicture}}$}
    \fcaption{Diagrammatic representation of a supermap $\cs$ acting on a pair of  routed channels $(\lambda, \ca)$ and $(\sigma, \cb)$ (also acting on auxiliary systems), yielding a routed channel $(\mu, (\ci \otimes \cs \otimes \ci)[\ca,\cb])$.}
    \label{fig:2RoutedSupermap}
\end{figure*}

We show how supermaps on routed channels can be represented graphically in Figure \ref{fig:RoutedSupermap}. The ${\tt CTRL}$ supermap described in Theorem \ref{th:ControlSupermap} can be characterised as a supermap on routed channels, with type $(A^k \overset{\delta}{\to} A^l) \to ( C \otimes S^1 \to C \otimes S^1)$.

Let us now turn to supermaps acting  on multiple routed channels. To avoid clutter, we will present the construction for supermaps acting on a pair of channels, the generalisation to $N\ge 2$ being immediate. Formally, these have to be defined as supermaps whose input channels should be product channels, with each channel in this product following a given route. For some partitioned spaces $A_\mr{in}^k$, $A_\mr{out}^l$, $B_\mr{in}^m$ and $B_\mr{out}^n$, and for two relations $\lambda: \cz_{A_\mr{in}} \to \cz_{A_\mr{out}}$ and $\sigma: \cz_{B_\mr{in}} \to \cz_{B_\mr{out}}$, we thus define $\mr{ProdChan}^{\lambda \times \sigma}(A_\mr{in}^k \otimes B_\mr{in}^m \to A_\mr{out}^l \otimes B_\mr{out}^n)$ to be the intersection of the set of product channels $\mr{ProdChan}(A_\mr{in} \otimes B_\mr{in} \to A_\mr{out} \otimes B_\mr{out})$ with $\mr{QChan}^{\lambda \times \sigma}(A_\mr{in}^k \otimes B_\mr{in}^m \to A_\mr{out}^l \otimes B_\mr{out}^n)$. One can then define supermaps  acting on such a set, once again following Definition \ref{def:SupermapRestricted}.

\begin{figure*}
    \centering
    $%
\begin{tikzpicture}
	\begin{pgfonlayer}{nodelayer}
		\node [style=none] (0) at (3.25, 3) {};
		\node [style=none] (1) at (3.25, 1.75) {};
		\node [style=none] (2) at (1.25, 1.75) {};
		\node [style=none] (3) at (1.25, -1.75) {};
		\node [style=none] (4) at (3.25, -1.75) {};
		\node [style=none] (5) at (3.25, -3) {};
		\node [style=none] (6) at (-3.25, -3) {};
		\node [style=none] (7) at (-3.25, 3) {};
		\node [style=none] (8) at (2.75, 1.75) {};
		\node [style=none] (9) at (2.75, 0.75) {};
		\node [style=none] (10) at (2.75, -0.75) {};
		\node [style=none] (11) at (2.75, -1.75) {};
		\node [style=none] (13) at (-1.25, 1.75) {};
		\node [style=none] (14) at (-1.25, -1.75) {};
		\node [style=none] (15) at (-2.75, 1.75) {};
		\node [style=none] (16) at (-2.75, -0.75) {};
		\node [style=none] (17) at (-2.75, -1.75) {};
		\node [style=none] (18) at (-3.25, 1.75) {};
		\node [style=none] (19) at (-3.25, -1.75) {};
		\node [style=none] (20) at (-2.75, 0.75) {};
		\node [style=none] (21) at (-1, -3) {};
		\node [style=none] (22) at (-1, -4) {};
		\node [style=right label] (23) at (-1, -3.625) {$C$};
		\node [style=none] (24) at (1, 3) {};
		\node [style=none] (25) at (1, 4) {};
		\node [style=right label] (26) at (1, 3.625) {$S^1_\mr{out}$};
		\node [style=right label] (27) at (2.75, 1) {$S^n_\mr{out}$};
		\node [style=right label] (28) at (2.75, -1.25) {$S^m_\mr{in}$};
		\node [style=right label] (29) at (-2.75, 1) {$S^l_\mr{out}$};
		\node [style=right label] (30) at (-2.75, -1.25) {$S^k_\mr{in}$};
		\node [style=none] (31) at (-4.25, 0) {$\delta^l_k$};
		\node [style=none] (32) at (3.5, 0) {$\delta_m^n$};
		\node [style=none] (36) at (0, 0) {$\tt{2 \mhyphen CTRL}$};
		\node [style=none] (39) at (-1, 3) {};
		\node [style=none] (40) at (-1, 4) {};
		\node [style=right label] (41) at (-1, 3.625) {$C$};
		\node [style=none] (42) at (1, -3) {};
		\node [style=none] (43) at (1, -4) {};
		\node [style=right label] (44) at (1, -3.625) {$S^1_\mr{in}$};
	\end{pgfonlayer}
	\begin{pgfonlayer}{edgelayer}
		\draw (0.center) to (1.center);
		\draw (1.center) to (2.center);
		\draw (2.center) to (3.center);
		\draw (3.center) to (4.center);
		\draw (4.center) to (5.center);
		\draw (5.center) to (6.center);
		\draw (7.center) to (0.center);
		\draw (8.center) to (9.center);
		\draw (10.center) to (11.center);
		\draw (18.center) to (13.center);
		\draw (13.center) to (14.center);
		\draw (14.center) to (19.center);
		\draw (15.center) to (20.center);
		\draw (16.center) to (17.center);
		\draw (18.center) to (7.center);
		\draw (19.center) to (6.center);
		\draw (21.center) to (22.center);
		\draw (24.center) to (25.center);
		\draw (39.center) to (40.center);
		\draw (42.center) to (43.center);
	\end{pgfonlayer}
\end{tikzpicture}} \quad\quad = \quad %
\begin{tikzpicture}
	\begin{pgfonlayer}{nodelayer}
		\node [style=none] (9) at (1, 0.75) {};
		\node [style=none] (10) at (1, -0.75) {};
		\node [style=none] (16) at (3, -0.75) {};
		\node [style=none] (20) at (3, 0.75) {};
		\node [style=none] (21) at (-1, 5.5) {};
		\node [style=none] (22) at (-1, -5.5) {};
		\node [style=right label] (23) at (-1, -5.125) {$C$};
		\node [style=right label] (26) at (1, 5.125) {$S^1_\mr{out}$};
		\node [style=right label] (27) at (1, 1) {$S^l_\mr{out}$};
		\node [style=right label] (28) at (1, -1.25) {$S^k_\mr{in}$};
		\node [style=right label] (29) at (3, 1) {$S^n_\mr{out}$};
		\node [style=right label] (30) at (3, -1.25) {$S^m_\mr{in}$};
		\node [style=none] (31) at (3.5, 0) {$\delta^n_m$};
		\node [style=none] (32) at (0, 0) {$\delta_k^l$};
		\node [style=right label] (41) at (-1, 5.125) {$C$};
		\node [style=none] (42) at (1, 5.5) {};
		\node [style=none] (43) at (1, -5.5) {};
		\node [style=right label] (44) at (1, -5.125) {$S^1_\mr{in}$};
		\node [style=white dot] (45) at (-1, -2.5) {};
		\node [style=white dot] (46) at (-1, 2.5) {};
		\node [style=large map] (47) at (2, -2.5) {$\mr{SWAP}$};
		\node [style=large map] (48) at (2, 2.5) {$\mr{SWAP}$};
		\node [style=none] (55) at (3, 4.25) {};
		\node [style=none] (60) at (3, -4.25) {};
		\node [style=left label] (61) at (-1.25, 0) {$C^p$};
		\node [style=left label] (62) at (-1.25, -2.5) {$\lambda^{pkm}$};
		\node [style=left label] (63) at (-1.25, 2.5) {$\lambda_{pln}$};
		\node [style=right label] (64) at (3, -3.75) {$S^0_\mr{in}$};
		\node [style=right label] (65) at (3, 3.75) {$S^0_\mr{out}$};
		\node [style=none] (67) at (3.075, -4.6) {$s^0$};
		\node [style=none] (68) at (3, -5.25) {};
		\node [style=none] (69) at (2.25, -4.25) {};
		\node [style=none] (70) at (3.75, -4.25) {};
		\node [style=none] (72) at (2.925, 4.6) {$s^0$};
		\node [style=none] (73) at (3, 5.25) {};
		\node [style=none] (74) at (2.25, 4.25) {};
		\node [style=none] (75) at (3.75, 4.25) {};
	\end{pgfonlayer}
	\begin{pgfonlayer}{edgelayer}
		\draw (9.center) to (42.center);
		\draw (55.center) to (20.center);
		\draw (10.center) to (43.center);
		\draw [in=90, out=-90] (16.center) to (60.center);
		\draw (21.center) to (22.center);
		\draw (45) to (47);
		\draw (46) to (48);
		\draw (68.center) to (70.center);
		\draw (70.center) to (69.center);
		\draw (69.center) to (68.center);
		\draw (73.center) to (75.center);
		\draw (75.center) to (74.center);
		\draw (74.center) to (73.center);
	\end{pgfonlayer}
\end{tikzpicture}}$
    \fcaption{Fully explicit formulation of the ${\tt 2 \mhyphen CTRL}$ supermap as a supermap on sector-preserving channels, in the framework of routed quantum circuits \cite{vanrietvelde2020routed}. $s^0$ is the only state on the one-dimensional sector $\cl(\ch_{S \mr{in}}^0)$. The Boolean vector $(\lambda^{pkm})_{p,k,m \in \{0,1\}}$ has coefficients $1$ for indices $001$ and $110$, and $0$ elsewhere. $(\lambda_{pln})_{p,l,n \in \{0,1\}}$ is its transpose.
    An advantage of the routed formulation is to allow one to get rid of \correction{ the superfluous   embedding operations that were} present in the standard formulation (Figure \ref{fig:2CoherentControl}).}
     \label{fig:2CoherentControlRouted}
\end{figure*}

\medskip 

\begin{definition}[Supermaps on pairs of routed channels] \label{def:SupermapRoutedBi}
Let $A_\mr{in}^k$, $A_\mr{out}^l$, $B_\mr{in}^m$, $B_\mr{out}^n$, $P^q$ and $F^r$ be partitioned Hilbert spaces, and let $\lambda: \cz_{A_\mr{in}} \to \cz_{A_\mr{out}}$, $\sigma: \cz_{B_\mr{in}} \to \cz_{B_\mr{out}}$ and $\mu: \cz_P \to \cz_F$ be relations. A deterministic supermap of type $(A_\mr{in}^k \overset{\lambda}{\to} A_\mr{out}^l) \otimes (B_\mr{in}^m \overset{\sigma}{\to} B_\mr{out}^n) \to (P^q \overset{\mu}{\to} F^r)$ is a linear map $\cs$ from $\mr{Herm}(A_\mr{in} \otimes B_\mr{in} \to A_\mr{out} \otimes B_\mr{out})$ to $\mr{Herm}(P \to F)$ such that, for any auxiliary systems $X_\mr{in}$, $X_\mr{out}$, $Y_\mr{in}$, $Y_\mr{out}$ and for any pair of channels $\ca \in \mr{QChan}^\lambda (A_\mr{in}^k \otimes X_\mr{in} \to A_\mr{out}^l \otimes X_\mr{out})$, $\cb \in \mr{QChan}^\sigma (B_\mr{in}^m \otimes Y_\mr{in} \to B_\mr{out}^n \otimes Y_\mr{out})$, one has

\be \begin{split}
     &(\ci_{X_\mr{in} \to X_\mr{out}} \otimes \cs \otimes \ci_{Y_\mr{in} \to Y_\mr{out}}) [\ca \otimes \cb] \\ &\in \mr{QChan}^\mu (X_\mr{in} \otimes P^q \otimes Y_\mr{in} \to X_\mr{out} \otimes F^r \otimes Y_\mr{out}) \, . 
\end{split}\ee
\end{definition}
\medskip

We show how  supermaps on pairs of routed channels can be represented graphically in Figure \ref{fig:2RoutedSupermap}. The ${\tt 2 \mhyphen CTRL}$ supermap described in Theorem \ref{th:2ControlSupermap} can be characterised as a supermap on routed channels, with type $(A^k \overset{\delta}{\to} A^l) \otimes (A^k \overset{\delta}{\to} A^l) \to (C \otimes S^1 \to C \otimes S^1)$; we show in Figure \ref{fig:2CoherentControlRouted} how it can then be written in a fully explicit way in the language of routed quantum circuits. Figure \ref{fig:2CoherentControlRouted} can thus be seen as a more compact rewriting of Figure \ref{fig:2CoherentControl}\correction{, which contained the additional operations $\cv$  and $\cd$.  The role of these operations was simply to embed  the target systems into suitable sectors.  While in some specific realisations these embeddings may correspond to  non-trivial physical operations, from the information-theoretic point of view they are irrelevant, and they can be completely absorbed into the graphical language of routed circuits.}

\section{Conclusion}

In this work, we showed that sector-preserving channels of type $(1,d)$ are the necessary resource for implementing controlled channels on a $d$-dimensional system. We  demonstrated that sector-preserving channels and  controlled channels  are into one-to-one correspondence, and can be faithfully parametrised by channels with a pinned Kraus operator.  In addition, we showed that this mathematical one-to-one correspondence can be implemented physically: for any given $d$, there exist two universal circuit architectures  that convert sector-preserving channels into controlled channels, and vice-versa. 

In addition to characterising the resources for the standard notion of control, we defined a generalised type of controlled channels, called  compositely-controlled,  in which several states of the control are associated with the `do-nothing' option.  Also in this case, we showed that the controlled channels are in (both mathematical and physical) correspondence with sector-preserving channels, in this case of type $(d,1,\dots,1)$. We also generalised these results to the implementation of coherent control between $N$ channels, and showed that, when these channels are not isometries, such an implementation requires the use of sector-preserving channels of type $(1,d)$ which reduce, on their $d$-dimensional sectors, to isometric extensions  of the channels to be controlled.

The framework of sector-preserving channels provides an information-theoretic underpinning to the existing experimental schemes for the implementation of universal coherent control \cite{zhou2011AddingControl, zhou2013CalculatingEigenvalues, Friis2014, Dunjko_2015, Friis_2015}\correction{, as well as a pathway to the generalisation of such schemes to more complex architectures. Furthermore, it lays down the conceptual and mathematical framework required to analyse and compare the performance of implementations of coherent control as well as the advantages that they yield, e.g. in computation or in communication. As a byproduct, it also motivates new experiments on the realisation of composite control and theoretical investigations of its uses.}



Finally, in order to properly characterise the supermaps we defined, we introduced the notion of supermaps on routed channels, and gave them a rigorous mathematical definition, building on the framework of Ref.\,\cite{vanrietvelde2020routed}. However, the supermaps presented in the present work can always be extended to act on general channels (even though this will make them lose their unitary-preserving property). An interesting open question  is whether there exist  supermaps on routed channels which cannot be extended to act on general channels. \correction{This might in particular prove relevant to the study of Indefinite Causal Order: it has been shown \cite{barrett2020cyclic} that some indefinite causal structures could be investigated using index-matching circuits, a specific type of routed circuits.}

 \section{Acknowledgments}  

It is a pleasure to thank Hlér Kristjánsson and Matt Wilson for helpful discussions, advice and comments. AV is supported by the EPSRC Centre for Doctoral Training in Controlled Quantum Dynamics. This work is supported by the Hong Research Grant Council  through  grant 17300918 and through the Senior Research Fellowship ``Quantum Causal Discovery and the Foundations of
Quantum Artificial Intelligence,'' by the Croucher Foundation,   by  the HKU Seed Funding for Basic Research, and by the John Templeton Foundation through grant  61466, The Quantum Information Structure of Spacetime  (qiss.fr).  Research at the Perimeter Institute is supported by the Government of Canada through the Department of Innovation, Science and Economic Development Canada and by the Province of Ontario through the Ministry of Research, Innovation and Science. The opinions expressed in this publication are those of the authors and do not necessarily reflect the views of the John Templeton Foundation.

\section*{References}
\bibliographystyle{utphys}
\bibliography{references}

\correction{
\appendix{\quad A review of the terminology on coherent control in previous literature}\label{app:terminology}  

The notion of `coherent control' has been studied under several different names in the  literature, which might lead to some confusion. In this appendix, we  provide a review of the different terms previously used, arguing that they all essentially refer to the same notion. We will then motivate the choice of the term `coherent control' employed in this paper.

Coherent control was first considered for unitary gates in the work of Aharonov and coauthors \cite{aharonov1990superpositions}. In this work, controlled unitary gates were used to build what  was called a `superposition of time evolutions'. More precisely, the authors discussed the possibility of implementing evolutions of the form $\sum_j c_j  \, U_j$, where the $U_j$'s are unitary operators, and the $c_j$'s are complex coefficients. It was proven that such an evolution could be realised, for arbitrary $c_j$'s, using auxiliary systems and postselection. The protocol described in Ref.\,\cite{aharonov1990superpositions} consists in realising the controlled unitary gate $\sum_j  \,  |j\rangle\langle j|  \otimes U_j$,  initialising the control system in a superposition state, measuring the control system in a suitable basis, and then postselecting on a specific measurement outcome. 

Another early instance was in the work of Åberg \cite{aaberg2004subspace, aaberg2004operations}, in which some of what would later come to be seen as the crucial features of coherent control were pointed out and analysed under different names. Indeed, Ref.\,\cite{aaberg2004subspace} introduces the concept of so-called subspace-preserving channels, asking how their mathematical form can be obtained from that of their restrictions to each subspace, a procedure called \textit{gluing of completely positive maps}, which is noted to be non-unique. This procedure is a mathematical avatar of the task of coherent control; and, even though the question of physical implementation is not discussed in detail, the comment on the non-uniqueness can be regarded as an  early observation of the ill-definedness of the control between two quantum channels. This ill-definedness is noted to be due to the incompleteness of the description of the channels one wants to glue.  An application of these methods to single-particle interferometry is described in Ref.\,\cite{aaberg2004operations}.

Around the same time, Oi \cite{oi2003interference} studied the \textit{interference of CP maps}, proposing  that the combination of quantum channels in an interferometic setup could reveal additional properties of their physical implementation that are not included in the mathematical expression of quantum channels.   In the light of our results, the ability to probe  additional properties of the implementation is due to the fact that the channels  inserted in the interferometric  setup are not the original channels,  but rather sector-preserving channels of type $(1,d)$ which coincide with them on their $d$-dimensional sector. It is the properties of these sector-preserving channels, not of the original ones, that become visible through interferometry. 

Finally, Chiribella and Kristjánsson \cite{Chiribella_2019} considered {\em superpositions of quantum channels}, in the context of a communication model where the information carriers move on a  superpositions of trajectories.  Even though this paper focused on applications to communication, its framework also yields an implementation of the task of coherent control, as shown by the present paper. In this perspective, superpositions of trajectories represent one of the possible physical implementations of coherent control. 

The term we adopted here, `coherent control' (or sometimes `quantum control', or simply `control'), is  commonly found in both  experimental \cite{zhou2011AddingControl, zhou2013CalculatingEigenvalues, Friis2014, Dunjko_2015, Friis_2015} and theoretical \cite{Araujo2014, Chiribella_2016, gavorova2020topological, dong2020controlled, Abbott_2020, clement2020, branciard2021} works. Consistency with this relatively large body of works is one of the benefits of choosing the term `control'. Moreover, this choice  has the advantage of referring to a clearly defined operational task, rather than to analogies with properties of quantum states (such as `superpositions of quantum evolutions' or `superpositions of quantum channels'), to mathematical procedures (`gluing of CP maps'), to possible phenomena  (`interference of CP maps'), or to specific types of physical implementations (`superpositions of trajectories').

}

\appendix{\quad Parametrising the coherent control between two channels}\label{app:2ControlParam}  

In this Appendix, we prove Theorem \ref{th:2ControlParam}. We fix a Kraus representation $(A_i)_{i=1}^n$ of minimal length of $\ca$. We first prove that any version of a controlled channel between $\ca$ and $\cb$ admits a Kraus decomposition $(K_j)_{j=1}^{m}$, where $m \geq n$, $K_j = {\tt ctrl} \mhyphen   (A_j,B_j)$ for $j \leq n$ and $K_j = {\tt ctrl} \mhyphen   (0,B_j)$ for $j > n$. Let us take such a channel, given by Kraus operators $({\tt ctrl} \mhyphen   (A_i',B_i')_{i=1}^{m}$. The $A_i'$ form a Kraus representation of $\ca$; therefore, $m \geq n$ and there exists an unitary matrix $(V_{ji})_{i,j=1}^m$ such that $\sum_i V_{ji} A_i' = A_j$ for $j \leq n$ and $0$ for $j > n$. Then, $(\sum_j V_{ji} K_i)_{j=1}^m$ is a Kraus representation of the right form for the controlled channel.

We now prove that, given two choices $(B_i)_{i=1}^m$ and $(B_j')_{j=1}^{m'}$ of Kraus representations for $\cb$, the controlled channels that they define are equal if and only if $\forall i \leq n, B_i = B_i'$. First, suppose that the latter equation holds. Then, taking an isometry matrix $(V_{ji})_{n < i \leq m}^{n < j \leq m'}$ relating the Kraus decompositions $(B_i)_{i=n+1}^m$ and $(B_j)_{j=n+1}^{m'}$, we can complete it into a unitary matrix $(V_{ji})_{1 \leq i \leq m}^{1 \leq j \leq m'}$ by taking $\forall i,j \leq n, V_{ji} = \delta_{ji}$; one then has $\forall i,j, \sum_i V_{ji} {\tt ctrl} \mhyphen   (A_i,B_i) = {\tt ctrl} \mhyphen   (A_j,B_j')$. Reciprocally, suppose that the controlled channels defined by the choices $(B_i)_{i=1}^m$ and $(B_j')_{j=1}^{m'}$ are equal. Taking then $(V_{ji})_{1 \leq i \leq m}^{1 \leq j \leq m'}$ to be an isometry matrix relating the associated Kraus decompositions, one has in particular $\forall i,j \leq n, \sum_i V_{ji} A_i = A_j$. Yet, that $(A_i)_{i=1}^n$ is a Kraus representation of minimal length implies in particular that the $A_i$'s are linearly independent; therefore $\forall i,j \leq n, V_{ji} = \delta_{ji}$, which implies $\forall i \leq n, B_i' = B_i$.

\appendix{\quad Control of two noisy channels} \label{app:2ControlGeneral}

In this Appendix, we propose a universal circuit implementation for all possible versions of the control between two noisy channels $\ca$ and $\cb$ from $\cl(\ch_{T_\mr{in}})$ to $\cl(\ch_{T_\mr{out}})$.  \correction{   To avoid clutter, we will take the isomorphisms $T_{\rm in}  \simeq S^1_{\rm in}$ and $T_{\rm out}  \simeq S^1_{\rm out}$ to be strict, that is, as will assume  $T_{\rm in}  = S^1_{\rm in}$ and $T_{\rm out}  = S^1_{\rm out}$.  }

Recall that, as proven in Section \ref{sec:CoherentControl2}, in the case where $\ca$ and $\cb$ are isometric channels the controlled  version could be implemented using as resources sector-preserving channels from $\cl(\ch^0_{S_\mr{in}} \oplus \ch^1_{S_\mr{in}})$ to $\cl(\ch^0_{S_\mr{out}} \oplus \ch^1_{S_\mr{out}})$, where $\ch^1_{S_\mr{out}} := \ch_{S_\mr{out}}$, $\ch^1_{S_\mr{in}} := \ch_{S_\mr{in}}$, and $\ch^0_{S_\mr{in}} \cong \ch^0_{S_\mr{out}} \cong \mathbb{C}$, with these channels restricting respectively to $\ca$ and $\cb$ on $\cl(\ch^1_{S_\mr{in}})$. However, the controlled channels yielded by this method can feature full coherence only between at most one Kraus operator of $\ca$ and one operator of $\cb$.

Here, we shall therefore make use of more complex resources. These resources will be sector-preserving channels whose multi-dimensional output sector will not be $\ch^1_{S_\mr{out}}$, but $\ch^1_{S_\mr{out}} \otimes \ch^1_E$, where $\ch^1_E$ is an auxiliary Hilbert space. The restrictions of these channels to this sector will have to yield $\ca$ and $\cb$ when $E^1$ is traced out. In other words, to get the full scope of controls between $\ca$ and $\cb$ we need to use sector-preserving channels that restrict to (possibly partial) purifications of $\ca$ and $\cb$ on their multi-dimensional sectors. Using such resources, the number of Kraus operators of $\ca$ and $\cb$ between which there can be full coherence in the controlled channel is capped by the dimension of $\ch_E^1$. In particular a sufficiently large $\ch^1_E$ will ensure that all possible controlled channels can be generated.

More formally, we define the supermap ${\tt 2 \mhyphen CTRL(E)}$ from the supermap ${\tt 2 \mhyphen CTRL}$ in the following way:\footnote{Here, we defined this supermap as a routed one (also using the convention of contracting Kronecker deltas) for clarity, but this could also be arbitrarily expanded into a supermap acting on all channels from $\cl(\ch_{S_\mr{in}})$ to $\cl(\oplus_{k\in \{0,1\}} \ch_{S_\mr{out}}^k \otimes \ch_E^k)$. Note that when writing such a non-routed supermap, one would have to write the combination of $S_\mr{out}$ and $E$ as a single wire, as the way in which they combine to form $S_\mr{out}^k E^k$ is not a tensor product and cannot be expressed using standard quantum circuits.}

\be \label{eq:2ControlESupermap} %
\begin{tikzpicture}
	\begin{pgfonlayer}{nodelayer}
		\node [style=none] (0) at (4.75, 3) {};
		\node [style=none] (1) at (4.75, 1.75) {};
		\node [style=none] (2) at (1.75, 1.75) {};
		\node [style=none] (3) at (1.75, -1.75) {};
		\node [style=none] (4) at (4.75, -1.75) {};
		\node [style=none] (5) at (4.75, -3) {};
		\node [style=none] (6) at (-4.75, -3) {};
		\node [style=none] (7) at (-4.75, 3) {};
		\node [style=none] (8) at (2.5, 1.75) {};
		\node [style=none] (9) at (2.5, 0.75) {};
		\node [style=none] (10) at (3.25, -0.75) {};
		\node [style=none] (11) at (3.25, -1.75) {};
		\node [style=none] (13) at (-1.75, 1.75) {};
		\node [style=none] (14) at (-1.75, -1.75) {};
		\node [style=none] (15) at (-4.25, 1.75) {};
		\node [style=none] (16) at (-3.25, -0.75) {};
		\node [style=none] (17) at (-3.25, -1.75) {};
		\node [style=none] (18) at (-4.75, 1.75) {};
		\node [style=none] (19) at (-4.75, -1.75) {};
		\node [style=none] (20) at (-4.25, 0.75) {};
		\node [style=none] (21) at (-1, -3) {};
		\node [style=none] (22) at (-1, -4) {};
		\node [style=right label] (23) at (-1, -3.625) {$C$};
		\node [style=right label] (27) at (2.5, 1) {$S^m_\mr{out}$};
		\node [style=right label] (28) at (3.25, -1.25) {$S^m_\mr{in}$};
		\node [style=right label] (29) at (-4.25, 1) {$S^k_\mr{out}$};
		\node [style=right label] (30) at (-3.25, -1.25) {$S^k_\mr{in}$};
		\node [style=none] (36) at (0, 0) {$2 \mhyphen \tt{CTRL} \mr{(E)}$};
		\node [style=none] (39) at (-1, 3) {};
		\node [style=none] (40) at (-1, 5.5) {};
		\node [style=right label] (41) at (-1, 4.625) {$C$};
		\node [style=none] (42) at (1, -3) {};
		\node [style=none] (43) at (1, -4) {};
		\node [style=right label] (44) at (1, -3.625) {$S^1_\mr{in}$};
		\node [style=none] (45) at (4, 1.75) {};
		\node [style=none] (46) at (4, 0.75) {};
		\node [style=right label] (47) at (4, 1) {$E^m$};
		\node [style=none] (48) at (-2.75, 1.75) {};
		\node [style=none] (49) at (-2.75, 0.75) {};
		\node [style=right label] (50) at (-2.75, 1) {$E^k$};
		\node [style=none] (54) at (1, 3) {};
		\node [style=none] (55) at (1, 5.5) {};
		\node [style=right label] (56) at (1, 4.625) {$S^1_\mr{out}$};
	\end{pgfonlayer}
	\begin{pgfonlayer}{edgelayer}
		\draw (0.center) to (1.center);
		\draw (1.center) to (2.center);
		\draw (2.center) to (3.center);
		\draw (3.center) to (4.center);
		\draw (4.center) to (5.center);
		\draw (5.center) to (6.center);
		\draw (7.center) to (0.center);
		\draw (8.center) to (9.center);
		\draw (10.center) to (11.center);
		\draw (18.center) to (13.center);
		\draw (13.center) to (14.center);
		\draw (14.center) to (19.center);
		\draw (15.center) to (20.center);
		\draw (16.center) to (17.center);
		\draw (18.center) to (7.center);
		\draw (19.center) to (6.center);
		\draw (21.center) to (22.center);
		\draw (39.center) to (40.center);
		\draw (42.center) to (43.center);
		\draw (45.center) to (46.center);
		\draw (48.center) to (49.center);
		\draw (54.center) to (55.center);
	\end{pgfonlayer}
\end{tikzpicture}} \quad\quad := \quad %
\begin{tikzpicture}
	\begin{pgfonlayer}{nodelayer}
		\node [style=none] (0) at (4.25, 3) {};
		\node [style=none] (1) at (4.25, 1.75) {};
		\node [style=none] (2) at (1.25, 1.75) {};
		\node [style=none] (3) at (1.25, -1.75) {};
		\node [style=none] (4) at (4.25, -1.75) {};
		\node [style=none] (5) at (4.25, -3) {};
		\node [style=none] (6) at (-4.25, -3) {};
		\node [style=none] (7) at (-4.25, 3) {};
		\node [style=none] (8) at (2, 1.75) {};
		\node [style=none] (9) at (2, 0.75) {};
		\node [style=none] (10) at (2.75, -0.75) {};
		\node [style=none] (11) at (2.75, -1.75) {};
		\node [style=none] (13) at (-1.25, 1.75) {};
		\node [style=none] (14) at (-1.25, -1.75) {};
		\node [style=none] (15) at (-3.75, 1.75) {};
		\node [style=none] (16) at (-2.75, -0.75) {};
		\node [style=none] (17) at (-2.75, -1.75) {};
		\node [style=none] (18) at (-4.25, 1.75) {};
		\node [style=none] (19) at (-4.25, -1.75) {};
		\node [style=none] (20) at (-3.75, 0.75) {};
		\node [style=none] (21) at (-1, -3) {};
		\node [style=none] (22) at (-1, -4) {};
		\node [style=right label] (23) at (-1, -3.625) {$C$};
		\node [style=right label] (27) at (2, 1) {$S^m_\mr{out}$};
		\node [style=right label] (28) at (2.75, -1.25) {$S^m_\mr{in}$};
		\node [style=right label] (29) at (-3.75, 1) {$S^k_\mr{out}$};
		\node [style=right label] (30) at (-2.75, -1.25) {$S^k_\mr{in}$};
		\node [style=none] (38) at (0, 0) {$\mr{2} \mhyphen \tt{CTRL}$};
		\node [style=none] (39) at (-2, 3) {};
		\node [style=none] (40) at (-2, 5.5) {};
		\node [style=right label] (41) at (-2, 4.625) {$C$};
		\node [style=none] (42) at (1, -3) {};
		\node [style=none] (43) at (1, -4) {};
		\node [style=right label] (44) at (1, -3.625) {$S^1_\mr{in}$};
		\node [style=none] (45) at (3.5, 1.75) {};
		\node [style=none] (46) at (3.5, 0.75) {};
		\node [style=right label] (47) at (3.5, 1) {$E^m$};
		\node [style=none] (48) at (-2.25, 1.75) {};
		\node [style=none] (49) at (-2.25, 0.75) {};
		\node [style=right label] (50) at (-2.25, 1) {$E^k$};
		\node [style=none] (51) at (2, 3) {};
		\node [style=upground] (52) at (2, 4.75) {};
		\node [style=right label] (53) at (2, 3.625) {$E^1$};
		\node [style=none] (54) at (0, 3) {};
		\node [style=none] (55) at (0, 5.5) {};
		\node [style=right label] (56) at (0, 4.625) {$S^1_\mr{out}$};
	\end{pgfonlayer}
	\begin{pgfonlayer}{edgelayer}
		\draw (0.center) to (1.center);
		\draw (1.center) to (2.center);
		\draw (2.center) to (3.center);
		\draw (3.center) to (4.center);
		\draw (4.center) to (5.center);
		\draw (5.center) to (6.center);
		\draw (7.center) to (0.center);
		\draw (8.center) to (9.center);
		\draw (10.center) to (11.center);
		\draw (18.center) to (13.center);
		\draw (13.center) to (14.center);
		\draw (14.center) to (19.center);
		\draw (15.center) to (20.center);
		\draw (16.center) to (17.center);
		\draw (18.center) to (7.center);
		\draw (19.center) to (6.center);
		\draw (21.center) to (22.center);
		\draw (39.center) to (40.center);
		\draw (42.center) to (43.center);
		\draw (45.center) to (46.center);
		\draw (48.center) to (49.center);
		\draw (51.center) to (52);
		\draw (54.center) to (55.center);
	\end{pgfonlayer}
\end{tikzpicture}} \, . \ee

Let us now prove that, for a given choice of $E$, ${\tt 2 \mhyphen CTRL(E)}$ can produce all controlled channels in which the number of coherent pairs of Kraus operators is less than the dimension of $E$.

\begin{theorem}
We fix an environment $E$ with dimension $D$, and use the one-to-one parametrisation of the control between two channels provided by Theorem \ref{th:2ControlParam}: i.e., given a Kraus representation $(A_i)_{i=1}^n$ of $\ca$ of minimal length, the parametrisation is given by the choice of $n$ Kraus operators $B_i$ of $\cb$.

Then any choice of a control in which only the $D$ first operators $B_i$ are non-zero can be obtained from the use of the ${\tt 2 \mhyphen CTRL(E)}$ supermap.
\end{theorem}

\begin{proof}
In the case $D=1$ (i.e. that of the ${\tt 2 \mhyphen CTRL}$ supermap), it can easily be computed, from the formula of Fig.\,\ref{fig:2CoherentControl}, that any controlled version in which there is coherence between $A_1$ and $B_1$ can be obtained by plugging the channels $\widetilde{\ca}^{A_1}$ and $\widetilde{\cb}^{B_1}$ in ${\tt 2 \mhyphen CTRL}$.

Considering now the case $D > 1$, let us take a version $\cc$ of a control between $\ca$ and $\cb$ for which a Kraus representation is $\Big( \ketbra{0}{0}_C \otimes A_1 + \ketbra{1}{1}_C \otimes B_1, \dots, \ketbra{0}{0}_C \otimes A_D + \ketbra{1}{1}_C \otimes B_D, \ketbra{0}{0}_C \otimes A_{D+1}, \dots, \ketbra{0}{0}_C \otimes A_{n}, \ketbra{1}{1}_C \otimes B_{D+1}, \dots, \ketbra{1}{1}_C \otimes B_{m} \Big)$. Then a (possibly partial) purification of $\cc$ is given by the channel of type $C S^1_\mr{in} \to C S^1_\mr{out} E$ for which a Kraus representation is $\Big( \sum_{i=1}^D \big( \ketbra{0}{0}_C \otimes A_i  + \ketbra{1}{1}_C \otimes B_i \big) \otimes \ket{i}_E, \ketbra{0}{0}_C \otimes A_{D+1} \otimes \ket{1}_E, \dots, \ketbra{0}{0}_C \otimes A_{n} \otimes \ket{1}_E, \ketbra{1}{1}_C \otimes B_{D+1} \otimes \ket{1}_E, \dots, \ketbra{1}{1}_C \otimes B_{m} \otimes \ket{1}_E \Big)$. This latter channel can be seen as being a version of a control between two channels $S_\mr{in}^1 \to S_\mr{out}^1 E^1$ with coherence between one pair of Kraus operators. By the first part of the proof, it can thus be obtained by applying the ${\tt 2 \mhyphen CTRL}$ supermap to suitable sector-preserving channels of type $S_\mr{in}^k \to S_\mr{out}^k E^k$. Discarding $E^1$ then yields $\cc$. The ${\tt 2 \mhyphen CTRL(E)}$ as defined in (\ref{eq:2ControlESupermap}) thus yields $\cc$ when applied to the same sector-preserving channels.
\end{proof}

In particular, as any channel $S_\mr{in}^1 \to S_\mr{out}^1$ admits a Kraus representation of length less than the product of the dimensions of $S_\mr{in}^1$ and $S_\mr{out}^1$, all versions of controlled channels can be obtained from the use of the supermap ${\tt 2 \mhyphen CTRL(E)}$ when $E$ is of that dimension.

\end{document}